\documentclass[aps,prd,reprint,superscriptaddress,longbibliography,showpacs,nofootinbib]{revtex4-2}

\usepackage{graphicx}
\usepackage{svg}
\usepackage{dcolumn}
\usepackage{bm}

\usepackage[figurename=Figure]{caption}
\usepackage{ragged2e}
\DeclareCaptionJustification{plain}{\justifying}
\usepackage[percent]{overpic}
\usepackage{float}    
\usepackage{verbatim} 
\usepackage{amsmath}  
\usepackage{bbold}
\usepackage{upgreek} 
\usepackage{amssymb}  
\usepackage{svg}
\usepackage{enumitem} 
\usepackage{latexsym,epsfig}
\usepackage{subcaption}
\usepackage{stackrel}
\usepackage{mathtools}
\usepackage{amsthm}
\newtheorem{proposition}{Proposition}

\newtheorem{corollary}{Corollary}
\newtheorem{theorem}{Theorem}

\usepackage{placeins} 

\usepackage{xr}
\usepackage[colorlinks=true,breaklinks=true,allcolors=blue]{hyperref}
\usepackage[capitalize]{cleveref}
\usepackage{lipsum}
\usepackage{dsfont}
\usepackage{coffeestains}
\usepackage{physics}

\usepackage[scr=boondox]{mathalfa}
 
\usepackage[scr=boondox]{mathalfa}

\def\Tr{\ensuremath{{\operatorname{Tr}}}}
\def\Re{\ensuremath{{\operatorname{Re}}}}

\newcommand\rj[1]{ {\color{magenta} #1} } 

\newcommand{\be}{\begin{equation}}
	\newcommand{\ee}{\end{equation}}
\newcommand{\bea}{\begin{equation}\begin{aligned}}
		\newcommand{\eea}{\end{aligned}\end{equation}}
\newcommand{\ben}{\begin{enumerate}}
	\newcommand{\een}{\end{enumerate}}

\DeclareDocumentCommand{\nint}{ O{} O{} m }{\ensuremath{ \int_{\mbox{\scriptsize $#1$}}^{\mbox{\scriptsize$#2$}}\!\!\! \mbox{\small $\,\mathrm{d}#3$\! }}}

\usepackage{tikz,bm} 
\usepackage{circuitikz}
\usetikzlibrary{decorations.markings}
\usetikzlibrary{decorations.pathreplacing,calligraphy}

\definecolor{mycolor}{rgb}{1,0.2,0.3}
\definecolor{brightgreen}{rgb}{0.4, 1.0, 0.0}
\definecolor{britishracinggreen}{rgb}{0.0, 0.26, 0.15}
\definecolor{cadmiumgreen}{rgb}{0.0, 0.42, 0.24}
\definecolor{ceruleanblue}{rgb}{0.16, 0.32, 0.75}
\definecolor{darkelectricblue}{rgb}{0.33, 0.41, 0.47}
\definecolor{darkpowderblue}{rgb}{0.0, 0.2, 0.6}
\definecolor{darktangerine}{rgb}{1.0, 0.66, 0.07}
\definecolor{emerald}{rgb}{0.31, 0.78, 0.47}
\definecolor{palatinatepurple}{rgb}{0.41, 0.16, 0.38}
\definecolor{pastelviolet}{rgb}{0.8, 0.6, 0.79}

\begin{document}

	\title{Many-Body Protection of Topological Edge Memory in Strong Interacting Quenches}
	\author{Yuxiao Hang}
	\email{yhang@usc.edu}
	\affiliation{%
		Department of Physics and Astronomy, University of Southern California, Los Angeles, CA 90089-0484, USA
	}
	\author{Stephan Haas}
	\email{shaas@usc.edu}
	\affiliation{%
		Department of Physics and Astronomy, University of Southern California, Los Angeles, CA 90089-0484, USA
	}
	\author{Rishabh Jha}
	\email{rishabh.jha@usc.edu}
	\affiliation{%
		Department of Physics and Astronomy, University of Southern California, Los Angeles, CA 90089-0484, USA
	}
	
	%
	

\begin{abstract}
Quantum quenches drive edge states far from equilibrium, yet whether the memory of a topological initial state survives in a non-integrable, interacting system has remained largely unexplored. We study this question in the bond-alternating XXZ chain---an interacting Su--Schrieffer--Heeger model hosting symmetry-protected topological edge modes with markedly enhanced boundary magnetization---and analyze quenches across all combinations of single-particle and many-body initial and final Hamiltonians. The results organize by a single distinction as we rigorously establish in this work: whether the post-quench Hamiltonian is free or genuinely interacting. For a free post-quench Hamiltonian, the dynamics is solved exactly by a correlation-matrix approach; the boundary-mode return amplitude decays as $t^{-3/2}$, and initial interactions enter only through a dressed one-body density matrix. For a genuinely interacting post-quench Hamiltonian, finite-time stability bounds prove that away from local resonances the first-dimer magnetization remains stable on time windows growing as arbitrarily large powers of the inverse inter-dimer coupling. Matrix product state simulations across all four protocols show that interactions in the final Hamiltonian markedly extend finite-time boundary memory---with local suppression near the isotropic $SU(2)$ point---revealing a many-body protection mechanism in a non-integrable system where scrambling would otherwise wash out initial-state memory fast.
\end{abstract}

	\maketitle
	
	
	\section{Introduction}
	\label{sec:intro}
	
	Quantum quenching---preparing a system in the ground state of an initial Hamiltonian $H_i$ and then suddenly switching to a new Hamiltonian $H_f$---has become a central paradigm for probing non-equilibrium dynamics in quantum many-body physics~\cite{PhysRevLett.96.136801,Calabrese_2007,Eisert2015,Polkovnikov2011}. Because the subsequent evolution is unitary and the total energy is conserved, an isolated many-body system is generically expected to thermalize, with local observables relaxing to values described by the eigenstate thermalization hypothesis~\cite{Deutsch1991,Srednicki1994,Rigol2008}. Exact results in conformal field theory and integrable models have characterized this process through entanglement entropy and correlation functions~\cite{Calabrese_2005,Pasquale_Calabrese_2004,Calabrese_2011,Calabrese_2016,PhysRevA.89.013609,10.21468/SciPostPhysLectNotes.20,PhysRevB.109.224308}: in integrable systems, quasi particle propagation produces an initial linear growth of entanglement entropy that eventually saturates to a volume-law plateau described by a generalized Gibbs ensemble~\cite{Calabrese_2005,10.21468/SciPostPhys.4.3.017,Vidmar_2016}, while in non-integrable or chaotic systems entanglement grows without bound~\cite{Mezei2017,PhysRevLett.111.127205,PhysRevLett.94.097203}, reflecting efficient scrambling of quantum information throughout the system, a picture that has been extensively corroborated by numerical studies of strongly interacting chains~\cite{Jha2025Jun,10.21468/SciPostPhys.4.3.017,10.21468/SciPostPhys.5.4.033}.
	
	A central question is whether and how the memory of the initial state can survive long-time unitary evolution in a non-integrable system. Many-body localization (MBL), driven by strong disorder, is the best-understood mechanism for preventing thermalization: even within the localized phase, entanglement grows logarithmically in time due to dephasing through emergent local integrals of motion~\cite{BASKO20061126,PhysRevB.75.155111,Nandkishore2015,PhysRevB.90.064201}. Beyond disorder-driven localization, quantum many body scars~\cite{Turner2018,Ho2019,Huang2025Aug} and Hilbert space fragmentation~\cite{Sala2020} have shown that non-thermal behavior and long-lived initial-state memory can emerge in disorder-free systems~\cite{PhysRevX.10.021051}, and recent work on periodically driven chains has revealed non-equilibrium entanglement transitions entirely invisible to local observables~\cite{Gadge2026Mar}. These findings have stimulated broad interest in identifying non-ergodic phenomena in clean interacting systems, where the interplay between topology and many-body interactions may offer a qualitatively distinct route to memory protection.
	
	Symmetry-protected topological phases host edge modes localized at the system boundaries that are robust against local perturbations respecting the protecting symmetry~\cite{RevModPhys.82.3045,ShunQingshen2012,RevModPhys.83.1057}. Early studies on the non-interacting Kitaev chain and Su--Schrieffer--Heeger (SSH) model established that quenching from a topologically non-trivial initial ground state to parameters whose ground state lies on the trivial side of the equilibrium phase diagram leads to collapse and revival of edge-state weight~\cite{PhysRevE.93.062117,GhoshMartin2023,Rossi2022Jan}, while quenches in two-dimensional topological systems---including the Bernevig-Hughes-Zhang (BHZ) model and the toric code---have revealed collapse and revival of edge currents and topological R\'{e}nyi entropy dynamics~\cite{Patel2013,PhysRevA.80.060302,PhysRevLett.110.170605}. Particularly relevant to the present work, Lee and Song established in a non-interacting setting that quenched topological boundary modes can persist unexpectedly long even after the bulk is driven to parameters whose equilibrium ground state is topologically trivial~\cite{Lee2021Jun}. However, these investigations were almost entirely confined to integrable or non-interacting settings, leaving the fate of topological edge modes in non-integrable, genuinely interacting quantum many-body systems essentially unexplored.
	
	The bond-alternating XXZ spin chain is a canonical one-dimensional model hosting a symmetry-protected topological phase with robust edge magnetization~\cite{PhysRevB.110.165145,vgnz-wcfq,Tzeng2016,Qiang_2013}. Via the Jordan--Wigner transformation it maps onto an interacting spinless-fermion model with bond-modulated hopping and nearest-neighbor density--density interactions; in the non-interacting limit $J_z=0$ it reduces exactly to the SSH free-fermion chain~\cite{PhysRevLett.42.1698,ShunQingshen2012,PhysRevB.101.235155}, whose structural motif is realized in condensed matter, photonic, and programmable quantum-simulation platforms~\cite{english2026topologicallyprotectedspatiallylocalized,sompet2022nature,deleseleuc2019science,Katz2025Oct,splitthoff2024prresearch}. For nonzero $J_z$ the phase diagram is considerably richer: classical antiferromagnetic and ferromagnetic phases emerge at large $|J_z|$~\cite{PhysRevB.110.165145,vgnz-wcfq}, the Jordan--Wigner image becomes genuinely interacting and non-integrable, and the edge magnetization in the topological phase is markedly enhanced relative to all bulk sites~\cite{vgnz-wcfq}, providing a sharp and experimentally accessible probe of topological order. Whether this edge magnetization survives a quantum quench in the presence of many-body interactions---a question that directly probes the competition between topological protection and interaction-driven thermalization---has so far remained open.

In this work we establish that post-quench boundary memory is not only survivable but can be \emph{enhanced} by many-body interactions in the final Hamiltonian, through exact analytical results, rigorous finite-time stability bounds, and large-scale time-evolving block-decimation (TEBD) simulations. The entire analysis organizes by a single distinction: whether the post-quench Hamiltonian is free or genuinely interacting. A self-contained summary of the principal results is given in Sec.~\ref{sec:results}, before the technical developments.

The remainder of this paper is organized as follows. Section~\ref{sec:model} introduces the model, its fermionic formulation, and the quench protocol, and provides a self-contained summary of the principal results. Section~\ref{sec:single} presents exact analytical results for the non-interacting limit $J_z=0$, based on the correlation-matrix approach and benchmarked against exact diagonalization, establishing the $t^{-3/2}$ relaxation scale that serves as the reference for the interacting problem. Section~\ref{sec:manybody} develops the interacting analysis, deriving the finite-time stability bounds, identifying the leakage mechanism, proving the asymptotic first-dimer-memory statement, and combining these analytical results with TEBD simulations across all four quench protocols. Section~\ref{sec:conclusion} summarizes the main results and discusses possible extensions. The Supplemental Material (SM) serves as a companion document containing detailed proofs and additional benchmarking material underlying the analytical and numerical results presented here.
	
	\section{Model and Quench Protocol}
	\label{sec:model}
	
	We study the bond-alternating XXZ spin chain with Hamiltonian
	\begin{equation}
		\begin{aligned}
			H =\;& \sum_{n\in A}\!\Bigl(J_1^{xy}\bigl(\sigma_n^x\sigma_{n+1}^x
			+\sigma_n^y\sigma_{n+1}^y\bigr)+J_1^z\sigma_n^z\sigma_{n+1}^z\Bigr) \\
			&+\sum_{n\in B}\!\Bigl(J_2^{xy}\bigl(\sigma_n^x\sigma_{n+1}^x
			+\sigma_n^y\sigma_{n+1}^y\bigr)+J_2^z\sigma_n^z\sigma_{n+1}^z\Bigr),
			\label{eq:Ham}
		\end{aligned}
	\end{equation}
	where $A$-bonds connect intracell pairs $(2m-1,2m)$ and $B$-bonds connect intercell pairs $(2m,2m+1)$, $m=1,\ldots,N/2$. We parametrize the couplings as
	\begin{equation}
		\begin{aligned}
			J_1^{xy}= t(1+\delta),&\quad J_2^{xy}=t(1-\delta),\\
			J_1^z   = J_z(1+\delta),&\quad J_2^z   =J_z(1-\delta),
			\label{eq:params}
		\end{aligned}
	\end{equation}
with $t=1$ setting the energy scale, $\delta\in(-1,1)$ the bond-dimerization, and $J_z$ the interaction strength. This parametrization keeps the intracell-to intercell ratio identical in both the $XY$ and $Z$ sectors, so $\delta$ alone drives the topological transition while $J_z$ independently tunes the interaction. At $J_z=1$ every bond carries full $SU(2)$ spin rotation symmetry. A schematic of the model and its equilibrium phase diagram are shown in Fig.~\ref{fig:model_and_phase}.
	
\begin{figure}[t]
	\centering
	\begin{subfigure}[b]{\columnwidth}
		\centering
		\includegraphics[width=0.9\linewidth]{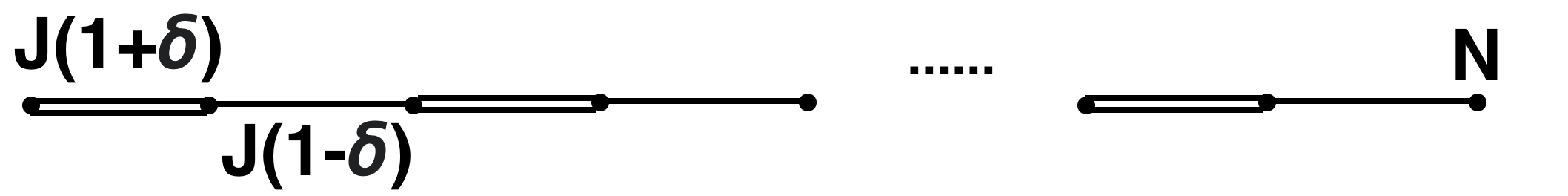}
		\caption{Schematic of the bond-alternating XXZ chain with $N$ sites and open boundary conditions. Double (single) bonds denote the  intracell (intercell) couplings $J_1^{xy,z}$ ($J_2^{xy,z}$) controlled by the dimerization $\delta$.}
		\label{fig:model}
	\end{subfigure}
	\par\vspace{6pt}
	\begin{subfigure}[b]{\columnwidth}
		\centering
		\includegraphics[width=\columnwidth]{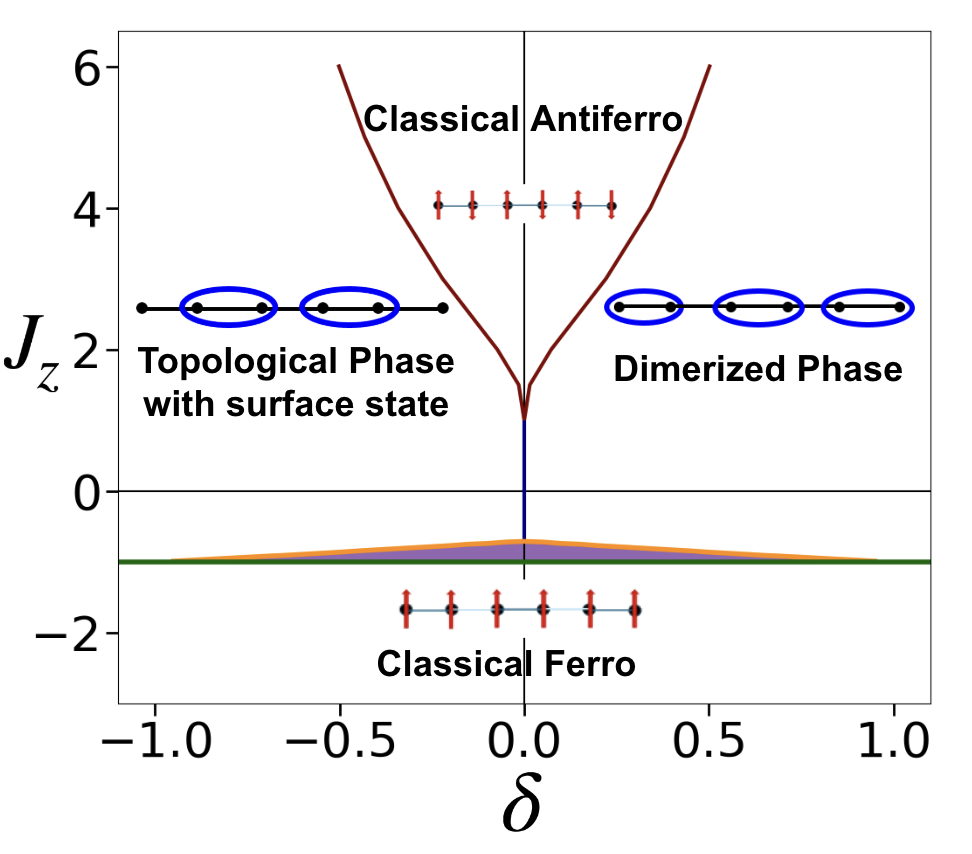}
		\caption{Equilibrium ground-state phase diagram of the bond-alternating XXZ chain~\cite{vgnz-wcfq} in the $(\delta,J_z)$ plane. For $\delta<0$ and moderate $J_z$ the chain is in the topologically non-trivial phase~\cite{PhysRevLett.42.1698,ShunQingshen2012}, which hosts a pair of zero-energy edge modes with markedly enhanced boundary magnetization; the gap closes at $\delta=0$; $\delta>0$ is topologically trivial. At $J_z=0$ the model reduces to the Su--Schrieffer--Heeger (SSH) chain. Increasing $J_z>0$ at small $|\delta|$ drives the system into a classical antiferromagnetic (AFM) phase; large $J_z<0$ drives it into a ferromagnetic (FM) phase~\cite{vgnz-wcfq,Tzeng2016,Qiang_2013}. All initial states studied here are ground states prepared deep in the topological phase ($\delta_i=-0.95$, varying $J_z^i$); the system is then quenched to a final Hamiltonian at $\delta_f=+0.95$, varying $J_z^f$. Since the post-quench evolution is unitary and far from equilibrium, the final state does not correspond to any ground state in this diagram.}
		\label{fig:phasediagram}
	\end{subfigure}
	\caption{Model and equilibrium phase diagram. (a) Chain geometry. (b) Ground-state phase diagram; parameters used in this work are indicated.}
	\label{fig:model_and_phase}
\end{figure}
	
	In the non-interacting limit $J_z=0$ the Jordan- Wigner transformation gives
	\begin{equation}
		J_\ell^{xy}\!\left(\sigma_j^x\sigma_{j+1}^x
		+\sigma_j^y\sigma_{j+1}^y\right)
		=2J_\ell^{xy}\!\left(c_j^\dagger c_{j+1}+c_{j+1}^\dagger c_j\right)
		\quad(\ell=1,2),
		\label{eq:jw_xy_factor}
	\end{equation}
	so the SSH hopping amplitudes are $J_1=2J_1^{xy}$ and $J_2=2J_2^{xy}$. For nonzero $J_z$ the Jordan--Wigner image acquires nearest-neighbor density--density interactions $\propto J_j^z n_j n_{j+1}$, making the system genuinely interacting and non-integrable~\cite{PhysRevB.110.165145,vgnz-wcfq}. The resulting equilibrium phase diagram is shown in Fig.~\ref{fig:phasediagram} and discussed there; within the topological phase, $\langle S_1^z\rangle$ (and equivalently $\langle S_N^z\rangle$) is markedly enhanced relative to all bulk sites~\cite{vgnz-wcfq}, making it the natural observable for tracking topological boundary memory far from equilibrium.

	\textit{Quench Protocol.---}
	We work with open boundary conditions and fix $S_{\rm tot}^z=0$ throughout, which corresponds to half-filling under the Jordan--Wigner transformation and ensures consistent comparisons across all values of $J_z^i$ and $J_z^f$. The dimerization is fixed at $|\delta|=0.95$ to place both Hamiltonians deep in their respective equilibrium phases and suppress finite-size corrections. The protocol prepares the ground state $ \Psi_0\rangle$ of $H_i=H(\delta_i,J_z^i)$ and evolves it under $H_f=H(\delta_f,J_z^f)$:
	\begin{equation}
		H_i|\Psi_0\rangle=E_0|\Psi_0\rangle,\qquad
		|\Psi(t)\rangle=e^{-iH_f t}|\Psi_0\rangle.
	\end{equation}
	We always prepare the initial state in the topological phase ($\delta_i=-0.95$) and evolve with a final Hamiltonian at $\delta_f=+0.95$, varying $J_z^i$ and $J_z^f$ independently. Here $\delta_i<0$ and $\delta_f>0$ refer to the equilibrium phases of $H_i$ and $H_f$; since the post-quench evolution is unitary and far from equilibrium, the time-evolved state carries no topological classification in that sense. Quenches in the opposite direction carry no initial edge magnetization and do not address topological memory; they are not considered here.
	
	We establish in this work that the four quench types fall into two qualitative classes determined by whether the post-quench Hamiltonian is free or genuinely interacting:
	\begin{itemize}
		\item \textbf{Free post-quench dynamics} (SP$\to$SP and MB$\to$SP): time evolution is governed by a quadratic Hamiltonian and is analytically controlled within the correlation-matrix framework, with initial interactions entering only through a dressed one-body density matrix.
		\item \textbf{Genuinely interacting post-quench dynamics} (SP$\to$MB and MB$\to$MB): the final Hamiltonian is many-body; analytic control is provided by the strongly dimerized limit, and TEBD simulations show that interactions in the final Hamiltonian can markedly enhance finite-time boundary memory relative to the free reference.
	\end{itemize}

	\subsection*{Summary of Results}
	\label{sec:results}
	
	Figure~\ref{fig:quench-schematic} shows the four quench protocols schematically; the two qualitative classes identified above and established in this work organize the entire analysis.
	
	\begin{figure}[t]
		\centering
		\begin{tikzpicture}[>=stealth,every node/.style={font=\small}]
			\node at (0,2.2)  {Initial};
			\node at (4,2.2)  {Final};
			\node[draw,rounded corners,minimum width=1.7cm,minimum height=0.9cm,
			fill=gray!10] (ISP) at (0, 1.2) {$J_z^i=0$};
			\node[draw,rounded corners,minimum width=1.7cm,minimum height=0.9cm,
			fill=gray!10] (IMB) at (0,-0.2) {$J_z^i\neq0$};
			\node[draw,rounded corners,minimum width=1.9cm,minimum height=0.9cm,
			fill=gray!10] (FSP) at (4, 1.2) {$J_z^f=0$};
			\node[draw,rounded corners,minimum width=1.9cm,minimum height=0.9cm,
			fill=gray!10] (FMB) at (4,-0.2) {$J_z^f\neq0$};
			\draw[->,thick] (ISP) -- node[midway,above]
			{\scriptsize SP$\to$SP} (FSP);
			\draw[->,thick] (ISP) -- node[pos=0.02,above,sloped,anchor=south west]
			{\scriptsize SP$\to$MB} (FMB);
			\draw[->,thick] (IMB) -- node[pos=0.02,below,sloped,anchor=north west]
			{\scriptsize MB$\to$SP} (FSP);
			\draw[->,thick] (IMB) -- node[midway,below]
			{\scriptsize MB$\to$MB} (FMB);
			\node[anchor=north] at (2,-1.0)
			{Topological initial phase $\delta_i<0$
				$\to$ final Hamiltonian at $\delta_f>0$};
		\end{tikzpicture}
		\caption{Schematic of the four quench protocols. Each arrow represents a strong quench from a topological initial ground state at $\delta_i=-0.95$ to dynamics generated by a final Hamiltonian at $\delta_f=0.95$, with SP or MB character determined by whether $J_z^{i,f}$ vanishes. The labels $\delta_i$ and $\delta_f$ refer to the equilibrium phase diagram of the Hamiltonians, not to a topological classification of the time-evolved state.}
		\label{fig:quench-schematic}
	\end{figure}
	
	\paragraph{Free post-quench class (SP$\to$SP and MB$\to$SP).}
	In the SP$\to$SP quench the dynamics is solved exactly by a correlation-matrix approach, benchmarked against exact diagonalization to machine precision. The boundary-mode return amplitude decays asymptotically as $t^{-3/2}$, with a dimerization-dependent envelope time and beat period computed analytically; this fixes the intrinsic relaxation scale of the free boundary dynamics and provides the quantitative reference for all other protocols.
	
	In the MB$\to$SP quench we prove that the edge dynamics is exactly governed by the free post-quench Hamiltonian together with a dressed one-body density matrix, derive a spectral representation and a uniform-in-time operator-norm bound on the deviation from the SP$\to$SP reference, and validate these statements by comparing full many-body and reduced one-body evolution, by the invariance of oscillation frequencies under changes of $J_z^i$, and by agreement with the diagonal ensemble at long times.
	
	\paragraph{Genuinely interacting post-quench class
		(SP$\to$MB and MB$\to$MB).}
	In the SP$\to$MB quench the final Hamiltonian decomposes into decoupled dimers plus a weak inter-dimer perturbation, and we derive finite-time stability bounds that perturbatively control the departure from the exactly dimerized limit. A continuity equation identifies the leakage mechanism: first-dimer magnetization is lost through the weak XY current across the boundary bond, while the ZZ interaction does not directly transport $z$ magnetization across that interface. Away from local resonances, a finite-order boundary normal-form argument shows that the first-dimer magnetization can be dressed into an almost conserved boundary charge whose lifetime grows as an arbitrarily large power of the inverse inter-dimer coupling $1/b$, giving precise meaning to arbitrarily long boundary memory without claiming the absence of eventual thermalization at fixed $b$. TEBD simulations show that final interactions enhance the time-averaged edge magnetization over the accessible window, with a non-monotonic $J_z^f$ dependence arising from the isotropic leakage resonance at $J_z^f=1$ and from channel-dependent perturbative effects at small $J_z^f$; crucially, even at the resonance the memory exceeds the SP reference. The MB$\to$MB quench confirms that these structural bounds and the memory enhancement persist when both preparation and post-quench evolution are many-body, with $J_z^i$ dressing the initial boundary correlations quantitatively while the dominant enhancement---controlled by $J_z^f$---remains intact.
	
	\section{Non-Interacting Limit: Exact Analytical Results}
	\label{sec:single}
	Before turning to the interacting problem, we treat the non-interacting ($J_z = 0$) limit, in which the bond-alternating XXZ chain maps under the Jordan-Wigner transformation to the free spinless-fermion SSH chain, as a controlled problem in its own right and derive a set of exact results that will also anchor the many-body discussion. As in the rest of this work, this section is formulated in the fixed total zero-magnetization sector, $S_{\rm tot}^z = 0$, which is equivalent after the Jordan--Wigner transformation to half-filling. Specifically, we formulate the quench dynamics in terms of the correlation matrix, obtain exact expressions for the edge magnetization and connected site autocorrelation function, and prove an exact integral representation and asymptotic $t^{-3/2}$ decay law for the boundary-mode return amplitude, together with its dimerization-dependent envelope and beat scales. These results furnish a self-contained description of the integrable boundary dynamics and, at the same time, a quantitatively controlled baseline for the interacting quenches studied in Sec.~\ref{sec:manybody}. In this free-fermion setting, the appropriate physical language is relaxation to the free diagonal ensemble rather than thermalization.

	\subsection{Correlation Matrix Approach}
	\label{sec:corrmat}
	In the non-interacting limit, the half-filled many-body ground state is a Slater determinant,
	\begin{equation}
		|\Psi_0\rangle = \prod_{\ell=1}^{N/2} d_\ell^\dagger |{\rm vac}\rangle,
		\label{eq:slater}
	\end{equation}
	where $d_\ell^\dagger$ creates a fermion in the $\ell$-th occupied single-particle eigenstate of the initial Hamiltonian $H_i$. We reserve $c_n^\dagger$ and $c_n$ for site-basis operators throughout, so that the correlation-matrix formulas below do not overload notation.
	All equal-time expectation values of products of fermionic operators can be computed exactly from the one-body correlation matrix~\cite{Ingo_Peschel_2003,Peschel_2009},
	\begin{equation}
		C_{mn}(t) = \langle \Psi(t) | c_m^\dagger c_n | \Psi(t) \rangle,
		\label{eq:corrmat}
	\end{equation}
	whose time evolution under the final Hamiltonian $H_f$ is given by
	\begin{equation}
		C(t) = U(t)\, C(0)\, U^\dagger(t),
		\label{eq:corrmat_time}
	\end{equation}
	where $U(t) = e^{-i H_f^{(1)}t}$ is the single-particle time-evolution matrix and $H_f^{(1)}$ (equivalently, also denoted by $h_f$) is the single-particle Hamiltonian matrix of $H_f$.
	Via the Jordan--Wigner transformation, the local spin expectation value at site $n$ is
	\begin{equation}
		\langle S_n^z(t) \rangle = C_{nn}(t) - \frac{1}{2},
		\label{eq:mag}
	\end{equation}
	where $S_n^z = \sigma_n^z/2$. Thus the edge magnetization is directly read off from the diagonal of $C(t)$.
	Higher-order correlation functions, such as the site spin autocorrelator, can be expressed similarly in terms of $C(t)$; a detailed derivation and a benchmark against exact diagonalization are provided in the Supplemental Material (SM)~\ref*{app:single_particle_benchmarks}.
	
	\subsection{Exact Decay of the Boundary-Mode Return Amplitude: SP $\to$ SP}
	\label{sec:theorem1}
	We now derive an exact closed-form expression for the return amplitude associated with the initial left boundary zero mode in the non-interacting ($J_z = 0$) case and extract its asymptotic decay law.
	The key quantity is the boundary return amplitude
	\begin{equation}
		A_L(t) \equiv \langle A_1 | e^{-i H_f^{(1)} t} | \phi_L \rangle,
		\label{eq:AL_def}
	\end{equation}
	where $H_f^{(1)}$ is the single-particle Hamiltonian matrix of the final Hamiltonian $H_f$ defined in Eq.~\eqref{eq:corrmat_time}, $|\phi_L\rangle$ is the initial left boundary zero mode, and $|A_1\rangle$ is the one-particle basis state localized on the leftmost $A$ site.
	Thus $A_L(t)$ is a transition amplitude for the initial left boundary mode to return to the left boundary site at time $t$, not itself an observable.
	If the initial one-body density matrix is decomposed into the rank-one contribution $|\phi_L\rangle\langle\phi_L|$ and the remaining occupied modes, then $|A_L(t)|^2$ is exactly the contribution of that boundary-mode component to the density on site $1$.
	The observables are the site density $n_1(t)$ and the edge magnetization $S_1^z(t) = n_1(t) - \frac{1}{2}$, which contain this contribution together with contributions from the other initially occupied modes.
	\begin{theorem}[Exact decay of the boundary-mode return amplitude, SP $\to$ SP]
		\label{thm:free_decay}
		Let the initial topological Hamiltonian have SSH one-body hopping amplitudes $J_1^i \equiv 2J_{1}^{xy,i}$ and $J_2^i \equiv 2J_{2}^{xy,i}$, with $J_1^i < J_2^i$, so that the left boundary zero mode is (see SM~\ref*{app:zero_mode})
		\begin{equation}
			|\phi_L\rangle = \mathcal{N} \sum_{m \geq 1} r^{m-1} |A_m\rangle,
			\qquad r = -\frac{J_1^i}{J_2^i}, \qquad |r| < 1,
			\label{eq:thm_zeromode}
		\end{equation}
		with $\mathcal{N} = \sqrt{1-r^2}$. Let the final Hamiltonian have SSH one-body hopping amplitudes $a = J_1^f = 2J_{1}^{xy,f}$ and $b = J_2^f = 2J_{2}^{xy,f}$, with $a>b>0$. Then the boundary return amplitude~(\ref{eq:AL_def}) has the exact representation
		\begin{equation}
			A_L(t) = \int_0^\pi dk\, F(k) \cos[\varepsilon(k)\, t],
			\label{eq:AL_exact}
		\end{equation}
		where $\varepsilon(k) = \sqrt{a^2 + b^2 + 2ab\cos k}$ is the single-particle dispersion of the final SSH chain. Since both $|A_1\rangle$ and $|\phi_L\rangle$ are real vectors and $H_f^{(1)}$ is a real symmetric matrix, one has $A_L(t) \in \mathbb{R}$ for all $t$, and the cosine representation~(\ref{eq:AL_exact}) gives the full amplitude, not merely its real part.
		The spectral weight function is
		\begin{equation}
			F(k) = \frac{2\mathcal{N}\, a(a + br)\sin^2 k}
			{\pi\, \varepsilon(k)^2 \left(1 - 2r\cos k + r^2\right)}.
			\label{eq:Fk}
		\end{equation}
		Moreover, as $t \to \infty$,
		\begin{equation}
			\begin{aligned}
				A_L(t) =&
				\frac{C_0 \cos\!\left[(a+b)t - \tfrac{3\pi}{4}\right]
					+ C_\pi \cos\!\left[(a-b)t + \tfrac{3\pi}{4}\right]}
				{t^{3/2}} \\
				&+ O(t^{-5/2}),
			\end{aligned}
			\label{eq:AL_asymp}
		\end{equation}
		with prefactors
		\begin{align}
			C_0 &= \frac{\mathcal{N}\, a(a+br)}{2\sqrt{\pi}(a+b)^2(1-r)^2}
			\left[\frac{2(a+b)}{ab}\right]^{3/2},
			\label{eq:C0} \\[4pt]
			C_\pi &= \frac{\mathcal{N}\, a(a+br)}{2\sqrt{\pi}(a-b)^2(1+r)^2}
			\left[\frac{2(a-b)}{ab}\right]^{3/2}.
			\label{eq:Cpi}
		\end{align}
		In particular, $|A_L(t)| = O(t^{-3/2})$ and $|A_L(t)|^2 = O(t^{-3})$.
	\end{theorem}
	Thus, we get the leading large-$t$ asymptotic limit, denoted by $A_L^{\rm asymp}(t)$, obtained from Eq.~(\ref{eq:AL_asymp})
	\begin{equation}
		A_L^{\rm asymp}(t) \equiv \frac{C_0 \cos\!\left[(a+b)t - \tfrac{3\pi}{4}\right]
			+ C_\pi \cos\!\left[(a-b)t + \tfrac{3\pi}{4}\right]}{t^{3/2}}.
		\label{eq:AL_asymp_def}
	\end{equation}
	The proof of Theorem~\ref{thm:free_decay} is given in SM~\ref*{app:theorem1_proof}.
	Figure~\ref{fig:AL_benchmark} benchmarks the exact integral representation~\eqref{eq:AL_exact} against direct propagation of the initial left boundary mode on a finite open chain for the main-text quench with $N=200$, $\delta_i=-0.95$, and $\delta_f=+0.95$. The two calculations agree to numerical precision throughout the plotted interval $0 \le t \le 200$, consistent with the fact that, for the strongly localized initial boundary mode considered here, the residual difference between the semi-infinite formula and the finite chain is exponentially small in system size. Figure~\ref{fig:AL_asymptotic_benchmark} further compares the exact result with its leading large-time approximation $A_L^{\rm asymp}(t)$ in Eq.~\eqref{eq:AL_asymp_def}, while SM~\ref*{app:delta07}, Fig.~\ref*{fig:AL_benchmark_d07}, shows the analogous exact-versus-finite-chain benchmark for the smaller-dimerization quench $\delta_i=-0.70$, $\delta_f=+0.70$.

	The theorem also implies an asymptotic envelope estimate for the boundary-mode return amplitude.
	From Eq.~(\ref{eq:AL_asymp_def}), one has
	\begin{equation}
		|A_L^{\rm asymp}(t)| \leq \frac{|C_0| + |C_\pi|}{t^{3/2}}.
		\label{eq:AL_asymp_bound}
	\end{equation}
	For any fixed sufficiently small threshold $\epsilon > 0$, it is natural to define the associated asymptotic envelope time by the condition that the right-hand side of Eq.~(\ref{eq:AL_asymp_bound}) equals $\epsilon$, namely
	\begin{equation}
		\tau_\epsilon^{\rm env} \equiv \left(\frac{|C_0| + |C_\pi|}{\epsilon}\right)^{2/3}.
		\label{eq:tau_env_def}
	\end{equation}
	This quantity is not a sharp relaxation time and it is not defined for the full edge magnetization.
	Instead, it is an asymptotic scale extracted from the leading large-$t$ envelope of the boundary-mode return amplitude.
	To determine how it depends on the final dimerization, consider the strongly dimerized trivial regime $b \ll a$.
	In this regime, one has $a \pm b = a + O(b)$, so the $b$-dependent corrections to $a$ can be neglected at leading order.
	Under this replacement, Eqs.~(\ref{eq:C0}) and~(\ref{eq:Cpi}) tells that both $C_0$ and $C_\pi$ are individually proportional to $b^{-3/2} + O(b^{-1/2})$ at leading order and therefore,
	\begin{equation}
		|C_0| + |C_\pi| \propto b^{-3/2}.
		\label{eq:Csum_smallb}
	\end{equation}
	Substituting this scaling into Eq.~(\ref{eq:tau_env_def}) gives the following corollary.
	
	\begin{corollary}[Dimerization scaling of the envelope time]
		\label{cor:tau_scaling}
		In the strongly dimerized trivial regime $b \ll a$, where $a$ and $b$ are SSH one-body hopping amplitudes, the asymptotic envelope time defined in Eq.~\eqref{eq:tau_env_def} satisfies
		\begin{equation}
			\tau_\epsilon^{\rm env} \propto b^{-1} = \frac{1}{J_2^f} = \frac{1}{2(1-\delta_f)}.
			\label{eq:tau_scaling}
		\end{equation}
		In particular, reducing $\delta_f$ from $0.95$ to $0.70$ reduces $\tau_\epsilon^{\rm env}$ to $b_{0.95}/b_{0.70} = 0.10/0.60 \approx 1/6$ of its value.
	\end{corollary}
	
	The asymptotic formula~\eqref{eq:AL_asymp_def} contains two cosines oscillating at frequencies $\omega_0 = a+b$ and $\omega_\pi = a-b$. Adding two sinusoids at close frequencies produces a \emph{beating} pattern in which the amplitude periodically swells and collapses. The period of this modulation is the beat period
	\begin{equation}
		T_{\rm beat} = \frac{2\pi}{\omega_0 - \omega_\pi} = \frac{2\pi}{2b} = \frac{\pi}{b}
		= \frac{\pi}{J_2^f}.
		\label{eq:Tbeat_def}
	\end{equation}
	At $\delta_f = 0.95$ one has $b = 0.10$ and $T_{\rm beat} = \pi/0.10 \approx 31.4$, so the two band-edge frequencies $\omega_0 = 4.00$ and $\omega_\pi = 3.80$ are nearly equal and interfere on a long timescale. At $\delta_f = 0.70$ one has $b = 0.60$ and $T_{\rm beat} = \pi/0.60 \approx 5.24$, so the interference is much faster. The beat period is reduced by the same factor of $1/6$ as the envelope time. These statements apply directly to $A_L(t)$ and to its asymptotic envelope.
	The faster decay seen numerically in the full edge magnetization at smaller dimerization is consistent with the same mechanism, but no separate theorem is claimed here for the full observable.
	\begin{figure}[t]
		\centering
		\includegraphics[width=\linewidth]{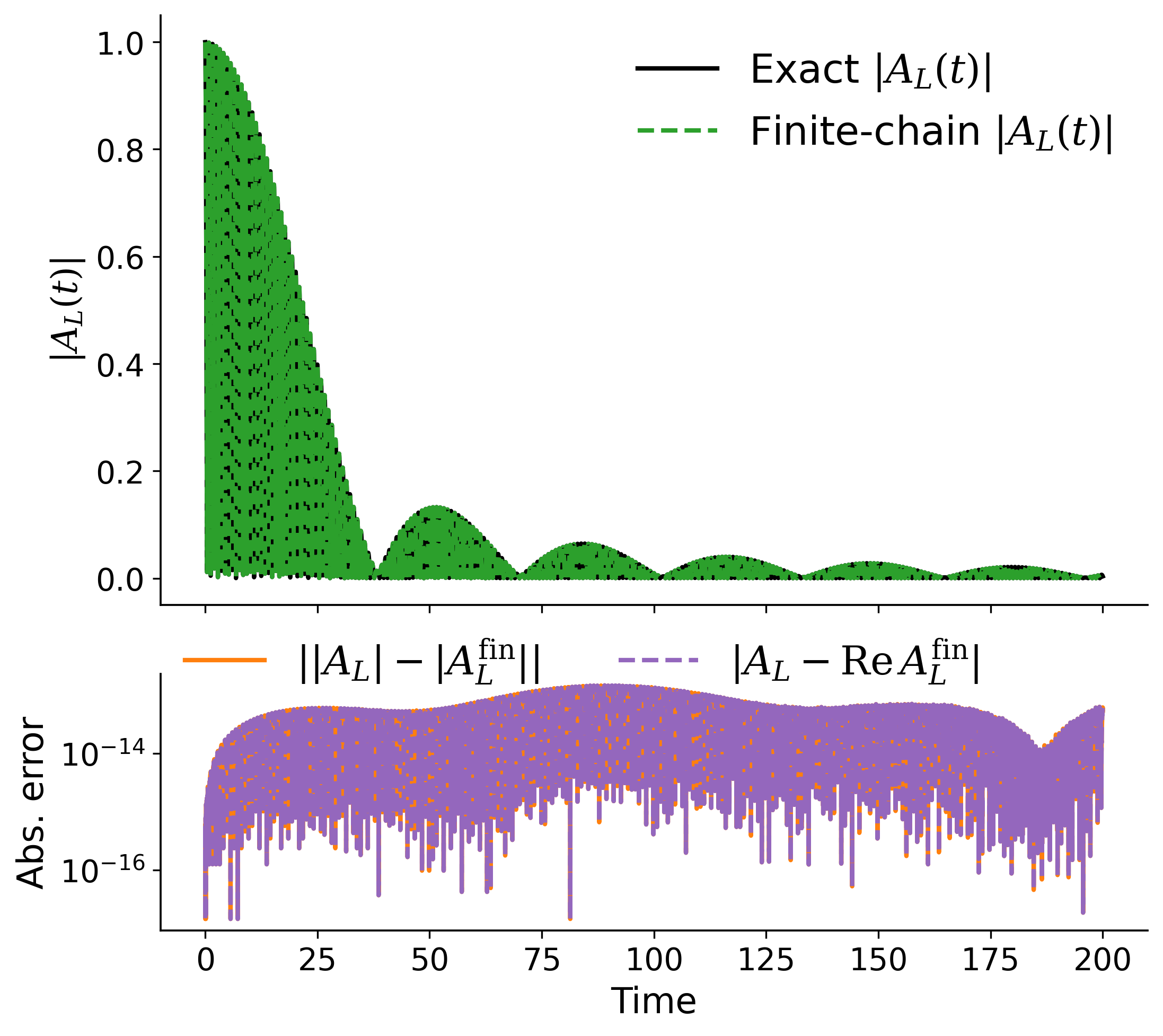}
		\caption{Benchmark of the exact boundary-return amplitude for the SP $\to$ SP quench with $L = 200$, $\delta_i = -0.95$, and $\delta_f = +0.95$. Since both $|\phi_L\rangle$ and $|A_1\rangle$ are real vectors and $H_f^{(1)}$ is a real symmetric matrix, $A_L(t)$ is real for all $t$ (see Theorem~\ref{thm:free_decay}), so $|A_L(t)| = |{\rm Re}\,A_L(t)|$ and the modulus comparison is unambiguous. The upper panel compares $|A_L(t)|$ obtained from the exact semi-infinite integral representation, Eq.~(\ref{eq:AL_exact})--(\ref{eq:Fk}), with direct propagation of the initial left boundary mode, Eq.~(\ref{eq:thm_zeromode}), on the finite open chain of length $L = 200$. The lower panel shows two error measures over the full interval $0 \leq t \leq 200$: the modulus difference $\bigl||A_L(t)| - |A_L^{\rm fin}(t)|\bigr|$ and the real-part difference $|A_L(t) - \mathrm{Re}\,A_L^{\rm fin}(t)|$, which coincide with $|A_L(t) - A_L^{\rm fin}(t)|$ since $A_L(t)$ is real. Both are at numerical precision throughout the plotted window.}
		\label{fig:AL_benchmark}
	\end{figure}
	
	\begin{figure}[htbp]
		\centering
		\includegraphics[width=\linewidth]{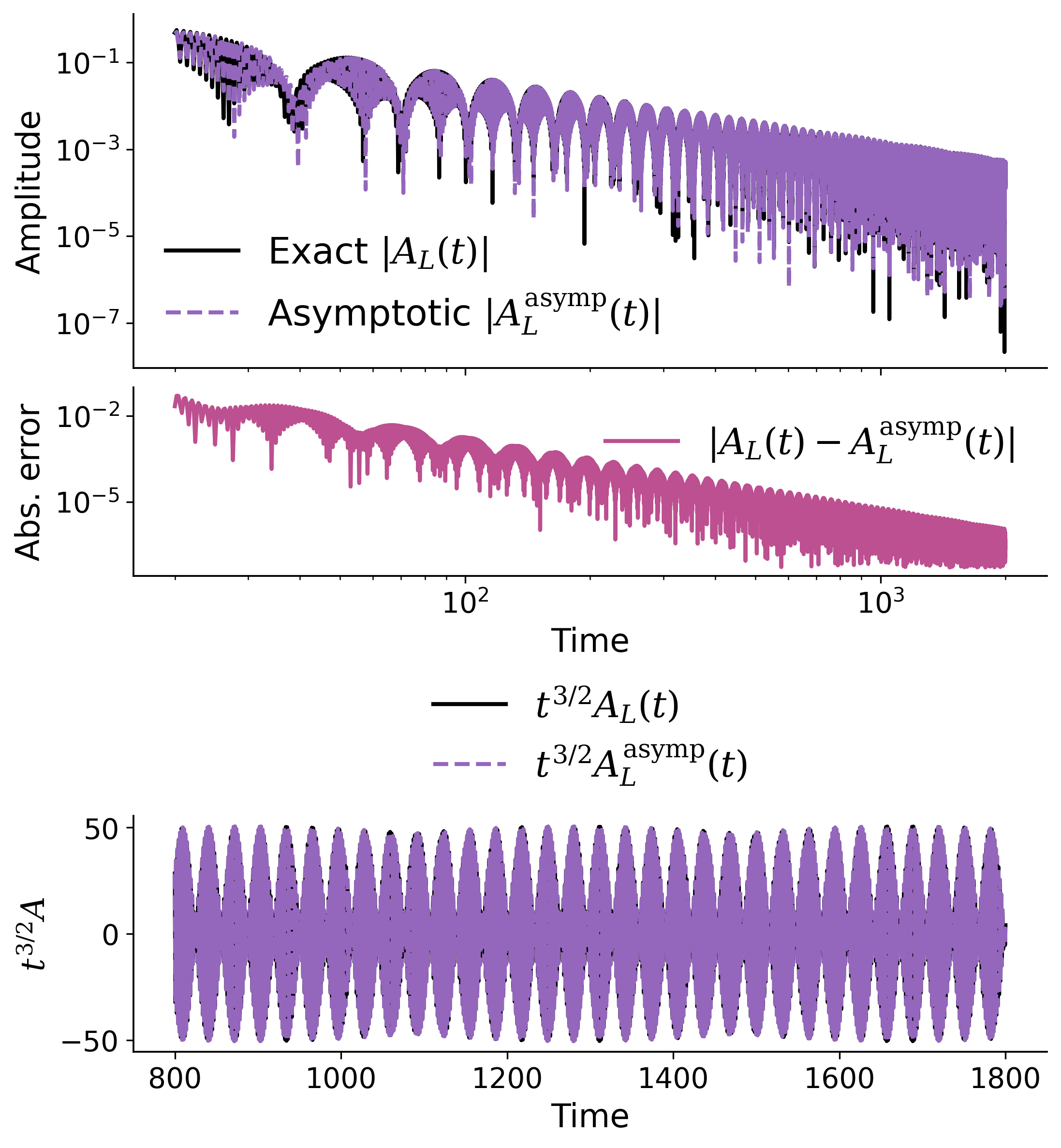}
		\caption{Benchmark of the large-time asymptotic formula for the SP $\to$ SP quench with $L = 200$, $\delta_i = -0.95$, and $\delta_f = +0.95$. The upper panel compares the exact amplitude $|A_L(t)|$ with the leading asymptotic approximation $|A_L^{\rm asymp}(t)|$ from Eq.~(\ref{eq:AL_asymp_def}) over the interval $20 \leq t \leq 2000$. The middle panel shows the absolute error $|A_L(t) - A_L^{\rm asymp}(t)|$. The lower panel shows the scaled quantities $t^{3/2}A_L(t)$ and $t^{3/2}A_L^{\rm asymp}(t)$ over the late-time window $800 \leq t \leq 1800$. The bounded oscillatory behavior in the lower panel is the numerical signature of the asymptotic law $A_L(t) = O(t^{-3/2})$.}
		\label{fig:AL_asymptotic_benchmark}
	\end{figure}
	
	\subsection{Quench Dynamics in the Non-Interacting Chain}
	\label{sec:sp_results}
	We now present results for the SP $\to$ SP quench ($J_z^i = J_z^f = 0$) at system size $L = 200$, quenching from $\delta_i = -0.95$ to $\delta_f = +0.95$. The evolution of the edge magnetization is shown in Fig.~\ref{fig:single_particle}.
	Both edge sites start with large magnitude, $|S_1^z(0)| \approx |S_N^z(0)| \approx 0.5$, reflecting the strong localization of the initial boundary modes. As explained in SM~\ref*{app:ssh_zero_mode}, this magnitude is fixed by the boundary zero-mode weight, whereas the sign depends on the occupation convention for the exponentially localized edge doublet; in a convention that pins the left zero mode to be occupied one has $S_1^z(0) \approx +1/2$, while in the opposite sector the sign flips.
	Under the final Hamiltonian, whose ground state lies on the trivial side of the equilibrium SSH phase diagram and therefore supports no protected edge mode, the edge signal undergoes oscillatory relaxation toward zero. By contrast, the center-site magnetization remains at numerical zero throughout the evolution, as shown in SM~\ref*{app:center_spin}, Fig.~\ref*{fig:center_magnetization}, consistent with the fact that the initial topological zero mode is exponentially localized at the boundary and carries negligible weight in the bulk.
	The theorem in Sec.~\ref{sec:theorem1} concerns the return amplitude of the initial left boundary mode and shows that the corresponding boundary-mode contribution decays asymptotically as $t^{-3/2}$. The full edge magnetization shown in Fig.~\ref{fig:single_particle} is a different quantity and also contains contributions from the other initially occupied modes, so no asymptotic power-law exponent is extracted here from the full observable. Nevertheless, the relaxation of both the edge magnetizations occurs on a timescale comparable to the first decay of $|A_L(t)|$ in Fig.~\ref{fig:AL_benchmark}, before the subsequent decayed oscillatory revival.
	Because the post-quench SSH Hamiltonian is free, it possesses an extensive set of conserved mode occupations $\hat{n}_\alpha = d^\dagger_\alpha d_\alpha$, precluding thermalization in the ETH sense; the correct long-time language is relaxation to the free diagonal ensemble, whose predictions are determined entirely by the single-particle mode occupations of the initial state and are derived explicitly in Corollary~\ref{cor:mbtosp_diagonal_ensemble}. This is not an infinite-temperature result: the free diagonal ensemble retains non-trivial, state-dependent memory of the initial occupation pattern, and coincides with $\langle S_1^z \rangle \to 0$ only because the final trivial SSH Hamiltonian supports no protected edge mode, not because all mode occupations are uniformly $\tfrac{1}{2}$. For comparison, SM~\ref*{app:delta07} shows the quench with $\delta_i = -0.70$ and $\delta_f = 0.70$ in Fig.~\ref*{fig:delta07_edge}.
	The same relaxation mechanism is observed there, but with visibly faster relaxation, consistent with the smaller asymptotic envelope scale \eqref{eq:tau_scaling} and shorter beat period obtained for the boundary-mode return amplitude \eqref{eq:Tbeat_def}.

	\begin{figure}[htbp]
		\centering
		\includegraphics[width=\linewidth]{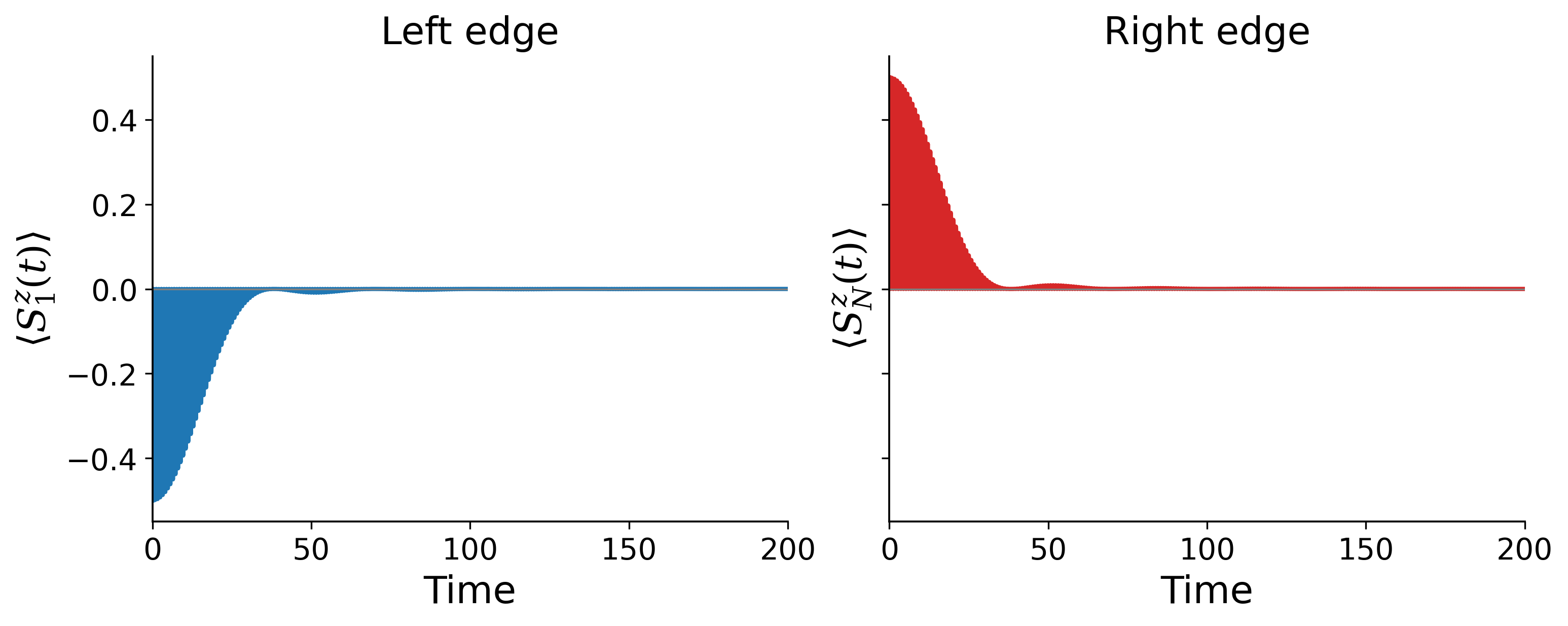}
		\caption{Quench dynamics of the edge magnetization in the non-interacting ($J_z = 0$) SSH chain for $L = 200$, quenching from $\delta_i = -0.95$ to $\delta_f = +0.95$ (recall hopping parameter $|t|=1$ throughout) using the correlation matrix method. The left and right panels show $S_1^z(t)$ and $S_N^z(t)$, respectively. Both edges start with magnitude close to $1/2$ and then exhibit oscillatory relaxation toward zero under evolution generated by the final Hamiltonian at $\delta_f=0.95$, whose ground state lies on the trivial side of the equilibrium phase diagram. The initial sign depends on the chosen edge-sector convention as further explained in SM~\ref*{app:ssh_zero_mode}, whereas the magnitude is fixed by the boundary zero-mode weight. The exact theorem in Sec.~\ref{sec:theorem1} applies to the return amplitude of the initial left boundary mode and to its induced contribution to the site-1 density, not to the full edge magnetization plotted here; nevertheless, the relaxation of both the edge signal occurs on a timescale comparable to the first decay of $|A_L(t)|$ seen in Fig.~\ref{fig:AL_benchmark}.}
		\label{fig:single_particle}
	\end{figure}

	\section{Interacting Quenches and Many-Body Protection of Edge Memory}
	\label{sec:manybody}

	We now turn to quenches in which at least one side of the protocol is interacting. Whereas Sec.~\ref{sec:single} established the free-fermion reference problem through exact correlation-matrix formulas and an exact theorem for the boundary-mode return amplitude, the interacting protocols separate naturally into two classes. We begin with the MB to SP quench, in which the initial state is interacting but the post-quench Hamiltonian is free, so that the dynamics remains analytically controlled through the interacting initial one-body density matrix. We then turn to the genuinely interacting post-quench protocols SP to MB and MB to MB, for which we combine rigorous analytical control in the strongly dimerized final regime with TEBD simulations of the edge magnetization in order to isolate the mechanisms of boundary-memory leakage and to demonstrate that interacting final Hamiltonians can enhance finite-time topological memory retention relative to the non-interacting reference, with a local suppression (still being higher than SP reference point) near resonant leakage regimes such as the isotropic point. Throughout this section, SP denotes $J_z=0$ and MB denotes $J_z\neq 0$ in Eq.~\eqref{eq:Ham} with couplings in Eq.~\eqref{eq:params}.

	\subsection{Many-Body to Single-Particle Quench}
	\label{sec:mbtosp}
	We first consider quenches from an interacting topological initial state to a non-interacting final Hamiltonian whose ground state lies on the trivial side of the equilibrium phase diagram. Thus $J_z^i\neq 0$ and $J_z^f=0$. As in the rest of this work, the quench is formulated in the fixed total zero-magnetization sector, $S_{\rm tot}^z=0$, which is equivalent after the Jordan--Wigner transformation to fixed half-filling. The interacting initial state and the SP reference state are therefore compared at the same conserved fermion number. After the Jordan--Wigner transformation, the post-quench Hamiltonian is quadratic, so the evolution of the local density is exactly determined by the one-body density matrix of the initial state. This statement does not require the initial state to be Gaussian. The interaction in the preparation Hamiltonian changes the initial one-body density matrix, but it does not change the single-particle spectrum or the dephasing structure of the final evolution. Let the final quadratic Hamiltonian be
	\begin{align}
		H_f=\sum_{m,n=1}^{N} c_m^\dagger h^{f}_{mn} c_n,
		\label{eqmbtospfreeH}
	\end{align}
	where $h_f \equiv H_f^{(1)}$ is the $N\times N$ single-particle SSH Hamiltonian matrix of the final chain at $\delta_f>0$, with $H_f^{(1)}$ as defined in Eq.~\eqref{eq:corrmat_time}. For the interacting initial ground state $|\Psi_i(J_z^i)\rangle$ in the fixed $S_{\rm tot}^z=0$ sector, define the site-basis one-body density matrix
	\begin{align}
		[C_i^{\rm MB}(J_z^i)]_{mn}
		=
		\langle\Psi_i(J_z^i)|c_m^\dagger c_n|\Psi_i(J_z^i)\rangle.
		\label{eq:mbtosp_C_MB}
	\end{align}
	Throughout this subsection, matrices denoted by $C$ are written in the site basis and carry site indices, such as $m,n$. Matrices denoted by $\mathcal{C}$ are written in the eigenmode basis of the final single-particle Hamiltonian and carry mode indices, such as $\mu,\nu$. These are two representations of the same one-body density operator.
	\begin{proposition}[Exact reduction of the MB to SP quench]
		\label{prop:mbtosp_exact_reduction}
		Define
		\begin{align}
			U_f(t)=e^{-ih_f t},
			\label{eq:mbtosp_U}
		\end{align}
		so that the Heisenberg evolution of the site operators is
		\begin{align}
			c_j(t)=e^{iH_f t}c_j e^{-iH_f t}
			=\sum_{m=1}^{N}[U_f(t)]_{jm}c_m.
			\label{eq:mbtosp_c_evol_main}
		\end{align}
		Then the MB to SP edge magnetization is given exactly by
		\begin{equation}
			\begin{aligned}
				S_{1,{\rm MB}\to{\rm SP}}^z&(t;J_z^i)
				= \\
				&\sum_{m,n=1}^{N}
				[U_f(t)]_{1m}^*
				[C_i^{\rm MB}(J_z^i)]_{mn}
				[U_f(t)]_{1n}
				-\frac{1}{2},
			\end{aligned}
			\label{eq:mbtosp_exact_sum}
		\end{equation}
		or equivalently by
		\begin{align}
			S_{1,{\rm MB}\to{\rm SP}}^z(t;J_z^i)
			=
			\left[
			U_f^*(t)C_i^{\rm MB}(J_z^i)U_f^T(t)
			\right]_{11}
			-\frac{1}{2}.
			\label{eq:mbtosp_exact_C}
		\end{align}
		Equations~\eqref{eq:mbtosp_exact_sum} and~\eqref{eq:mbtosp_exact_C} are the precise sense in which the MB to SP quench is governed by the same free post-quench problem as the SP to SP quench. The replacement is
		\begin{align}
			C_i^{\rm SP}\longrightarrow C_i^{\rm MB}(J_z^i),
			\label{eq:mbtosp_C_replacement}
		\end{align}
		where $C_i^{\rm SP}$ is the one-body density matrix of the non-interacting topological initial state in the same fixed $S_{\rm tot}^z=0$ sector.
	\end{proposition}
	Equation~\eqref{eq:mbtosp_exact_sum} follows by inserting Eq.~\eqref{eq:mbtosp_c_evol_main} into $S_1^z(t)=c_1^\dagger(t)c_1(t)-\frac{1}{2}$ and evaluating the expectation value in $|\Psi_i(J_z^i)\rangle$. Equation~\eqref{eq:mbtosp_exact_C} is the corresponding matrix form.
	To display the frequency content explicitly, diagonalize the final single-particle Hamiltonian. Let $\{|\mu\rangle\}_{\mu=1}^{N}$ be an orthonormal eigenbasis of $h_f$,
	\begin{align}
		h_f|\mu\rangle=\varepsilon_\mu|\mu\rangle,
		\qquad
		\sum_{j=1}^{N}\varphi_\mu^*(j)\varphi_\nu(j)=\delta_{\mu\nu},
		\label{eq:mbtosp_final_eigenvectors}
	\end{align}
	where
	\begin{align}
		\varphi_\mu(j)=\langle j|\mu\rangle
		\label{eq:mbtosp_phi_def}
	\end{align}
	is the amplitude of eigenmode $\mu$ on site $j$. The corresponding fermionic normal-mode operators are defined by
	\begin{align}
		d_\mu^\dagger=\sum_{j=1}^{N}\varphi_\mu(j)c_j^\dagger,
		\qquad
		d_\mu=\sum_{j=1}^{N}\varphi_\mu^*(j)c_j,
		\label{eq:mbtosp_d_def}
	\end{align}
	and the inverse transformation is
	\begin{align}
		c_j^\dagger=\sum_{\mu=1}^{N}\varphi_\mu^*(j)d_\mu^\dagger,
		\qquad
		c_j=\sum_{\mu=1}^{N}\varphi_\mu(j)d_\mu.
		\label{eq:mbtosp_inverse_modes}
	\end{align}
	The same initial one-body density matrix can be represented in the final-mode basis. We denote this representation by $\mathcal{C}^{\rm MB}$:
	\begin{align}
		\mathcal{C}_{\mu\nu}^{\rm MB}(J_z^i)
		&=
		\langle\Psi_i(J_z^i)|d_\mu^\dagger d_\nu|\Psi_i(J_z^i)\rangle
		\nonumber\\
		&=
		\sum_{m,n=1}^{N}
		\varphi_\mu(m)
		[C_i^{\rm MB}(J_z^i)]_{mn}
		\varphi_\nu^*(n).
		\label{eq:mbtosp_C_final_basis}
	\end{align}
	Thus $C_i^{\rm MB}$ and $\mathcal{C}^{\rm MB}$ contain the same one-body information, but in different bases. The edge magnetization is then
	\begin{align}
		S_{1,{\rm MB}\to{\rm SP}}^z(t;J_z^i)
		=
		\sum_{\mu,\nu=1}^{N}
		\varphi_\mu^*(1)\varphi_\nu(1)
		e^{i(\varepsilon_\mu-\varepsilon_\nu)t}
		\mathcal{C}_{\mu\nu}^{\rm MB}(J_z^i)
		-\frac{1}{2}.
		\label{eq:mbtosp_spectral_formula}
	\end{align}
	All frequencies in Eq.~\eqref{eq:mbtosp_spectral_formula} are fixed by the final non-interacting SSH Hamiltonian. The initial interaction affects only the coefficients $\mathcal{C}_{\mu\nu}^{\rm MB}(J_z^i)$. Thus the MB to SP quench is exactly a free post-quench relaxation problem with an interacting initial one-body density matrix. Equations~\eqref{eq:mbtosp_C_final_basis} and~\eqref{eq:mbtosp_spectral_formula} follow by expanding the site operators in the final normal-mode basis using Eqs.~\eqref{eq:mbtosp_d_def} and~\eqref{eq:mbtosp_inverse_modes}.
	
	The comparison with the SP to SP reference is summarized by the following corollary. Let
	\begin{align}
		C_i^{\rm SP}
		=
		\sum_{\ell\in{\rm occ}}
		|\phi_\ell^i\rangle\langle\phi_\ell^i|
		\label{eq:mbtosp_C_SP}
	\end{align}
	be the site-basis one-body projector of the SP initial state in the same fixed $S_{\rm tot}^z=0$ sector, where $|\phi_\ell^i\rangle$ are the occupied eigenstates of the initial single-particle Hamiltonian selected by that sector. Define the site-basis difference
	\begin{align}
		\Delta C_i(J_z^i)
		=
		C_i^{\rm MB}(J_z^i)-C_i^{\rm SP}.
		\label{eq:mbtosp_delta_C}
	\end{align}
	The subtraction is performed in the site basis, so both matrices carry site indices.
	\begin{corollary}[Exact comparison with the SP to SP reference]
		\label{cor:mbtosp_difference}
		One has
		\begin{align}
			S_{1,{\rm MB}\to{\rm SP}}^z(t;J_z^i)
			-
			S_{1,{\rm SP}\to{\rm SP}}^z(t)
			=
			\left[
			U_f^*(t)\Delta C_i(J_z^i)U_f^T(t)
			\right]_{11},
			\label{eq:mbtosp_difference_identity}
		\end{align}
		and therefore
		\begin{align}
			\left|
			S_{1,{\rm MB}\to{\rm SP}}^z(t;J_z^i)
			-
			S_{1,{\rm SP}\to{\rm SP}}^z(t)
			\right|
			\le
			\|\Delta C_i(J_z^i)\|.
			\label{eq:mbtosp_difference_bound}
		\end{align}
		This bound is uniform in time. It gives a rigorous criterion for closeness to the SP to SP reference: the two initial one-body density matrices must be close in operator norm. No Wick decomposition or Gaussian assumption is involved.
	\end{corollary}
	The proof is provided in SM~\ref*{app:mbtosp_free_reduction}.
	
	The long-time average is fixed by the same final single-particle spectrum.
	\begin{corollary}[Free diagonal ensemble structure]
		\label{cor:mbtosp_diagonal_ensemble}
		For a finite chain,
		\begin{align}
			\overline{S_{1,{\rm MB}\to{\rm SP}}^z}
			=
			\sum_{\mu,\nu:\varepsilon_\mu=\varepsilon_\nu}
			\varphi_\mu^*(1)\varphi_\nu(1)
			\mathcal{C}_{\mu\nu}^{\rm MB}(J_z^i)
			-\frac{1}{2},
			\label{eq:mbtosp_diag_ensemble}
		\end{align}
		and, if the final single-particle spectrum is nondegenerate, this reduces to
		\begin{align}
			\overline{S_{1,{\rm MB}\to{\rm SP}}^z}
			=
			\sum_{\mu=1}^{N}
			|\varphi_\mu(1)|^2
			\mathcal{C}_{\mu\mu}^{\rm MB}(J_z^i)
			-\frac{1}{2}.
			\label{eq:mbtosp_diag_non_degenerate}
		\end{align}
		Equations~\eqref{eq:mbtosp_spectral_formula}--\eqref{eq:mbtosp_diag_non_degenerate} show that the diagonal ensemble is a free diagonal ensemble of the final SSH Hamiltonian. The initial interaction enters only through the mode occupations $\mathcal{C}_{\mu\mu}^{\rm MB}$ and coherences $\mathcal{C}_{\mu\nu}^{\rm MB}$ in the final eigenbasis.
	\end{corollary}

The result follows by taking the infinite-time average of Eq.~\eqref{eq:mbtosp_spectral_formula} and retaining only the stationary terms.

	It is therefore not generally correct to describe the MB to SP quench as the SP to SP quench with a single scalar renormalization. In the site basis, the exact statement is the matrix replacement
	\begin{align}
		C_i^{\rm SP}\longrightarrow C_i^{\rm MB}(J_z^i).
		\label{eq:mbtosp_matrix_renorm}
	\end{align}
	In the final-mode basis, the corresponding SP matrix is
	\begin{align}
		\mathcal{C}_{\mu\nu}^{\rm SP}
		=
		\sum_{m,n=1}^{N}
		\varphi_\mu(m)
		[C_i^{\rm SP}]_{mn}
		\varphi_\nu^*(n),
		\label{eq:mbtosp_CSP_mode_def}
	\end{align}
	and the equivalent replacement is
	\begin{align}
		\mathcal{C}_{\mu\nu}^{\rm SP}
		\longrightarrow
		\mathcal{C}_{\mu\nu}^{\rm MB}(J_z^i).
		\label{eq:mbtosp_mode_matrix_renorm}
	\end{align}
	Within the fixed $S_{\rm tot}^z=0$ sector, these matrices are expected to vary smoothly with $J_z^i$ only away from level crossings and phase boundaries. For the parameter regime considered here, the initial Hamiltonian is taken deep in the topological phase, away from the AFM and FM regimes. The sensitivity to $J_z^i$ is therefore controlled by how strongly the one-body boundary correlations entering Eq.~\eqref{eq:mbtosp_difference_identity} change under the interacting preparation. The same structure explains why the decay mechanism is inherited from the SP to SP problem. In the thermodynamic limit, the oscillatory terms in Eq.~\eqref{eq:mbtosp_spectral_formula} are governed by stationary points of the final SSH dispersion. If the spectral weights generated by $\mathcal{C}_{\mu\nu}^{\rm MB}(J_z^i)$ are smooth near the relevant stationary points and do not cancel the leading contribution, the power-law exponents are the same as in the free relaxation problem, while the prefactors are changed by the interacting initial state. This is the precise sense in which the MB to SP quench shares the SP to SP decay structure with interaction-dependent weights.

	The exact reduction provides several direct numerical diagnostics. First, starting from the interacting initial state $|\Psi_i(J_z^i)\rangle$, one may compute $S_1^z(t)$ in two independent ways under the same non-interacting final Hamiltonian with $J_z^f=0$: either by evolving the initial one-body density matrix through Eq.~\eqref{eq:mbtosp_exact_C}, or by evolving the full many-body state under $H_f$ and evaluating $S_1^z(t)$ directly. Agreement between these two calculations verifies that the MB to SP quench is completely determined by one-body data of the interacting initial state, without any Gaussian or Wick-factorization assumption. Second, the comparison with the SP to SP reference can be tested using Corollary~\ref{cor:mbtosp_difference}: the full time-dependent difference $S_{1,{\rm MB}\to{\rm SP}}^z(t;J_z^i)-S_{1,{\rm SP}\to{\rm SP}}^z(t)$ must coincide with the propagation of the matrix difference $\Delta C_i(J_z^i)$ in Eq.~\eqref{eq:mbtosp_difference_identity}, and its magnitude must remain bounded by $\|\Delta C_i(J_z^i)\|$ uniformly in time. This check isolates the origin of the MB to SP deviation from the SP to SP curve and confirms that it comes only from the interacting preparation. Third, the spectral representation in Eq.~\eqref{eq:mbtosp_spectral_formula} gives a frequency-domain and diagonal-ensemble test. As $J_z^i$ is varied, the Fourier spectrum of $S_1^z(t)$ may change in peak weights because the coefficients $\mathcal{C}_{\mu\nu}^{\rm MB}(J_z^i)$ change, but the peak positions must remain fixed by the final SSH energy differences $\varepsilon_\mu-\varepsilon_\nu$. In addition, the long-time average of the same signal must agree with the free diagonal-ensemble prediction in Corollary~\ref{cor:mbtosp_diagonal_ensemble}, reducing to Eq.~\eqref{eq:mbtosp_diag_non_degenerate} when the final single-particle spectrum is nondegenerate. These three checks together verify the real-time reduction, the controlled comparison with the SP to SP reference, and the final-Hamiltonian origin of both the oscillation frequencies and the stationary component of the MB to SP dynamics. These benchmarks are presented in SM~\ref*{app:mbtosp_benchmark}.

	\subsection{Single-Particle to Many-Body Quench}
	\label{sec:sptomb}
	
	We next consider the quench from the non-interacting topological ground state to an interacting trivial final Hamiltonian. Thus $J_z^i=0$ and $J_z^f\neq0$. This protocol isolates the effect of interactions in the post-quench evolution, since the initial state is the same as in the SP to SP reference calculation of Sec.~\ref{sec:single}. With the SSH one-body hopping amplitudes $a = J_1^f \equiv 2t(1+\delta_f)$ and
	$b = J_2^f \equiv 2t(1-\delta_f)$ (see SM~\ref*{app:ssh_conventions}), the spin
	couplings of the final Hamiltonian read
	\begin{align}
		J_1^{xy,f} = \tfrac{a}{2},\qquad
		J_2^{xy,f} = \tfrac{b}{2}, \nonumber\\
		J_1^{z,f}  = \tfrac{a}{2}\,J_z^f,\qquad
		J_2^{z,f}  = \tfrac{b}{2}\,J_z^f,
		\label{eq:ab_def_mb}
	\end{align}
	with $\delta_f > 0$, so $a > b > 0$. This scaling is essential for the analytic expansion below: when $b\to0$, both the inter-dimer $XY$ coupling and the inter-dimer Ising coupling vanish. The final Hamiltonian then becomes an exact sum of isolated two-site dimers.

	\subsubsection{Finite-time control near the decoupled-dimer limit}
	\label{sec:main}
	
	With open boundary conditions, the post-quench Hamiltonian can be written as
	\begin{align}
		H_b=H_0+\frac{b}{2}\,W,
		\label{eq:Hb_decomp}
	\end{align}
	where the decoupled-dimer Hamiltonian is
	\begin{align}
		H_0=\sum_{m=1}^{N/2} h_{2m-1,2m},
		\label{eq:H0_def}
	\end{align}
	with
	\begin{align}
		h_{2m-1,2m}
		=&\,\frac{a}{2}\Bigl(
		\sigma_{2m-1}^x\sigma_{2m}^x
		+\sigma_{2m-1}^y\sigma_{2m}^y
		\nonumber\\
		&\hspace{1.3cm}
		+J_z^f\sigma_{2m-1}^z\sigma_{2m}^z
		\Bigr),
		\label{eq:local_dimer_h}
	\end{align}
	and the unit-strength inter-dimer perturbation is
	\begin{align}
		W=\sum_{m=1}^{N/2-1} w_{2m,2m+1},
		\label{eq:W_def}
	\end{align}
	with
	\begin{align}
		w_{2m,2m+1}
		=&\,\sigma_{2m}^x\sigma_{2m+1}^x
		+\sigma_{2m}^y\sigma_{2m+1}^y
		\nonumber\\
		&+J_z^f\sigma_{2m}^z\sigma_{2m+1}^z .
		\label{eq:w_def}
	\end{align}
	Thus $H_0$ contains the strong intra-dimer bonds, while $\frac{b}{2}W$ contains the weak inter-dimer bonds. This decomposition is exact for the Hamiltonian in Eq.~\eqref{eq:params}.
	
	\begin{proposition}[Finite-time stability near the decoupled-dimer limit]
		\label{prop:finite_time_stability_sptomb}
		Let
		\begin{align}
			S_{1,b}^z(t)=e^{iH_bt}S_1^ze^{-iH_bt},\qquad
			S_{1,0}^z(t)=e^{iH_0t}S_1^ze^{-iH_0t},
			\label{eq:S1_evol_defs}
		\end{align}
		where $S_1^z=\sigma_1^z/2$. Then, for all $t\ge0$,
		\begin{align}
			\bigl\|S_{1,b}^z(t)-S_{1,0}^z(t)\bigr\|
			\le \frac{|b|}{2}\,(2+|J_z^f|)\,t .
			\label{eq:bound_main}
		\end{align}
		Consequently, for every normalized initial state $|\psi\rangle$,
		\begin{align}
			\left|
			\langle S_{1,b}^z(t)\rangle_\psi
			-\langle S_{1,0}^z(t)\rangle_\psi
			\right|
			\le \frac{|b|}{2}\,(2+|J_z^f|)\,t .
			\label{eq:expect_main}
		\end{align}
	\end{proposition}
	
	The proof is given in SM~\ref*{app:proof}. The estimate is a finite-time stability bound. It shows that, for fixed $J_z^f$, the exact interacting edge dynamics remains close to the decoupled-dimer dynamics on times $t\ll[|b|(2+|J_z^f|)]^{-1}/2$. It does not imply a monotonic improvement with increasing $J_z^f$, since the upper bound itself grows with $|J_z^f|$. Its role is instead to identify $b = J_2^f = 2(1-\delta_f)$, equivalently $J_2^{xy,f} = b/2$, as the small parameter controlling the departure from the exactly dimerized limit.

	\begin{proposition}[Boundary-dimer continuity equation]
		\label{prop:boundary_dimer_continuity}
		Define
		\begin{align}
			M_{12}^z=S_1^z+S_2^z
			=\frac{1}{2}(\sigma_1^z+\sigma_2^z).
			\label{eq:M12_def_main}
		\end{align}
		For the full Hamiltonian $H_b$, one has
		\begin{align}
			\frac{d}{dt}M_{12}^z(t)
			=\frac{b}{2}\,j_{2,3}^z(t),
			\label{eq:M12_current_main}
		\end{align}
		where
		\begin{align}
			j_{2,3}^z=
			\sigma_2^y\sigma_3^x-\sigma_2^x\sigma_3^y .
			\label{eq:j23_def_main}
		\end{align}
		Therefore,
		\begin{align}
			\left|
			\frac{d}{dt}\langle M_{12}^z(t)\rangle
			\right|
			\le |b|,
			\label{eq:current_bound_main}
		\end{align}
		and
		\begin{align}
			\left|
			\langle M_{12}^z(t)\rangle-\langle M_{12}^z(0)\rangle
			\right|
			\le |b|\,t .
			\label{eq:M12_integrated_bound_main}
		\end{align}
	\end{proposition}
	
	This identity is derived in SM~\ref*{app:leakageproof}. It shows that the total magnetization of the first dimer changes only through the weak $XY$ bond connecting sites $2$ and $3$. The $ZZ$ interactions do not directly carry $z$ magnetization across this boundary.
	
	We next record a stronger asymptotic statement in the strongly dimerized final limit. The observable controlled here is the first-dimer magnetization $Q_0=M_{12}^z=S_1^z+S_2^z$, not the single-site magnetization $S_1^z$. This is the natural rigorous object because $Q_0$ is exactly conserved when the final Hamiltonian is a sum of isolated dimers, whereas $S_1^z$ may still oscillate inside the first dimer. The result is a singular statement as $b=J_2^f=2(1-\delta_f)\to0$: it proves memory on time windows that grow as powers of $1/b$, but it does not claim absence of eventual thermalization at fixed nonzero $b$ or monotonic improvement with increasing $J_z^f$.
	\begin{proposition}[Asymptotically long first-dimer memory]
		\label{prop:asymptotic_boundary_memory}
		Let
		\begin{align}
			H_b=H_0+\varepsilon W,\qquad \varepsilon=\frac{b}{2},
			\label{eq:normal_form_Hb_main}
		\end{align}
		with $H_0$ and $W$ defined in Eqs.~\eqref{eq:H0_def}--\eqref{eq:w_def}, and let
		\begin{align}
			Q_0=M_{12}^z=S_1^z+S_2^z .
			\label{eq:Q0_normal_form_main}
		\end{align}
		Fix an integer $p\ge1$. The meaning of $p$ is that we construct a corrected boundary charge whose commutator with $H_b$ has no terms of order $\varepsilon,\varepsilon^2,\ldots,\varepsilon^p$. Assume that the finite-cluster equations needed to construct this corrected charge are solvable, or equivalently that no resonant energy differences of $H_0$ are encountered up to order $p$ (see SM~\ref*{app:boundary_normal_form}). Then there exist local operators $Q_1,\ldots,Q_p$, with $Q_n$ supported on the first $n+1$ dimers, and a constant $C_p>0$ independent of the total chain length such that, for $|\varepsilon|\le 1$ (i.e., $|b|\le 2$),
		\begin{align}
			Q^{(p)}=Q_0+\varepsilon Q_1+\cdots+\varepsilon^p Q_p
			\label{eq:dressed_charge_main}
		\end{align}
		satisfies
		\begin{align}
			\bigl\|Q^{(p)}-Q_0\bigr\|
			\le C_p|\varepsilon|,
			\label{eq:dressed_charge_close_main}
		\end{align}
		and
		\begin{align}
			\bigl\|[H_b,Q^{(p)}]\bigr\|
			\le C_p|\varepsilon|^{p+1}.
			\label{eq:dressed_charge_comm_main}
		\end{align}
		Consequently, for every normalized initial state $|\psi\rangle$,
		\begin{align}
			\left|
			\langle Q_0(t)\rangle_\psi-\langle Q_0(0)\rangle_\psi
			\right|
			\le
			2C_p|\varepsilon|+C_p|\varepsilon|^{p+1}|t|.
			\label{eq:asymptotic_memory_bound_main}
		\end{align}
		In particular, for every $\alpha<p+1$,
		\begin{align}
			\lim_{b\to0}
			\sup_{0\le t\le b^{-\alpha}}
			\left|
			\langle Q_0(t)\rangle_\psi-\langle Q_0(0)\rangle_\psi
			\right|
			=0 .
			\label{eq:asymptotic_memory_limit_main}
		\end{align}
	\end{proposition}
	For the SP$\to$MB quench, the chosen topological edge sector gives $|\langle Q_0(0)\rangle|=q_*>0$, up to corrections controlled by the initial boundary-mode localization. Proposition~\ref{prop:asymptotic_boundary_memory} then implies
	\begin{align}
		|\langle Q_0(t)\rangle|
		\ge
		q_*-2C_p|\varepsilon|-C_p|\varepsilon|^{p+1}|t| .
		\label{eq:asymptotic_memory_lower_bound_main}
	\end{align}
	Thus, for any fixed $J_z^f$ away from the finite-order resonance set, the first-dimer memory remains nonzero throughout any time window $0\le t\le b^{-\alpha}$ with $\alpha<p+1$, once the final trivial Hamiltonian is sufficiently strongly dimerized. Since $p$ may be fixed arbitrarily large before taking $b\to0$, this gives a controlled sense in which the boundary-dimer memory can be made arbitrarily long. The same state-independent estimate applies to MB$\to$MB quenches, with the interacting preparation changing only the initial value $q_*$. The isotropic point $J_z^f=1$ is excluded from this non-resonant statement, consistently with the leakage channel identified in Sec.~\ref{sec:isotropic_resonance}.
	
	\subsubsection{Interaction-dependent transition amplitudes}
	\label{sec:transition_amplitude_main}
	
	The preceding bounds control the small-$b$ regime but do not by themselves determine whether interactions enhance or suppress the edge signal at fixed $b$. To identify where $J_z^f$ enters, it is useful to diagonalize the isolated final dimer. For
	\begin{align}
		h_{12}
		=&\,\frac{a}{2}(\sigma_1^x\sigma_2^x+\sigma_1^y\sigma_2^y)
		+\frac{a}{2}J_z^f\sigma_1^z\sigma_2^z ,
		\label{eq:h12_main}
	\end{align}
	the polarized states satisfy
	\begin{align}
		E_{\uparrow\uparrow}=E_{\downarrow\downarrow}=\frac{a}{2}J_z^f,
		\label{eq:polarized_energies_main}
	\end{align}
	while the remaining two eigenstates are
	\begin{align}
		|+\rangle
		= \frac{|\uparrow\downarrow\rangle+|\downarrow\uparrow\rangle}{\sqrt{2}},
		\qquad
		E_+=-\frac{a}{2}J_z^f+a,
		\label{eq:plus_state_main}\\
		|-\rangle
		= \frac{|\uparrow\downarrow\rangle-|\downarrow\uparrow\rangle}{\sqrt{2}},
		\qquad
		E_-=-\frac{a}{2}J_z^f-a.
		\label{eq:minus_state_main}
	\end{align}
	Thus the interaction already shifts the local many-body level spacings at $b=0$.
	
	Restoring the weak inter-dimer coupling $\frac{b}{2}W$, first-order time-dependent perturbation theory gives, for distinct eigenstates $H_0|\alpha\rangle=E_\alpha|\alpha\rangle$ and $H_0|\beta\rangle=E_\beta|\beta\rangle$,
	\begin{align}
		P_{\beta\leftarrow\alpha}(t)
		=
		b^2|W_{\beta\alpha}|^2
		\frac{\sin^2[(E_\beta-E_\alpha)t/2]}{(E_\beta-E_\alpha)^2},
		\label{eq:transition_prob_main}
	\end{align}
	with $W_{\beta\alpha}=\langle\beta|W|\alpha\rangle$. A simple proof is provided in SM~\ref*{app:perturbationproof}. Since the isolated-dimer eigenvectors are independent of $J_z^f$, the product eigenbasis of $H_0$ is also independent of $J_z^f$. Writing
	\begin{align}
		W=W^{xy}+J_z^f W^{zz},
		\qquad
		H_0=a\bigl(H_0^{xy}+J_z^f H_0^{zz}\bigr),
		\nonumber
	\end{align}
	one may therefore parameterize each channel by
	\begin{align}
		W_{\beta\alpha}=&X_{\beta\alpha}+J_z^f Z_{\beta\alpha},\\
		\Delta_{\beta\alpha}\equiv& E_\beta-E_\alpha
		=a\bigl(\omega_{\beta\alpha}^{xy}+J_z^f\omega_{\beta\alpha}^{z}\bigr),
		\label{eq:Jz_explicit_matrix_elements_main}
	\end{align}
	where $X_{\beta\alpha}$, $Z_{\beta\alpha}$, $\omega_{\beta\alpha}^{xy}$, and $\omega_{\beta\alpha}^{z}$ are independent of $J_z^f$. For nonresonant pairs,
	\begin{align}
		P_{\beta\leftarrow\alpha}(t)
		=
		b^2 |X_{\beta\alpha}+J_z^f Z_{\beta\alpha}|^2
		\frac{\sin^2(\Delta_{\beta\alpha} t/2)}{\Delta_{\beta\alpha}^2}.
		\label{eq:transition_prob_Jz_explicit_main}
	\end{align}
	This expression makes the perturbative mechanism explicit: $J_z^f$ modifies leakage through both the transition matrix element and the energy detuning.
	
	Several immediate consequences follow. First, for nonresonant channels,
	\begin{align}
		P_{\beta\leftarrow\alpha}(t)
		\le
		\frac{b^2 |X_{\beta\alpha}+J_z^f Z_{\beta\alpha}|^2}{\Delta_{\beta\alpha}^2},
		\label{eq:transition_prob_channel_bound_main}
	\end{align}
	so leakage is perturbatively suppressed when
	$|b|\,|X_{\beta\alpha}+J_z^f Z_{\beta\alpha}|\ll 2|\Delta_{\beta\alpha}|$.
	Second, for short times,
	\begin{align}
		P_{\beta\leftarrow\alpha}(t)
		=
		\frac{b^2}{4}|X_{\beta\alpha}+J_z^f Z_{\beta\alpha}|^2 t^2 + O(t^4).
		\label{eq:transition_prob_short_time_main}
	\end{align}
	Third, for nonresonant channels with $\Delta_{\beta\alpha}\neq0$,
	\begin{align}
		\lim_{T\to\infty}\frac{1}{T}\int_0^T dt\,P_{\beta\leftarrow\alpha}(t)
		=
		\frac{b^2 |X_{\beta\alpha}+J_z^f Z_{\beta\alpha}|^2}{2\Delta_{\beta\alpha}^2}.
		\label{eq:transition_prob_timeavg_main}
	\end{align}
	Finally, if $Z_{\beta\alpha}\neq0$ and $\omega_{\beta\alpha}^{z}\neq0$, then
	\begin{align}
		\lim_{|J_z^f|\to\infty}
		\left[
		\lim_{T\to\infty}\frac{1}{T}\int_0^T dt\,P_{\beta\leftarrow\alpha}(t)
		\right]
		=
		\frac{b^2 |Z_{\beta\alpha}|^2}{2a^2(\omega_{\beta\alpha}^{z})^2}.
		\label{eq:transition_prob_largeJz_main}
	\end{align}
	These channel-resolved relations provide explicit $O(b^2)$ checks on the perturbative picture, but they do not imply a universal monotonic dependence on $J_z^f$, since both numerator and denominator are channel dependent.

	\subsubsection{Isotropic leakage resonance at $J_z^f=1$}
	\label{sec:isotropic_resonance}
	The same dimer spectrum explains why the finite-window memory in Fig.~\ref{fig:sptomb_late_time_average} can show a suppression near $J_z^f=1$. At this value the local dimer Hamiltonian in Eq.~\eqref{eq:h12_main} is isotropic and has an $SU(2)$-symmetric triplet. In the notation of Eqs.~\eqref{eq:polarized_energies_main}--\eqref{eq:minus_state_main}, the triplet states are $|t_+\rangle=|\uparrow\uparrow\rangle$, $|t_0\rangle=|+\rangle$, and $|t_-\rangle=|\downarrow\downarrow\rangle$, and their energies satisfy
	\begin{align}
		E_{t_+}=E_{t_0}=E_{t_-}
		\qquad
		\text{at}
		\qquad
		J_z^f=1 .
		\label{eq:triplet_degeneracy_main}
	\end{align}
	For two neighboring dimers, the weak inter-dimer $XY$ term connects the leakage channel
	\begin{align}
		|t_+\rangle_{12}|t_-\rangle_{34}
		\longleftrightarrow
		|t_0\rangle_{12}|t_0\rangle_{34},
		\label{eq:isotropic_channel_main}
	\end{align}
	with nonzero matrix element and detuning
	\begin{align}
		\Delta_{\rm iso}
		=
		2E_{t_0}-E_{t_+}-E_{t_-}
		=
		2a(1-J_z^f) \xrightarrow[]{J_z^f \to 1} 0.
		\label{eq:isotropic_detuning_main}
	\end{align}
	The explicit spin-basis derivation of Eqs.~\eqref{eq:isotropic_channel_main} and~\eqref{eq:isotropic_detuning_main} is given in SM~\ref*{app:isotropic_channel_derivation}. Consequently this leakage channel is resonant at the isotropic point. Taking the resonant limit of Eq.~\eqref{eq:transition_prob_main} gives $P_{\beta\leftarrow\alpha}^{(1)}(t)=(b^2/4)|W_{\beta\alpha}|^2t^2$, rather than the bounded oscillatory form obtained for a fixed nonzero detuning. The dip near $J_z^f=1$ should therefore be interpreted as a symmetry-enabled leakage resonance, not as an additional protection point. By contrast, in the small-$J_z^f$ perturbative regime, Eqs.~\eqref{eq:Jz_explicit_matrix_elements_main}--\eqref{eq:transition_prob_short_time_main} show that the correction has no universal sign because both the matrix elements and detunings are channel dependent. This is why very small $J_z^f$ can either increase or decrease the finite-window memory depending on the state, size, and averaging window, whereas away from the perturbative small-$J_z^f$ regime, the TEBD data show enhanced memory over broad interacting ranges with a local suppression at the isotropic $SU(2)$ point at $J_z^f=1$.
	
	\begin{figure*}[t]
		\centering
		\includegraphics[width=0.95\textwidth]{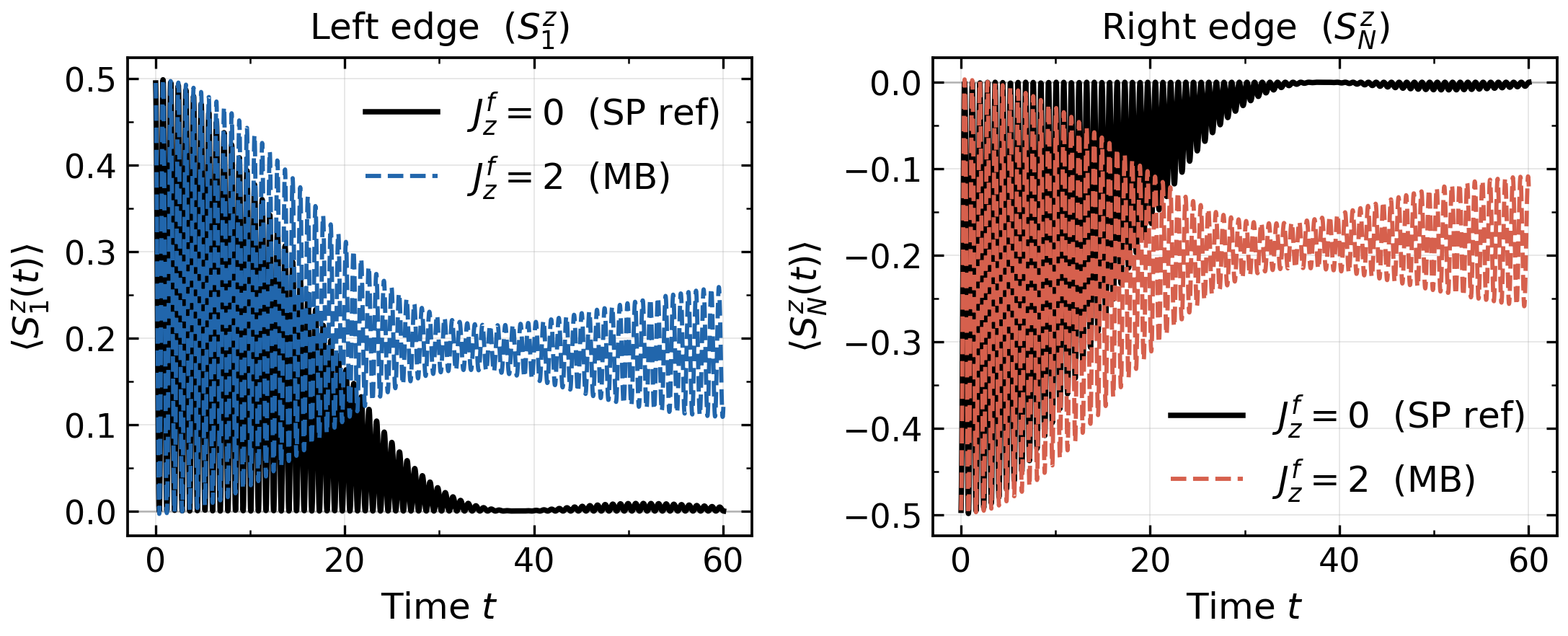}
		\caption{SP to MB quench at fixed dimerization, $\delta_i=-0.95$ and $\delta_f=0.95$, for $N=36$ sites with open boundary conditions in the fixed $S_{\rm tot}^z=0$ sector. The initial state is the non-interacting topological ground state, and the time evolution is generated by the interacting final Hamiltonian at $J_z^f=2$. The left and right panels show $S_1^z(t)$ and $S_N^z(t)$, respectively. The curve with $J_z^f=0$ is the SP reference at the same dimerization; the interacting curve shows enhanced finite-time retention of the edge magnetization relative to the SP reference.}
		\label{fig:sptomb_edge}
	\end{figure*}
	
	\begin{figure}[t]
		\centering
		\includegraphics[width=0.95\columnwidth]{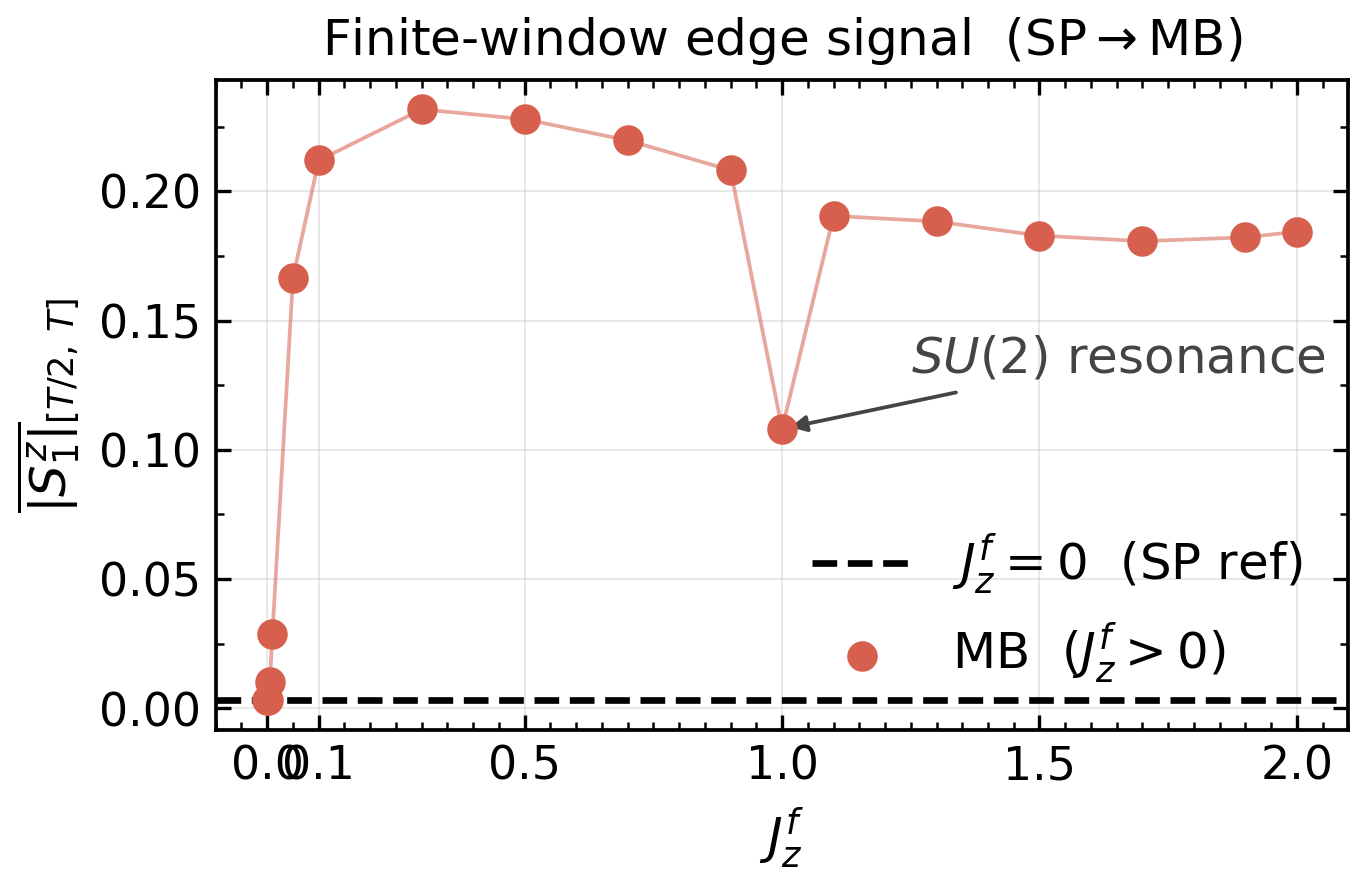}
		\caption{Finite-window edge signal for the SP to MB quench at $\delta_i=-0.95$, $\delta_f=0.95$, $N=36$ sites, with open boundary conditions in the fixed $S_{\rm tot}^z=0$ sector. The plotted quantity is $\overline{|S_1^z|}_{[T_1,T_2]}$, defined in Eq.~\eqref{eq:late_time_average_def}, as a function of the final interaction strength $J_z^f$. The $J_z^f=0$ point is the SP reference at the same $\delta_i$ and $\delta_f$. Interacting final Hamiltonians generically enhance the finite-window edge signal above the SP reference, while the local suppression near $J_z^f=1$ reflects the resonant leakage	channel discussed in Sec.~\ref{sec:isotropic_resonance}.}
		\label{fig:sptomb_late_time_average}
	\end{figure}
	
	\subsubsection{TEBD memory dynamics}
	\label{sec:sptomb_tebd}
	
	We now compare the analytic results above with TEBD simulations of the SP$\to$MB quench. The primary observable is the edge magnetization
	\begin{align}
		S_1^z(t)=\langle\Psi(t)|S_1^z|\Psi(t)\rangle,
		\label{eq:S1_obs_sptomb}
	\end{align}
	where $|\Psi(0)\rangle$ is the non-interacting topological ground state and
	$|\Psi(t)\rangle=e^{-iH_bt}|\Psi(0)\rangle$.
	The non-interacting evolution at the same $\delta_i$ and $\delta_f$ defines the SP reference curve,
	\begin{align}
		S_{1,\mathrm{SP}}^z(t) = S_1^z(t;\,J_z^f=0),
		\label{eq:SP_reference_def}
	\end{align}
	and interacting curves correspond to $J_z^f\neq0$. To reduce the comparison to a single number over a finite window we use
	\begin{align}
		\overline{|S_1^z|}_{[T_1,T_2]}
		=
		\frac{1}{T_2-T_1}
		\int_{T_1}^{T_2} dt\,|S_1^z(t)|,
		\label{eq:late_time_average_def}
	\end{align}
	with $[T_1,T_2]=[T_{\rm max}/2,T_{\rm max}]$ throughout.
	
	Figure~\ref{fig:sptomb_edge} shows $S_1^z(t)$ at $\delta_i=-0.95$, $\delta_f=0.95$ for $J_z^f=2$ alongside the SP reference. Away from the isotropic leakage resonance, the interacting final Hamiltonian yields visibly stronger retention of the edge signal throughout the simulated window.
	
	Figure~\ref{fig:sptomb_late_time_average} summarizes the $J_z^f$ dependence via Eq.~\eqref{eq:late_time_average_def}. Interacting final Hamiltonians generically raise the finite-window average above the SP reference, but the dependence on $J_z^f$ is non-monotonic. The local suppression near $J_z^f=1$ is the isotropic leakage resonance of Sec.~\ref{sec:isotropic_resonance}; the memory at that point nevertheless remains above the SP reference value. Away from the resonance, Eq.~\eqref{eq:bound_main} shows that increasing $J_z^f$ grows the norm of the inter-dimer perturbation while Eq.~\eqref{eq:transition_prob_main} shows that the same interaction changes both matrix elements and detunings in a channel-dependent way, so no universal sign is expected at small $J_z^f$. Numerically exact checks of the analytic dimer-limit bounds are in SM~\ref*{app:dimer_benchmarks}.
	
	The SP$\to$MB data therefore establish a clear organizing principle: the memory enhancement is a property of the interacting final Hamiltonian, present broadly in $J_z^f$ with a single resonant exception, and is not a fine-tuned consequence of the free initial preparation. As shown in Sec.~\ref{sec:mbtomb}, this conclusion carries over without qualitative change when the initial state is also interacting.

	\subsection{Many-Body to Many-Body Quench}
	\label{sec:mbtomb}
	
	Finally, we consider quenches in which both the initial state and the final Hamiltonian are interacting. This is the most general case studied in this work and the least accessible to closed-form analysis. Accordingly, we organize the subsection around the strongly dimerized final regime, where the exact decomposition in Eq.~\eqref{eq:Hb_decomp} remains available and Propositions~\ref{prop:finite_time_stability_sptomb} and \ref{prop:boundary_dimer_continuity} can be used without modification. The goal is to identify which aspects of the SP to MB dynamics survive when the initial state is itself interacting, and which aspects depend quantitatively on the many-body correlations of the preparation.
	
	Let
	\begin{align}
		\rho_i(J_z^i)
		=
		|\Psi_i(J_z^i)\rangle\langle\Psi_i(J_z^i)|
	\end{align}
	denote the topological initial ground state of the Hamiltonian $H(\delta_i,J_z^i)$ in the fixed $S_{\rm tot}^z=0$ sector, and let
	\begin{align}
		\rho_{12}^{(i)}(J_z^i)
		=
		\Tr_{3,\dots,N}\rho_i(J_z^i)
		\label{eq:mbtomb_rho12_def}
	\end{align}
	be its reduced density matrix on the first dimer. Here and below, $\|\cdot\|_1$ denotes the trace norm,
	\begin{align}
		\|A\|_1=\Tr\sqrt{A^\dagger A}.
		\label{eq:trace_norm_def_mbtomb}
	\end{align}
	Because Propositions~\ref{prop:finite_time_stability_sptomb} and \ref{prop:boundary_dimer_continuity} were proved for arbitrary normalized initial states, they apply here verbatim once the initial state is replaced by $\rho_i(J_z^i)$.
	
	\subsubsection{Exact dimer limit and finite-time comparison}
	\label{sec:mbtomb_dimerlimit}
	
	At $b=0$, the final Hamiltonian is the decoupled-dimer Hamiltonian $H_0=\sum_{m=1}^{N/2}h_{2m-1,2m}$. Since the evolution factorizes over dimers, the edge dynamics is determined exactly by the initial reduced state on the first dimer.
	
	\begin{proposition}[Exact first-dimer reduction at $b=0$]
		\label{prop:mbtomb_b0_reduction}
		For the MB to MB quench with final Hamiltonian $H_0$, one has
		\begin{align}
			\langle S_1^z(t)\rangle_{b=0}
			=
			\Tr_{12}\!\left[
			\rho_{12}^{(i)}(J_z^i)\,
			e^{ih_{12}t}S_1^z e^{-ih_{12}t}
			\right].
			\label{eq:mbtomb_b0_reduction}
		\end{align}
		In the eigenbasis
		$\{|\uparrow\uparrow\rangle, |\downarrow\downarrow\rangle, |+\rangle, |-\rangle\}$
		of $h_{12}$ defined in Eqs.~\eqref{eq:polarized_energies_main}--\eqref{eq:minus_state_main}, this becomes
		\begin{align}
			\langle S_1^z(t)\rangle_{b=0}
			=
			\frac{1}{2}\bigl(
			p_{\uparrow\uparrow}-p_{\downarrow\downarrow}
			\bigr)
			+
			\Re\!\left[
			e^{-i2at}\rho_{+-}^{(i)}
			\right],
			\label{eq:mbtomb_b0_explicit}
		\end{align}
		where
		\begin{align}
			p_{\alpha}
			=
			\langle \alpha|
			\rho_{12}^{(i)}(J_z^i)
			|\alpha\rangle,
			\qquad
			\rho_{+-}^{(i)}
			=
			\langle +|
			\rho_{12}^{(i)}(J_z^i)
			|-\rangle.
			\label{eq:mbtomb_b0_popcoh}
		\end{align}
	\end{proposition}
	
	Thus, in the exactly dimerized final limit, the full dependence on the interacting initial state enters only through the first-dimer reduced density matrix. The oscillation frequency is fixed by the final intra-dimer scale, since $E_+-E_-=2a$ is independent of both $J_z^i$ and $J_z^f$.
	
	A second exact statement concerns the first-dimer magnetization.
	
	\begin{corollary}[Exact conservation of $M_{12}^z$ at $b=0$]
		\label{cor:mbtomb_M12_b0}
		For the decoupled-dimer final Hamiltonian $H_0$, one has
		\begin{align}
			\langle M_{12}^z(t)\rangle_{b=0}
			=
			\langle M_{12}^z(0)\rangle
			=
			\Tr_{12}\!\left[
			\rho_{12}^{(i)}(J_z^i) M_{12}^z
			\right].
			\label{eq:mbtomb_M12_b0}
		\end{align}
	\end{corollary}
	
	This follows immediately from $[h_{12},M_{12}^z]=0$. The first dimer is therefore an exact closed subsystem at $b=0$ for both observables $S_1^z$ and $M_{12}^z$.
	
	For $0<b\ll1$, the previously established operator bounds give the corresponding finite-time comparison with the exact dimer limit,
	\begin{align}
		\left|
		\langle S_1^z(t)\rangle_b
		-
		\langle S_1^z(t)\rangle_{b=0}
		\right|
		\le
		\frac{|b|}{2}\,(2+|J_z^f|)\,t,
		\label{eq:mbtomb_S1_b0_comparison}
	\end{align}
	and
	\begin{align}
		\left|
		\langle M_{12}^z(t)\rangle
		-
		\langle M_{12}^z(0)\rangle
		\right|
		\le
		|b|\,t.
		\label{eq:mbtomb_M12_leakage_bound}
	\end{align}
	Equation~\eqref{eq:mbtomb_S1_b0_comparison} controls the deviation of the physical edge observable from the exact two-site limit, while Eq.~\eqref{eq:mbtomb_M12_leakage_bound} isolates the leakage of total $z$ magnetization from the first dimer through the weak bond.
	
	The same structure also yields a comparison bound that separates the effect of the initial preparation, which appears as a time-independent offset, from the final-Hamiltonian contribution, which grows linearly in time.
	
	\begin{proposition}[Initial-state comparison near the dimer limit]
		\label{prop:mbtomb_initial_state_comparison}
		Let $\rho_i$ and $\sigma_i$ be two initial states for two different quenches, and let both evolve under the same final Hamiltonian $H_b$. Denote their first-dimer reduced density matrices by
		\begin{align}
			\rho_{12}^{(i)}=\Tr_{3,\dots,N}\rho_i,
			\qquad
			\sigma_{12}^{(i)}=\Tr_{3,\dots,N}\sigma_i.
			\nonumber
		\end{align}
		Then, for all $t\ge0$,
		\begin{equation}
			\begin{aligned}
				\left|
				\langle S_1^z(t)\rangle_{\rho_i,b}
				-\right.&\left.
				\langle S_1^z(t)\rangle_{\sigma_i,b}
				\right| \\
				&\le 
				\frac{1}{2}
				\left\|
				\rho_{12}^{(i)}-\sigma_{12}^{(i)}
				\right\|_1
				+
				|b|\,(2+|J_z^f|)\,t.
			\end{aligned}
			\label{eq:mbtomb_initial_state_bound}
		\end{equation}
	\end{proposition}
	
	The proof is given in SM~\ref*{app:initialstatecomparison}. The bound separates the effect of the initial preparation, which appears as the time-independent trace-norm term, from the final-Hamiltonian contribution, which grows linearly in $t$ through inter-dimer leakage.
	
	All of these analytical bounds, including the initial-state comparison bound, are verified by exact diagonalization in SM~\ref*{app:dimer_benchmarks}.
	
	\begin{figure*}[t]
		\centering
		\includegraphics[width=0.95\textwidth]{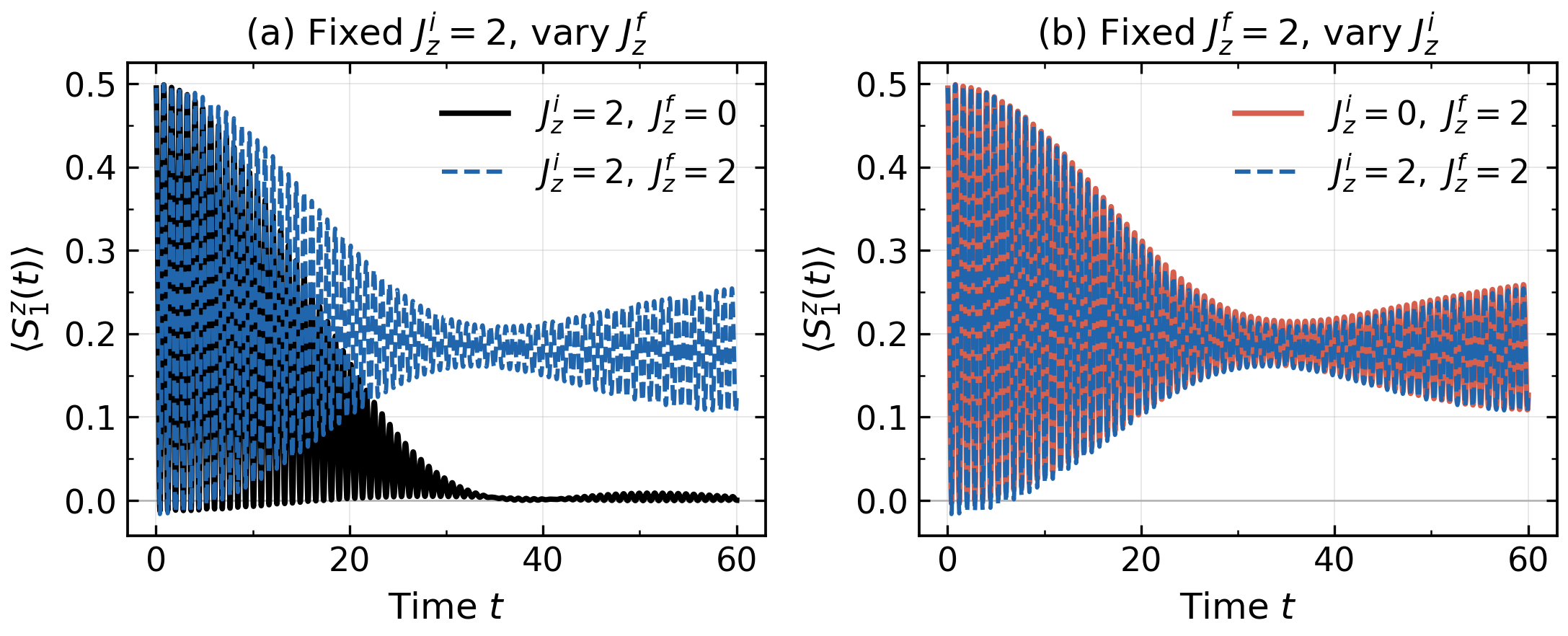}
		\caption{MB$\to$MB left-edge memory dynamics at fixed $\delta_i=-0.95$ and $\delta_f=0.95$, for $N=36$ with open boundary conditions in the fixed $S_{\rm tot}^z=0$ sector. Panel~(a) compares $\{J_z^i,J_z^f\}=\{2,0\}$ and $\{2,2\}$ at fixed interacting preparation, isolating the role of post-quench interactions; the interacting-final curve at $J_z^f=2$ shows enhanced finite-time retention relative to the non-interacting-final reference, in direct analogy with the SP$\to$MB protocol. Panel~(b) fixes $J_z^f=2$ and compares $J_z^i=0$ (SP preparation, SP$\to$MB) and $J_z^i=2$ (MB preparation, MB$\to$MB); the two  curves are qualitatively identical---both exhibiting the same interaction-enhanced memory---with only a quantitative shift from the different initial-state correlations. The SP$\to$MB and MB$\to$MB protocols therefore belong to the same universality class of post-quench behavior, distinguished only by the initial one-body density matrix.}
		\label{fig:mbtomb_memory}
	\end{figure*}

	\begin{figure}[htbp]
		\centering
		\includegraphics[width=0.95\columnwidth]{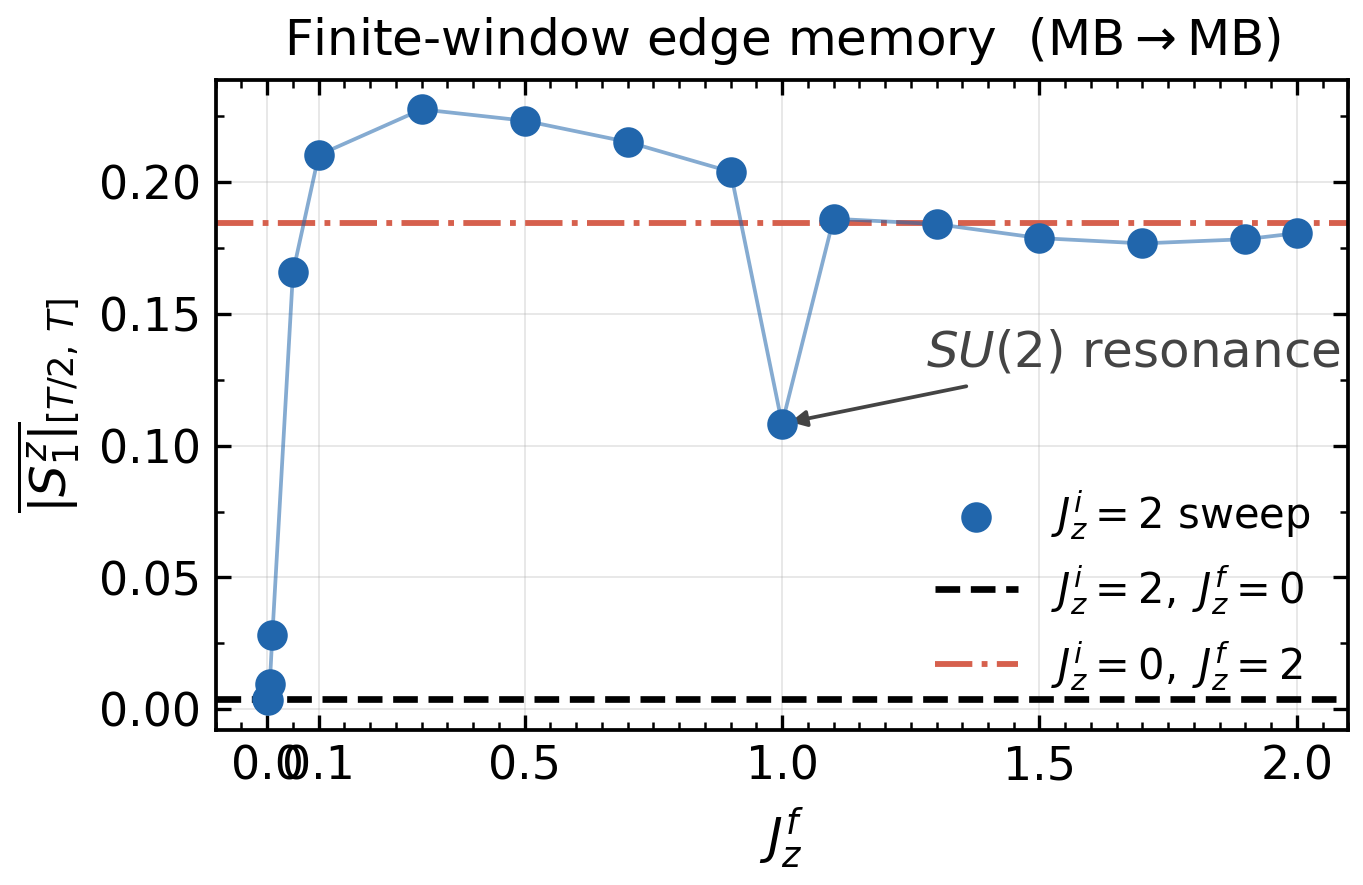}
		\caption{Finite-window edge-memory diagnostic $\overline{|S_1^z|}_{[T_1,T_2]}$ (Eq.~\eqref{eq:late_time_average_def}) for MB$\to$MB quenches with $\delta_i=-0.95$, $\delta_f=0.95$, $N=36$, as a function of $J_z^f$ at fixed $J_z^i=2$ (blue markers and line). The non-interacting-final reference $(J_z^i,J_z^f)=(2,0)$ (black dashed) and the SP$\to$MB reference $(J_z^i,J_z^f)=(0,2)$ (red dash-dot) are shown for comparison. The MB$\to$MB sweep reproduces the same qualitative structure as the SP$\to$MB protocol: interacting final Hamiltonians generically raise the finite-window edge signal well above the SP reference, a local suppression appears near the isotropic leakage resonance $J_z^f=1$ but the memory remains above the non-interacting-final reference even there, and the SP$\to$MB and MB$\to$MB values at $J_z^f=2$ agree within the quantitative shift due to the different initial-state correlations. The two protocols are thus qualitatively equivalent, confirming that the memory enhancement is generated by the post-quench interactions and is insensitive to whether the initial preparation is interacting or not.}
		\label{fig:mbtomb_memory_window}
	\end{figure}

	\subsubsection{TEBD memory dynamics}
	\label{sec:mbtomb_numerics}
	
	The analytic results above identify the strongly dimerized final regime as a controlled limit in which the first-dimer magnetization is stable on arbitrarily long time windows away from local resonances. The remaining question is whether the enhancement of boundary memory by an interacting final Hamiltonian---found in the SP$\to$MB protocol with a local dip at the isotropic point---persists when the initial state is itself interacting. We fix $\delta_i=-0.95$, $\delta_f=0.95$, and disentangle two effects: (i)~at fixed $J_z^i=2$, we compare $J_z^f=0$ and $J_z^f=2.0$ to isolate the role of the final Hamiltonian; (ii)~at fixed $J_z^f=2.0$, we compare $J_z^i=0$ and $J_z^i=2$ to test sensitivity to the interacting preparation. The value $J_z^f=2.0$ lies away from the isotropic leakage resonance. The primary observable is
	\begin{align}
		S_1^z(t) = \langle\Psi(t)|S_1^z|\Psi(t)\rangle,
		\label{eq:mbtomb_S1_obs}
	\end{align}
	where $|\Psi(0)\rangle$ is the topological ground state in the fixed $S_{\rm tot}^z=0$ sector and $|\Psi(t)\rangle=e^{-iH_bt}|\Psi(0)\rangle$. The first-dimer magnetization $M_{12}^z(t)$ is examined in SM~\ref*{app:dimer_benchmarks} as a numerical check of the continuity equation and dimer-limit bounds.
	
	Figure~\ref{fig:mbtomb_memory} shows the results. Panel~(a) holds $J_z^i=2$ fixed and varies $J_z^f$: the non-interacting-final curve ($J_z^f=0$) loses boundary memory on the same timescale as the MB$\to$SP protocol, whereas the interacting-final curve ($J_z^f=2.0$) retains a substantially larger edge signal throughout the simulated window. Panel~(b) holds $J_z^f=2.0$ fixed and varies the preparation ($J_z^i=0,2$): the two traces are qualitatively identical---the same memory-enhanced behavior found in the SP$\to$MB protocol is reproduced in the MB$\to$MB protocol---differing only in the detailed oscillation amplitudes set by the initial-state correlations. The SP$\to$MB and MB$\to$MB protocols therefore belong to the same class of post-quench behavior, organized by the interacting final Hamiltonian rather than by the initial preparation.
	
	To reduce each protocol to a single number we use the finite-window diagnostic $\overline{|S_1^z|}_{[T_1,T_2]}$ with the same window $[T_1,T_2]=[T_{\rm max}/2,T_{\rm max}]$ used in Eq.~\eqref{eq:late_time_average_def}. Figure~\ref{fig:mbtomb_memory_window} plots this quantity for $(J_z^i,J_z^f)=(2,0),(0,2),(2,2)$ alongside the $J_z^i=2$ sweep over $J_z^f$. All interacting-final protocols remain well above the non-interacting-final reference, the local suppression at the isotropic point $J_z^f=1$ is present but still above that reference, and the variation with $J_z^i$ is comparatively modest. The memory enhancement is therefore a robust property of the interacting final Hamiltonian: $J_z^f$ controls whether memory is preserved or suppressed, while $J_z^i$ dresses the initial boundary correlations quantitatively without changing the qualitative outcome.

	\section{Conclusions}
	\label{sec:conclusion}

	We have investigated the fate of topological edge memory under strong quenches in the bond-alternating XXZ chain, which realizes an interacting SSH-type model with symmetry-protected boundary modes. Organizing the problem according to the single-particle (SP) versus many-body (MB) character of the initial and final Hamiltonians provides a natural and exhaustive classification of the quench space, and each of the four resulting protocols admits a distinct and controlled analytical treatment. In the SP reference problem, the quench dynamics is solved exactly via a correlation-matrix approach, and the boundary-mode return amplitude is proved to decay asymptotically as $t^{-3/2}$, with analytically determined, dimerization-dependent envelope and beat scales that fix the intrinsic relaxation pattern of the free boundary dynamics. In the MB$\to$SP quench, the evolution is governed exactly by the same free post-quench framework with a dressed one-body density matrix, and uniform-in-time operator-norm bounds control the deviation from the SP reference. For the genuinely interacting post-quench protocols, decomposing the final Hamiltonian into decoupled dimers plus a weak inter-dimer perturbation yields finite-time stability bounds and an exact continuity equation for the leakage channel, and away from local resonances it also yields a rigorous asymptotic statement that the first-dimer magnetization remains stable on time windows that grow as arbitrarily large powers of the inverse weak inter-dimer coupling.

	TEBD simulations across all four protocols demonstrate that post-quench interactions can substantially enhance finite-time edge-memory retention relative to the non-interacting reference, with a non-monotonic dependence on the interaction strength. This non-monotonicity has two distinct origins: the small-$J_z^f$ perturbative correction is channel dependent and therefore need not have a fixed sign, while the isotropic point $J_z^f=1$ produces a symmetry-enabled leakage resonance that suppresses the finite-window memory. This constitutes a qualitatively new mechanism by which many-body interactions can protect boundary memory in clean, disorder-free one-dimensional systems---complementary to many-body localization, quantum many-body scars, and dynamical symmetry-protected phases. The key insight is that the strongly dimerized structure of the final Hamiltonian, combined with the topological structure of the initial state, produces a boundary degree of freedom whose relaxation is parametrically delayed relative to bulk modes, even though the post-quench system is non-integrable and hosts no equilibrium topological phase.

	The mechanisms identified here are within reach of several experimental platforms. The symmetry-protected topological phase of interacting bosons was observed in a Rydberg tweezer array~\cite{deleseleuc2019science}, and the Haldane phase with directly resolved edge states was realized in Fermi--Hubbard ladders using single-site-resolved measurements~\cite{sompet2022nature}. Most directly relevant, Katz \textit{et~al.}~\cite{Katz2025Oct} implemented a spin-based SSH model in a trapped-ion chain of up to 22 interacting spins using Floquet engineering of bond dimerization, observing edge-state dynamics and studying spin propagation across different interaction regimes---a setting that maps almost directly onto the SP$\to$MB and MB$\to$MB quench protocols studied here. Gate-tunable SSH chains on superconducting resonator lattices~\cite{splitthoff2024prresearch} offer independent in-situ control of intra- and intercell couplings and provide a complementary route to emulating the strongly dimerized post-quench limit. In all of these settings, the quench protocols of this work translate into experimentally natural operations---preparing a symmetry-protected topological ground state, suddenly switching the dimerization and interaction parameters, and tracking the decay of a boundary spin or edge population---and the analytical results here delineate precisely the parameter regimes where final-state interactions are expected to enhance rather than degrade boundary memory, as well as the isotropic point where resonant leakage is expected that causes local suppression of memory, while remaining above the SP reference point. 
	
	Several theoretical directions follow naturally from this work. Extending the present framework to open-system dynamics would address whether the many-body protection mechanism survives realistic decoherence, and whether engineered dissipation or quantum Zeno effects can further stabilize boundary memory. Combining the quench protocols studied here with Floquet driving, in the spirit of recent experiments realizing dynamical symmetry-protected topological phases, would allow one to ask how emergent dynamical symmetries and the interaction-assisted protection identified here cooperate or compete. It would further be valuable to classify the universality of the $t^{-3/2}$ boundary-mode decay and its dimerization scaling across broader families of one-dimensional topological chains, including higher-spin and longer-range variants, and to develop effective boundary descriptions that capture the competition between dimer stabilization and XY-mediated leakage in strongly correlated post-quench dynamics. We expect that the combination of rigorous analytical control, TEBD benchmarking, and direct experimental contact established in this work will stimulate further investigation into interacting topological edge memories as building blocks for robust non-equilibrium quantum information protocols.
	
	\begin{acknowledgments}
		The authors thank the Center for Advanced Research Computing (CARC) at the University of Southern California (https://carc.usc.edu) for providing the computational resources used in this work. RJ and SH acknowledge partial and full support, respectively, from the U.S. Department of Energy, Office of Advanced Scientific Computing Research (DOE-ASCR), under Award No. DE-SC0026337.	\end{acknowledgments}

	\bibliography{refs_tq}

\clearpage
\appendix

\makeatletter
\@removefromreset{equation}{section}
\makeatother

\setcounter{page}{1}

\setcounter{section}{0}
\setcounter{subsection}{0}
\setcounter{equation}{0}
\setcounter{figure}{0}
\setcounter{table}{0}

\renewcommand{\thesection}{S\arabic{section}}
\renewcommand{\thesubsection}{S\arabic{section}.\arabic{subsection}}
\renewcommand{\theequation}{S\arabic{equation}}
\renewcommand{\thefigure}{S\arabic{figure}}
\renewcommand{\thetable}{S\arabic{table}}

\clearpage
\thispagestyle{plain}
\setcounter{page}{1}

\onecolumngrid 

\begin{center}
	{\Large Supplemental Material --- Many-Body Protection of Topological Edge Memory in Strong Interacting Quenches}\\[1.2em]
	
	{\normalsize
		Yuxiao Hang,$^{a,}$\footnote{\href{mailto:yhang@usc.edu}{yhang@usc.edu}} Stephan Haas,$^{a,}$\footnote{\href{mailto:shaas@usc.edu}{shaas@usc.edu}} and Rishabh Jha$^{a,}$\footnote{\href{mailto:rishabh.jha@usc.edu}{rishabh.jha@usc.edu}} \\[0.5em]
		$^{a}$Department of Physics and Astronomy, University of Southern California, Los Angeles, CA 90089-0484, USA\\[0.5em]
	}
	
\end{center}

\twocolumngrid 

\vspace{1em}

In this Supplemental Material, all numbered items carry an S prefix, whereas references without an S belong to the main text.


\twocolumngrid

\section{Single-Particle SSH Chain: Conventions and Boundary Zero Mode}
\label{app:ssh_zero_mode}

	This appendix derives from first principles the single-particle structure of the SSH chain and the explicit form of the topological boundary zero mode. The results here are used in Appendix~\ref{app:theorem1_proof} and Sec.~\ref*{sec:single}.
	\subsection{Single-Particle Hamiltonian and Conventions}
	\label{app:ssh_conventions}
	In the non-interacting limit $J_z = 0$, the bond-alternating $XX$ spin chain maps via a Jordan--Wigner transformation onto a free spinless-fermion chain. The single-particle Hamiltonian in the site basis is
	\begin{equation}
		h = \sum_{j=1}^{N-1} t_j \left( c_j^\dagger c_{j+1} + c_{j+1}^\dagger c_j \right),
		\label{eq:sp_ham}
	\end{equation}
	where $t_j$ alternates between two values according to the dimerization. It is convenient to label sites by their unit-cell index $m$ and sublattice index. For a finite open chain of even length $N$, one has $m = 1, 2, \ldots, N/2$; for the semi-infinite boundary-mode derivation used below, one takes $m = 1, 2, \ldots$. We define
	\begin{equation}
		A_m \equiv \text{odd site } 2m - 1, \qquad B_m \equiv \text{even site } 2m,
		\label{eq:sublattice}
	\end{equation}
	so that each unit cell $m$ contains one $A$-site and one $B$-site. The single-particle state $|A_m\rangle \equiv c_{A,m}^\dagger |0\rangle$ means that exactly one spinless fermion occupies the $A$-sublattice site of unit cell $m$; it is a one-particle basis vector, not a many-body state. In this notation, and in the semi-infinite geometry relevant for the zero-mode derivation, the single-particle Hamiltonian takes the compact form
	\begin{equation}
		h = \sum_{m \geq 1} \left[ J_1 |A_m\rangle\langle B_m| + J_2 |A_{m+1}\rangle\langle B_m| + \mathrm{h.c.} \right],
		\label{eq:sp_ham_unitcell}
	\end{equation}
	where $J_1$ is the intracell SSH one-body hopping amplitude and $J_2$ is the intercell SSH one-body hopping amplitude. With the Pauli-matrix convention of Eq.~(\ref*{eq:Ham}) and the dimerization parametrization introduced in Sec.~\ref*{sec:model},
	\begin{equation}
		J_1 = 2J_1^{xy}=2t(1 + \delta), \qquad J_2 = 2J_2^{xy}=2t(1 - \delta).
		\label{eq:J_param}
	\end{equation}
	The condition for the topological phase is $|J_1| < |J_2|$, that is,
	\begin{equation}
		\delta < 0 \quad \Longleftrightarrow \quad J_1 < J_2.
		\label{eq:topo_condition}
	\end{equation}
	This can be understood intuitively: the leftmost site is $A_1$, and its only direct coupling to the rest of the chain is the intracell bond $J_1$ to $B_1$. When $|J_1| < |J_2|$, this boundary link is the weak bond, while the bulk sites pair predominantly through the stronger intercell bonds $J_2$. The boundary site is therefore only weakly attached to the dimerized bulk, which is the mechanism that localizes a zero-energy mode at the left edge.
	
	\subsection{Derivation of the Boundary Zero Mode}
	\label{app:zero_mode}
	Consider a general single-particle state
	\begin{equation}
		|\psi\rangle = \sum_{m \geq 1} \left[ \alpha_m |A_m\rangle + \beta_m |B_m\rangle \right].
		\label{eq:general_state}
	\end{equation}
	Applying $h|\psi\rangle = E|\psi\rangle$ and projecting onto $\langle A_m|$ and $\langle B_m|$ gives the coupled recursion relations
	\begin{align}
		E \alpha_m &= J_1 \beta_m + J_2 \beta_{m-1}, \label{eq:recA} \\
		E \beta_m &= J_1 \alpha_m + J_2 \alpha_{m+1}. \label{eq:recB}
	\end{align}
	The open boundary condition at the left end of the chain requires that there is no site to the left of $m = 1$, which imposes
	\begin{equation}
		\beta_0 = 0.
		\label{eq:OBC}
	\end{equation}
	We now seek a \emph{zero-energy} edge mode, so we set $E = 0$. With $E = 0$, the recursion~(\ref{eq:recA}) becomes
	\begin{equation}
		0 = J_1 \beta_m + J_2 \beta_{m-1}.
		\label{eq:recA_zero}
	\end{equation}
	We make the ansatz that the zero mode lives entirely on the $A$-sublattice, that is,
	\begin{equation}
		\beta_m = 0 \quad \text{for all } m \geq 1.
		\label{eq:A_sublattice}
	\end{equation}
	With this ansatz, Eq.~(\ref{eq:recA_zero}) is satisfied automatically since both $\beta_m$ and $\beta_{m-1}$ vanish. The open boundary condition~(\ref{eq:OBC}) is also satisfied trivially. The remaining equation~(\ref{eq:recB}) with $E = 0$ becomes
	\begin{equation}
		0 = J_1 \alpha_m + J_2 \alpha_{m+1},
		\label{eq:recB_zero}
	\end{equation}
	which gives the recursion
	\begin{equation}
		\alpha_{m+1} = -\frac{J_1}{J_2} \alpha_m.
		\label{eq:recursion}
	\end{equation}
	Define
	\begin{equation}
		r \equiv -\frac{J_1}{J_2}.
		\label{eq:r_def}
	\end{equation}
	Then, iterating Eq.~(\ref{eq:recursion}) from $m = 1$,
	\begin{equation}
		\alpha_m = \alpha_1 \cdot r^{m-1}.
		\label{eq:alpha_m}
	\end{equation}
	The unnormalized zero mode is obtained by substituting the closed-form solution $\alpha_m=\alpha_1$. $r^{m-1}$ and the ansatz $\beta_m=0$ back into the general state~\eqref{eq:general_state}. Setting the free constant $\alpha_1=1$, we define $\left|\phi_L\right\rangle$ as the resulting left edge mode, so that the general state reduces to:
	\begin{equation}
		|\phi_L\rangle_{\rm unnorm} = \sum_{m \geq 1} r^{m-1} |A_m\rangle.
		\label{eq:zero_mode_unnorm}
	\end{equation}
	This state is normalizable if and only if the sum $\sum_{m \geq 1} |r|^{2(m-1)}$ converges, which requires $|r| < 1$, that is,
	\begin{equation}
		|r| < 1 \quad \Longleftrightarrow \quad |J_1| < |J_2|,
		\label{eq:normalizability}
	\end{equation}
	which is exactly the topological phase condition~(\ref{eq:topo_condition}). Outside the topological phase, $|r| \geq 1$ and the amplitude grows (or stays constant) with $m$, so no normalizable edge mode exists. This is the mathematical origin of the bulk-boundary correspondence for the SSH chain: a normalizable zero mode exists if and only if the chain is in the topological phase.
	The normalization is fixed by
	\begin{equation}
		1 = \langle \phi_L | \phi_L \rangle = |\mathcal{N}|^2 \sum_{m \geq 1} |r|^{2(m-1)} = \frac{|\mathcal{N}|^2}{1 - |r|^2},
		\label{eq:norm_sum}
	\end{equation}
	where we used the geometric series formula $\sum_{m=0}^{\infty} x^m = 1/(1-x)$ for $|x| < 1$ with $x = |r|^2$. Therefore
	\begin{equation}
		\mathcal{N} = \sqrt{1 - |r|^2}.
		\label{eq:norm}
	\end{equation}
	The normalized left boundary zero mode is
	\begin{equation}
		|\phi_L\rangle = \mathcal{N} \sum_{m \geq 1} r^{m-1} |A_m\rangle = \sqrt{1-|r|^2} \sum_{m \geq 1} r^{m-1} |A_m\rangle.
		\label{eq:zero_mode}
	\end{equation}
	\subsection{Physical Interpretation and Numerics}
	\label{app:zero_mode_numerics}
	For the parameters used in this work, $\delta_i = -0.95$, so
	\begin{equation}
		J_1^i = 2t(1 + \delta_i) = 0.10, \qquad J_2^i = 2t(1 - \delta_i) = 3.90,
		\label{eq:params_num}
	\end{equation}
	giving
	\begin{equation}
		r = -\frac{0.10}{3.90} \approx -0.02564.
		\label{eq:r_num}
	\end{equation}
	The normalization factor is $\mathcal{N}^2 = 1 - r^2 \approx 0.99934$, so the zero mode is almost entirely localized on site $1$ ($m=1$, $A$-sublattice). The amplitude on unit cell $m$ decays as $|r|^{m-1} \approx (0.026)^{m-1}$, which falls to below $1\%$ of its initial value by $m = 2$. This extreme localization is a consequence of the strong dimerization $|\delta_i| = 0.95$.
	The initial edge magnetization is related to the zero-mode occupation. Via the Jordan--Wigner transformation, $\sigma_j^z = 2n_j - 1$, so the edge magnetization at site $1$ is
	\begin{equation}
		S_1^z(0) = \langle n_1 \rangle - \frac{1}{2}.
		\label{eq:S1z_def}
	\end{equation}
	The magnitude of $S_1^z(0)$ is fixed by the zero-mode weight, but its sign depends on which member of the edge doublet is occupied. In a convention that pins the left zero mode to be occupied, the density on site $1$ receives a contribution $\frac{1}{2}|\phi_L(A_1)|^2 = \frac{1}{2}\mathcal{N}^2$ above the half-filled background, giving
	\begin{equation}
		S_1^z(0) \approx +\frac{1}{2}(1 - r^2) \approx +0.4997.
		\label{eq:S1z_zero_mode}
	\end{equation}
	If instead the opposite member of the edge doublet is occupied, then the same argument gives
	\begin{equation}
		S_1^z(0) \approx -\frac{1}{2}(1 - r^2) \approx -0.4997.
		\label{eq:S1z_zero_mode_negative}
	\end{equation}
	In a finite chain without an explicit sector-selection convention, the exponentially split near-zero modes can mix the two edges, so the sign is not universal. What is robust in the strongly dimerized limit is the magnitude $|S_1^z(0)| \approx \frac{1}{2}$.

	\section{Proof of Theorem~\ref*{thm:free_decay}: Exact Decay of the Boundary-Mode Return Amplitude}
	\label{app:theorem1_proof}
	We prove the exact integral representation~(\ref*{eq:AL_exact})--(\ref*{eq:Fk}) and the asymptotic decay~(\ref*{eq:AL_asymp}) stated in Theorem~\ref*{thm:free_decay}.
	Throughout, $a = J_1^f > b = J_2^f > 0$ are the final trivial SSH one-body hopping amplitudes, with $J_\ell^f=2J_\ell^{xy,f}$ in the Pauli-spin convention of Eq.~(\ref*{eq:Ham}); $r = -J_1^i/J_2^i$ with $|r| < 1$, and $\mathcal{N} = \sqrt{1-r^2}$.
	\subsection*{Step 1: Eigenstates of the Final Trivial SSH Chain}
	For a semi-infinite open chain with sites $m = 1, 2, \ldots$, the open boundary condition requires $u_{B,0}(k) = 0$.
	The unique normalizable solutions for the $B$-sublattice amplitudes are the sine waves
	\begin{equation}
		u_{B,m}(k) = \sqrt{\frac{2}{\pi}} \sin(mk), \qquad k \in (0,\pi).
		\label{eq:app_uB}
	\end{equation}
	The $A$-sublattice amplitudes then follow from the Schr\"{o}dinger equation
	$\varepsilon(k)\, u_{A,m}(k) = a\, u_{B,m}(k) + b\, u_{B,m-1}(k)$:
	\begin{equation}
		u_{A,m}(k) = \sqrt{\frac{2}{\pi}}
		\frac{a\sin(mk) + b\sin[(m-1)k]}{\varepsilon(k)}.
		\label{eq:app_uA}
	\end{equation}
	To confirm the dispersion, substitute~(\ref{eq:app_uA}) into the second Schr\"{o}dinger equation $\varepsilon(k)\, u_{B,m}(k) = a\, u_{A,m}(k) + b\, u_{A,m+1}(k)$ and use the identity $\sin[(m\pm 1)k] = \sin(mk)\cos k \pm \cos(mk)\sin k$.
	After simplification, the $\sin(mk)$ and $\cos(mk)$ terms collect separately to give
	\begin{equation}
		\varepsilon(k)^2 = a^2 + b^2 + 2ab\cos k,
		\label{eq:app_disp}
	\end{equation}
	as required.
	At $m = 1$, since $\sin(0) = 0$, Eq.~(\ref{eq:app_uA}) reduces to
	\begin{equation}
		u_{A,1}(k) = \sqrt{\frac{2}{\pi}}\, \frac{a\sin k}{\varepsilon(k)}.
		\label{eq:app_uA1}
	\end{equation}
	
	\subsection*{Step 2: The $A$-to-$A$ Propagator}
	The initial zero mode $|\phi_L\rangle = \mathcal{N}\sum_{m\geq1}r^{m-1}|A_m\rangle$ lives entirely on the $A$-sublattice. Because $h_f$ is bipartite and chiral, every continuum mode with energy $+\varepsilon(k)$ has a partner with energy $-\varepsilon(k)$ and the same probability weights on each sublattice. Writing the spectral resolution with both partners and restricting to the $A$-to-$A$ matrix element therefore gives
	\begin{equation}
		\begin{aligned}
			\langle A_1 | e^{-ih_f t} | A_m \rangle
			=& \int_0^\pi dk\, u_{A,1}(k)\, u_{A,m}(k)\, \frac{e^{-i\varepsilon(k)t} + e^{+i\varepsilon(k)t}}{2} \\
			=& \int_0^\pi dk\, u_{A,1}(k)\, u_{A,m}(k)\, \cos[\varepsilon(k)\, t].
		\end{aligned}
		\label{eq:app_prop}
	\end{equation}
	This is the precise reason the cosine appears in the $A$-to-$A$ propagator.
	Substituting~(\ref{eq:app_prop}) into $A_L(t) = \sum_{m\geq1}\mathcal{N}r^{m-1} \langle A_1|e^{-i H_f^{(1)} t}|A_m\rangle$ and exchanging the sum and integral (justified by the absolute convergence of the geometric series for $|r|<1$):
	\begin{equation}
		A_L(t) = \int_0^\pi dk\, u_{A,1}(k)
		\left[\mathcal{N}\sum_{m\geq1}r^{m-1}u_{A,m}(k)\right]\cos[\varepsilon(k)\,t].
		\label{eq:app_AL_step}
	\end{equation}
	\subsection*{Step 3: Evaluating the Geometric Sum}
	Substituting~(\ref{eq:app_uA}) into the bracketed sum in~(\ref{eq:app_AL_step}):
	\begin{equation}
		\begin{aligned}
			\mathcal{N}\sum_{m\geq1}r^{m-1}u_{A,m}(k)
			=& \mathcal{N}\sqrt{\frac{2}{\pi}}\,\frac{1}{\varepsilon(k)}
			\sum_{m\geq1}r^{m-1} \\
			& \times \bigl[a\sin(mk)+b\sin[(m-1)k]\bigr].
		\end{aligned}
		\label{eq:app_geo_setup}
	\end{equation}
	The first sum is evaluated by writing $\sin(mk) = \mathrm{Im}[e^{imk}]$:
	\begin{align}
		\sum_{m\geq1}r^{m-1}\sin(mk)
		&= \mathrm{Im}\!\left[\sum_{m\geq1}r^{m-1}e^{imk}\right]
		= \mathrm{Im}\!\left[\frac{e^{ik}}{1-re^{ik}}\right]
		\nonumber\\[4pt]
		&= \frac{\sin k}{1-2r\cos k+r^2}.
		\label{eq:app_geo1}
	\end{align}
	For the second sum, let $\ell = m-1$ so that the index starts at $\ell = 0$.
	The $\ell = 0$ term vanishes since $\sin(0) = 0$, leaving
	\begin{equation}
		\sum_{m\geq1}r^{m-1}\sin[(m-1)k]
		= \sum_{\ell\geq1}r^\ell\sin(\ell k)
		= r\cdot\frac{\sin k}{1-2r\cos k+r^2}.
		\label{eq:app_geo2}
	\end{equation}
	Combining~(\ref{eq:app_geo1}) and~(\ref{eq:app_geo2}) in~(\ref{eq:app_geo_setup}):
	\begin{equation}
		\mathcal{N}\sum_{m\geq1}r^{m-1}u_{A,m}(k)
		= \mathcal{N}\sqrt{\frac{2}{\pi}}\,
		\frac{(a+br)\sin k}{\varepsilon(k)(1-2r\cos k+r^2)}.
		\label{eq:app_geo_result}
	\end{equation}
	\subsection*{Step 4: Assembling $F(k)$}
	Multiplying~(\ref{eq:app_geo_result}) by $u_{A,1}(k)$ from~(\ref{eq:app_uA1}):
	\begin{align}
		F(k) &= u_{A,1}(k)\cdot
		\mathcal{N}\sum_{m\geq1}r^{m-1}u_{A,m}(k)
		\nonumber\\[4pt]
		&= \sqrt{\frac{2}{\pi}}\,\frac{a\sin k}{\varepsilon(k)}
		\cdot
		\mathcal{N}\sqrt{\frac{2}{\pi}}\,
		\frac{(a+br)\sin k}{\varepsilon(k)(1-2r\cos k+r^2)}
		\nonumber\\[4pt]
		&= \frac{2\mathcal{N}\,a(a+br)\sin^2 k}
		{\pi\,\varepsilon(k)^2(1-2r\cos k+r^2)}.
		\label{eq:app_Fk}
	\end{align}
	This completes the proof of the exact integral representation~(\ref*{eq:AL_exact}) and~(\ref*{eq:Fk}).
	
	\subsection*{Step 5: Stationary-Phase Asymptotics and the $t^{-3/2}$ Decay}
	For large $t$, the integral $\int_0^\pi dk\,F(k)\cos[\varepsilon(k)\,t]$ is dominated by the \emph{stationary points} of the phase $\varepsilon(k)\,t$, where $d\varepsilon/dk = 0$. From~(\ref{eq:app_disp}),
	\begin{equation}
		\frac{d\varepsilon}{dk} = -\frac{ab\sin k}{\varepsilon(k)},
		\label{eq:app_deps}
	\end{equation}
	so the stationary points are $k = 0$ and $k = \pi$.
	\paragraph{Near $k=0$:}
	Taylor expanding to leading order in $k$:
	\begin{equation}
		\varepsilon(k) \approx (a+b) - \frac{ab}{2(a+b)}k^2,
		\qquad
		F(k) \approx f_0\, k^2,
		\label{eq:app_near0}
	\end{equation}
	where
	\begin{equation}
		f_0 = \frac{2\mathcal{N}\,a(a+br)}{\pi(a+b)^2(1-r)^2}.
		\label{eq:app_f0}
	\end{equation}
	\paragraph{Near $k=\pi$:}
	Setting $q = \pi - k$ and expanding to leading order in $q$:
	\begin{equation}
		\varepsilon(k) \approx (a-b) + \frac{ab}{2(a-b)}q^2,
		\qquad
		F(k) \approx f_\pi\, q^2,
		\label{eq:app_nearpi}
	\end{equation}
	where
	\begin{equation}
		f_\pi = \frac{2\mathcal{N}\,a(a+br)}{\pi(a-b)^2(1+r)^2}.
		\label{eq:app_fpi}
	\end{equation}
	The crucial observation is that $F(k)$ vanishes \emph{quadratically} at both stationary points, namely $F(k) \sim k^2$ near $k = 0$ and $F(k) \sim q^2$ near $k = \pi$. If $F$ were nonzero at an endpoint stationary point, the standard contribution would scale as $t^{-1/2}$. The additional quadratic factor contributes one extra power of $t^{-1}$, so each endpoint contributes at order $t^{-3/2}$ rather than $t^{-1/2}$.
	The relevant prototype is the endpoint integral (see, e.g.,~\cite{Bender_Orszag}):
	\begin{equation}
		\begin{aligned}
			\int_0^\infty dk\, k^2 &
			\cos\!\left[\left(\Omega - \lambda k^2\right)t\right] \\
			=& \frac{\sqrt{\pi}}{4}\,\lambda^{-3/2}\,t^{-3/2}
			\cos\!\left(\Omega t - \frac{3\pi}{4}\right) 
			+ O(t^{-5/2}),
		\end{aligned}
		\label{eq:app_prototype}
	\end{equation}
	valid for $\lambda, t > 0$. Applying~(\ref{eq:app_prototype}) at $k = 0$ with $\Omega = a+b$ and $\lambda = ab/[2(a+b)]$, and the analogous formula at $k = \pi$ (where the sign of $\lambda$ reverses, shifting the phase by $\pi/2$), gives the two-term asymptotic~(\ref*{eq:AL_asymp}) with prefactors~(\ref*{eq:C0})--(\ref*{eq:Cpi}). This completes the proof. \qed

	\section{Supplementary Non-Interacting Results}
	\label{app:sp_supplementary}
	This appendix collects the numerical validation benchmarks for the non-interacting results used in Sec.~\ref*{sec:single} and provides two supplementary checks: a reduced-dimerization quench and the center-site magnetization as a bulk reference. The latter is expected to remain at numerical zero for the strongly localized initial edge mode used here, since the topological zero-mode weight is exponentially concentrated near the boundary rather than in the bulk.
	
	\subsection{Numerical Benchmarks}
	\label{app:single_particle_benchmarks}
	The direct comparison between the exact boundary-return integral and finite-chain propagation for the main-text quench ($\delta_i = -0.95$, $\delta_f = +0.95$) is shown in Fig.~\ref*{fig:AL_benchmark} of Sec.~\ref*{sec:theorem1}, because it directly validates Theorem~\ref*{thm:free_decay}. The analogous direct comparison for the smaller-dimerization quench is shown in Fig.~\ref*{fig:AL_benchmark_d07} below. Fig.~\ref*{fig:AL_asymptotic_benchmark} in Sec.~\ref*{sec:theorem1} compares the exact boundary-return amplitude with the leading large-time asymptotic approximation defined in Eq.~\ref*{eq:AL_asymp_def}.
	
	As a further check of the correlation-matrix approach, we compute the site spin autocorrelation function $\langle S_j^z(t) S_j^z(0) \rangle$ and compare it with full exact diagonalization.
	Using $S_j^z = c_j^\dagger c_j - \frac{1}{2}$, we write
	\begin{equation}
		\begin{aligned}
			\langle S_j^z(t) & S_j^z(0) \rangle \\
			&= \langle c_j^\dagger(t)c_j(t)c_j^\dagger(0)c_j(0)\rangle
			- \frac12 C_{jj}(t) - \frac12 C_{jj}(0) + \frac14,
		\end{aligned}
		\label{eq:autocorr_raw_app}
	\end{equation}
	where $C_{jj}(t)=\langle c_j^\dagger(t)c_j(t)\rangle$ is the instantaneous density.
	The four-operator term is evaluated using Wick's theorem for the free-fermion ground state,
	\begin{equation}
		\begin{aligned}
			\langle c_j^\dagger(t)c_j(t)c_j^\dagger(0) & c_j(0)\rangle \\
			&= C_{jj}(t)C_{jj}(0) + \langle c_j^\dagger(t)c_j(0)\rangle \langle c_j(t)c_j^\dagger(0)\rangle .
		\end{aligned}
		\label{eq:wick_app}
	\end{equation}
	Introducing the mixed-time correlator
	\begin{equation}
		G_{jj} \equiv \langle c_j^\dagger(t)c_j(0)\rangle = [U^\dagger(t) C(0)]_{jj},
		\label{eq:Gjj_app}
	\end{equation}
	and noting the anticommutation relation $\{c_j(t),c_j^\dagger(0)\}=U_{jj}(t)$, we obtain
	\begin{equation}
		\langle c_j(t)c_j^\dagger(0)\rangle = U_{jj}(t) - G_{jj}^*.
		\label{eq:anticomm_app}
	\end{equation}
	Substituting these expressions into Eq.~\eqref{eq:autocorr_raw_app} yields the compact result
	\begin{align}
		\langle S_j^z(t) S_j^z(0) \rangle
		&= C_{jj}(t)C_{jj}(0) + G_{jj}\bigl(U_{jj}-G_{jj}^*\bigr) \nonumber \\
		&\quad - \frac12 C_{jj}(t) - \frac12 C_{jj}(0) + \frac14 .
		\label{eq:autocorr_final}
	\end{align}
	All quantities on the right-hand side are computed from the single-particle time-evolution matrix $U(t)$ and the initial correlation matrix $C(0)$, making the evaluation $\mathcal{O}(N^3)$ at each time step.
	Figure~\ref{fig:autocorr_benchmark} compares Eq.~\eqref{eq:autocorr_final} with the exact-diagonalization result for a chain of $N=10$ with $\delta_i=-0.95$, $\delta_f=+0.95$; the two calculations agree to machine precision throughout the entire evolution.
	\begin{figure}[htbp]
		\centering
		\includegraphics[width=\linewidth]{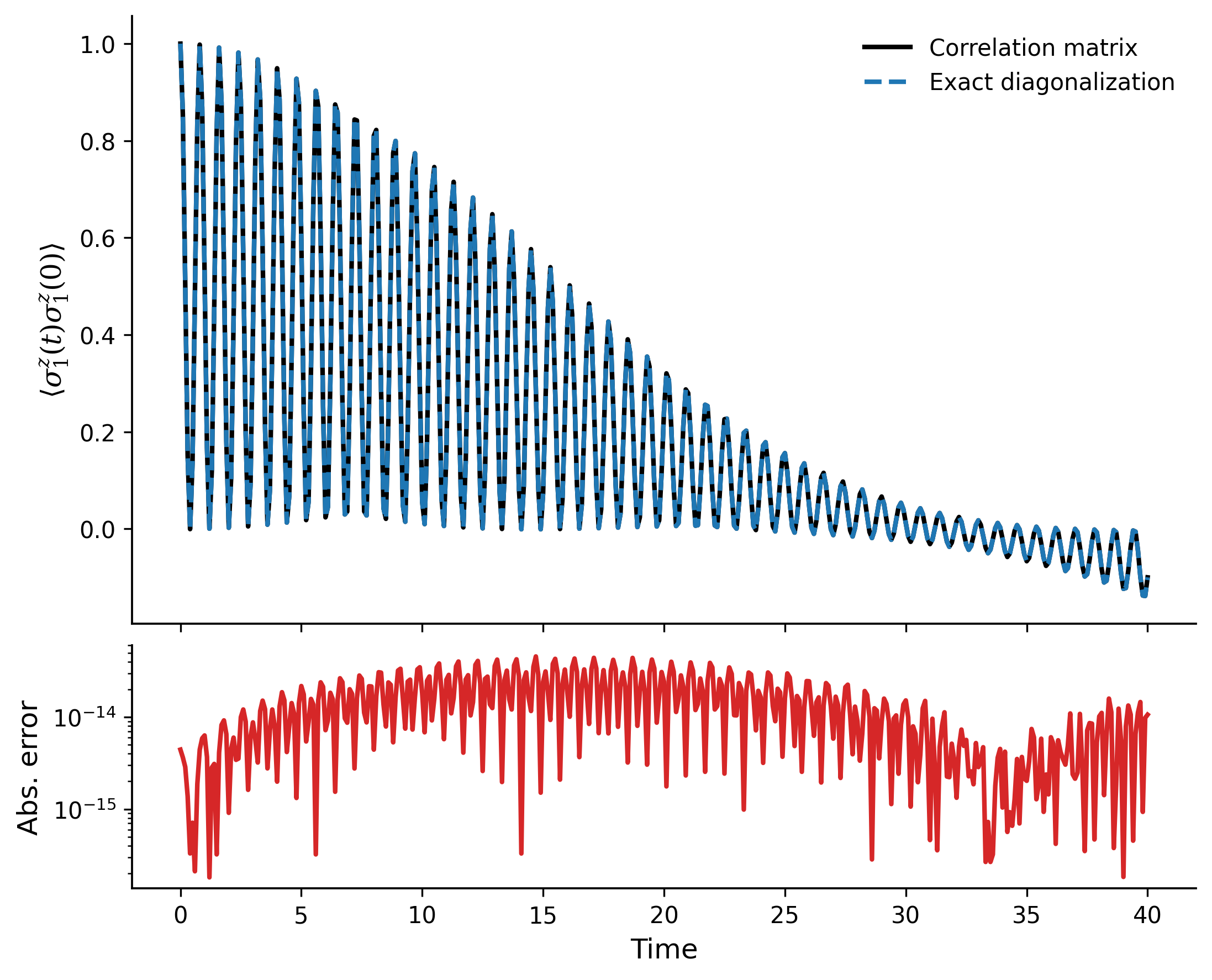}
		\caption{Benchmark of the correlation-matrix formula for the site autocorrelation function against full exact diagonalization in the non-interacting SSH chain. The calculation uses $L = 10$, $\delta_i = -0.95$, $\delta_f = +0.95$, and the initial half-filled ground state. The upper panel compares $\langle \sigma_1^z(t)\sigma_1^z(0)\rangle$ obtained from Eq.~\eqref{eq:autocorr_final} with the corresponding many-body exact-diagonalization result. The lower panel shows the absolute difference between the two calculations. The agreement is at machine precision throughout the plotted interval.}
		\label{fig:autocorr_benchmark}
	\end{figure}

	\subsection{Smaller Dimerization}
	\label{app:delta07}
	This subsection gives a parameter-variation check of the non-interacting SP $\to$ SP quench by reducing the dimerization magnitude from $|\delta| = 0.95$ to
	$|\delta| = 0.70$. We take $\delta_i = -0.70$, $\delta_f = +0.70$, and $J_z^i = J_z^f = 0$ at the same system size $L = 200$ as in the main text.
	The resulting edge-magnetization dynamics are shown in Fig.~\ref{fig:delta07_edge}. The same qualitative relaxation mechanism is present, but the relaxation is visibly faster than for the strongly dimerized case shown in Fig.~\ref*{fig:single_particle}. For this comparison quench, the initial and final SSH one-body hopping amplitudes are (via Eq.~\eqref{eq:J_param})
	\begin{equation}
		J_1^i = 0.60, \qquad J_2^i = 3.40, \qquad J_1^f = 3.40, \qquad J_2^f = 0.60.
		\label{eq:delta07_couplings}
	\end{equation}
	Hence the initial zero-mode ratio is
	\begin{equation}
		r = -\frac{J_1^i}{J_2^i} = -\frac{0.60}{3.40} \approx -0.1765,
		\label{eq:delta07_r}
	\end{equation}
	so the initial boundary mode is still localized but substantially less tightly bound to the edge than for $\delta_i = -0.95$.
	On the final trivial side one has SSH one-body hopping amplitudes $a = 3.40$ and $b = 0.60$. The asymptotic estimate of Sec.~\ref*{sec:theorem1} applies to the boundary-mode return amplitude and to its envelope scale $\tau_\epsilon^{\rm env}$ defined in Eq.~(\ref*{eq:tau_env_def}), not directly to the full edge magnetization. Nevertheless, $|A_L(t)|^2$ is the boundary-mode contribution to the site-1 density, and the full edge magnetizations include this component together with contributions from the other initially occupied modes. Since the SSH one-body hopping $b$ increases from $0.10$ at $\delta_f = 0.95$ to $0.60$ at $\delta_f = 0.70$, Corollary~\ref*{cor:tau_scaling} predicts an envelope time scale smaller by a factor close to $6$ for the boundary-mode return amplitude. The beat period is likewise reduced to
	\begin{equation}
		T_{\rm beat} = \frac{\pi}{b} = \frac{\pi}{0.60} \approx 5.24,
		\label{eq:delta07_tbeat}
	\end{equation}
	which is much shorter than the value $T_{\rm beat} \approx 31.4$ at $\delta_f = 0.95$. These analytical changes apply directly to $A_L(t)$, while the visibly faster relaxation of $S_1^z(t)$ and $S_N^z(t)$ in Fig.~\ref{fig:delta07_edge} occurs on a comparable reduced scale and is consistent with the same boundary-mode relaxation mechanism. No separate asymptotic theorem is claimed for the full magnetization.
	\begin{figure}[htbp]
		\centering
		\includegraphics[width=\linewidth]{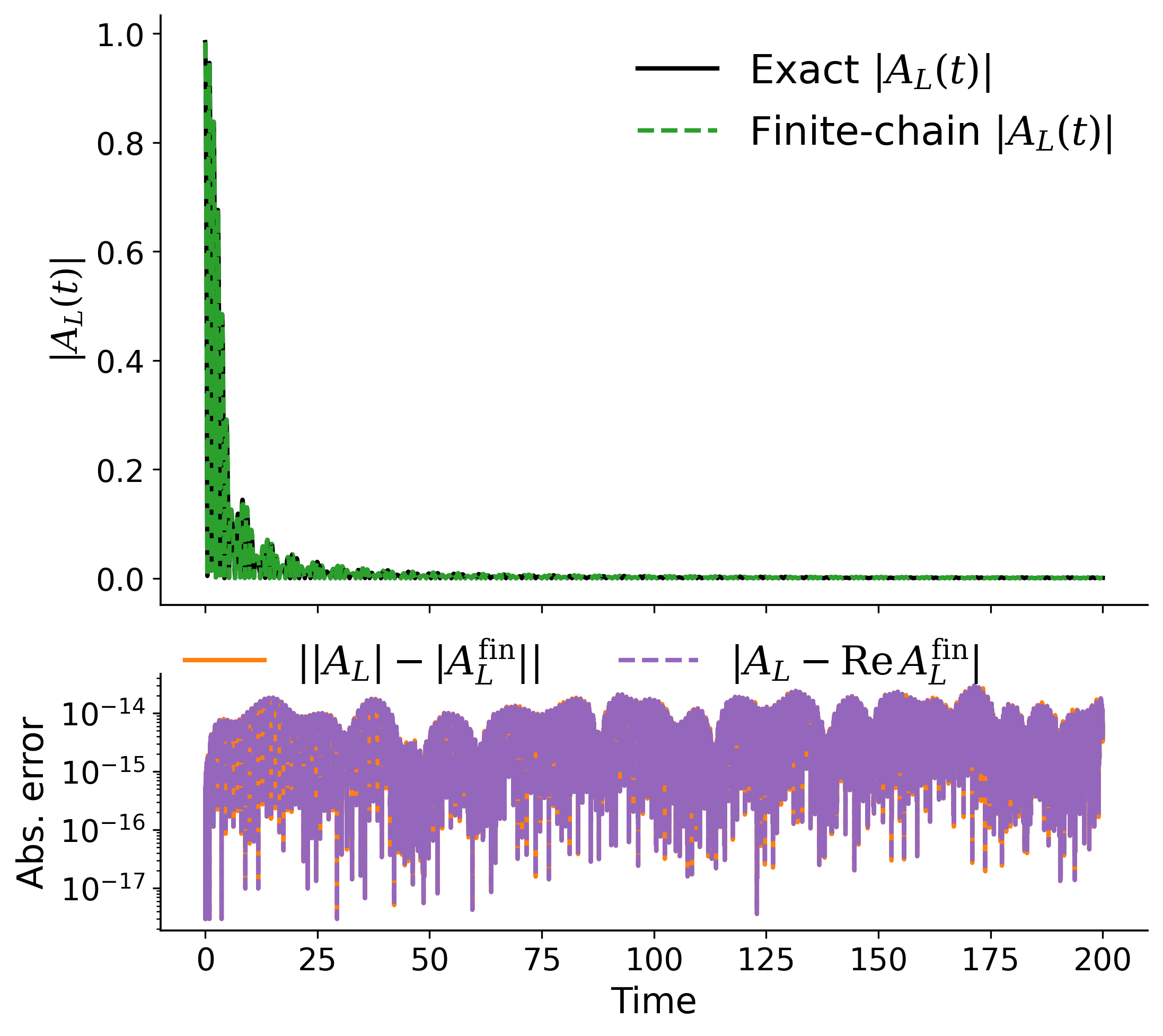}
		\caption{Benchmark of the exact boundary-return amplitude for the smaller-dimerization SP $\to$ SP quench with $L = 200$, $\delta_i = -0.70$, and $\delta_f = +0.70$. The upper panel compares $|A_L(t)|$ obtained from the exact semi-infinite integral representation, Eq.~(\ref*{eq:AL_exact})--(\ref*{eq:Fk}), with direct propagation of the initial left boundary mode, Eq.~(\ref*{eq:thm_zeromode}), on the finite open chain. The lower panel shows the absolute difference between the two results over the full interval $0 \leq t \leq 200$. As in the strongly dimerized benchmark of Fig.~\ref*{fig:AL_benchmark}, the agreement is at numerical precision throughout the plotted window, confirming that the exact boundary-mode return amplitude remains quantitatively reliable for the reduced dimerization.}
		\label{fig:AL_benchmark_d07}
	\end{figure}
	\begin{figure}[htbp]
		\centering
		\includegraphics[width=\linewidth]{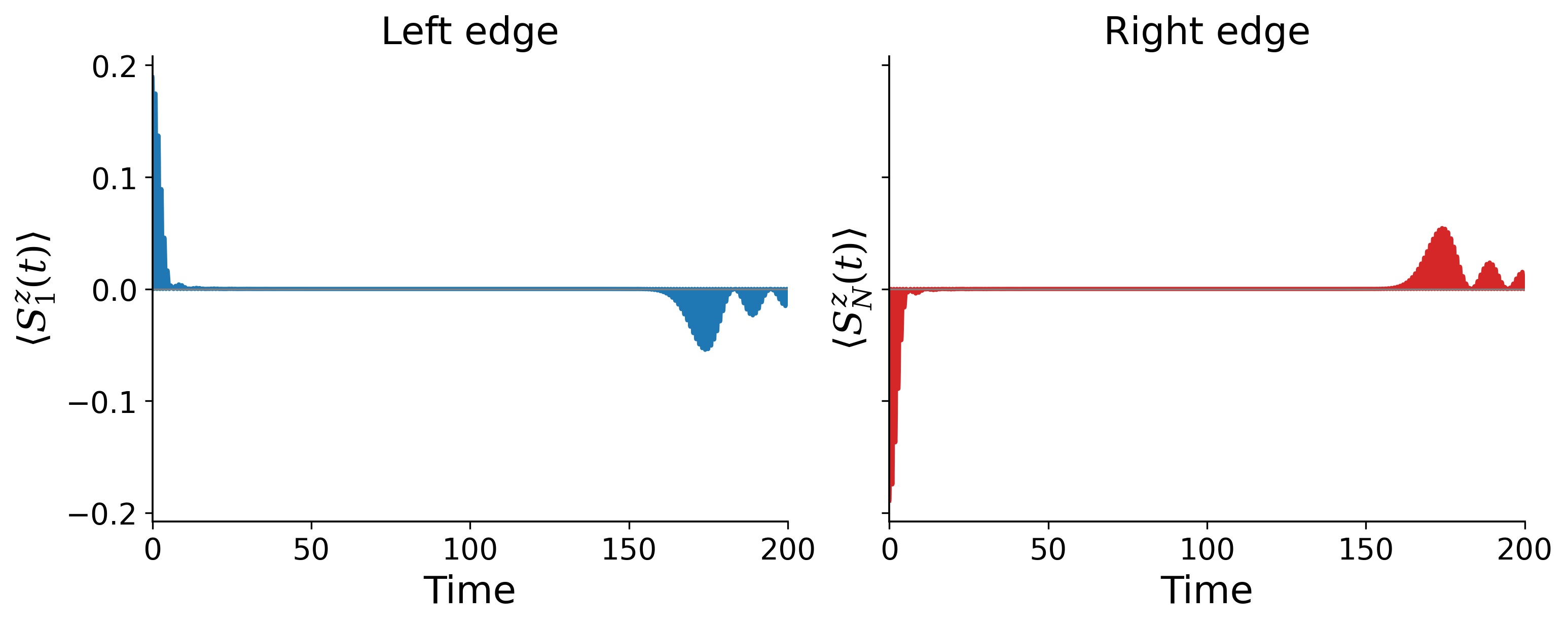}
		\caption{Non-interacting SP $\to$ SP quench for smaller dimerization magnitude, with $L = 200$, $\delta_i = -0.70$, $\delta_f = +0.70$, and $J_z^i = J_z^f = 0$. The left and right panels show $S_1^z(t)$ and $S_N^z(t)$, respectively. The full edge magnetizations are observables containing the boundary-mode density contribution together with contributions from the other initially occupied modes, so the theorem of Sec.~\ref*{sec:theorem1} is not an asymptotic theorem for these curves. Nevertheless, the relaxation of both edge signals occurs on a timescale comparable to the first decay of the boundary-mode return amplitude and is faster than in Fig.~\ref*{fig:single_particle}, consistent with the reduced envelope scale and the shorter beat period $T_{\rm beat} = \pi/[2(1-\delta_f)] \approx 5.24$ for $A_L(t)$.}
		\label{fig:delta07_edge}
	\end{figure}
	
	\subsection{Center-Site Magnetization}
	\label{app:center_spin}
	As a bulk reference, Fig.~\ref{fig:center_magnetization} shows the time evolution of the center-site magnetization $S_{N/2}^z(t)$ for the SP $\to$ SP quench with $L = 200$, $\delta_i = -0.95$, and $\delta_f = +0.95$. For the strongly localized initial edge mode considered here, one expects the center site to carry negligible weight already at $t=0$, because the boundary zero mode of Appendix~\ref{app:zero_mode} decays as $r^{m-1}$ with $|r| \approx 0.02564$, so its amplitude is exponentially small deep in the bulk. The numerical result is even stronger: within floating-point precision, $S_{N/2}^z(t)$ is indistinguishable from zero throughout the plotted interval. The tiny visible deviations in the raw output are at the level of numerical roundoff and have no physical significance. This provides a useful bulk reference against which the large edge response in Fig.~\ref*{fig:single_particle} can be contrasted.
	
	\begin{figure}[htbp]
		\centering
		\includegraphics[width=\linewidth]{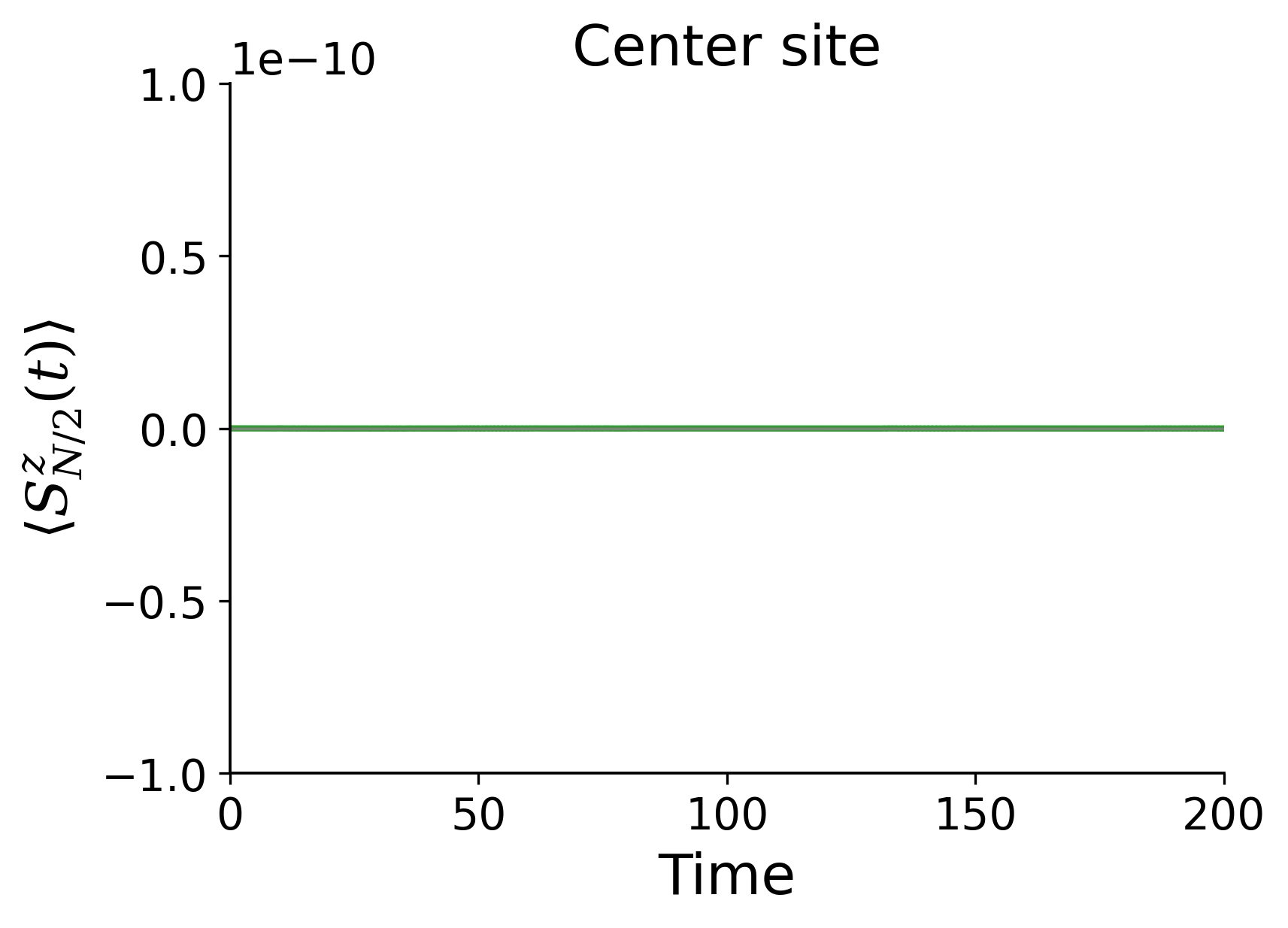}
		\caption{Quench dynamics of the center-site magnetization $S_{N/2}^z(t)$ in the non-interacting ($J_z = 0$) SSH chain for $L = 200$, $\delta_i = -0.95$, and $\delta_f = +0.95$, computed using the correlation-matrix method, Eq.~(\ref*{eq:corrmat_time}). The plotting window is fixed to a tiny symmetric range about zero, so the curve appears as a flat line. Any residual fluctuations in the underlying data are at the level of numerical roundoff, showing that the center site remains at numerical zero throughout the evolution.}
		\label{fig:center_magnetization}
	\end{figure}

	\section{Bound for the edge magnetization}
	\label{app:proof}
	
	This appendix proves Eq.~(\ref*{eq:bound_main}). Throughout, $a = J_1^f$ and $b = J_2^f$ are the final SSH one-body hopping amplitudes defined in Eq.~\eqref{eq:J_param}, with $a > b > 0$ for $\delta_f > 0$. The final Hamiltonian is written as
	\begin{align}
		H_b=H_0+bW,
		\label{eq:app_Hb}
	\end{align}
	where
	\begin{align}
		H_0=\sum_{m=1}^{N/2}h_{2m-1,2m},
		\label{eq:app_H0}
	\end{align}
	with
	\begin{align}
		h_{2m-1,2m}
		=&\,a\Bigl(
		\sigma_{2m-1}^x\sigma_{2m}^x
		+\sigma_{2m-1}^y\sigma_{2m}^y
		\nonumber\\
		&\hspace{1.3cm}
		+J_z^f\sigma_{2m-1}^z\sigma_{2m}^z
		\Bigr),
		\label{eq:app_h}
	\end{align}
	and
	\begin{align}
		W=\sum_{m=1}^{N/2-1}w_{2m,2m+1},
		\label{eq:app_W}
	\end{align}
	with
	\begin{align}
		w_{2m,2m+1}
		=&\,\sigma_{2m}^x\sigma_{2m+1}^x
		+\sigma_{2m}^y\sigma_{2m+1}^y
		\nonumber\\
		&+J_z^f\sigma_{2m}^z\sigma_{2m+1}^z .
		\label{eq:app_w}
	\end{align}
	This decomposition is exact for the bond-scaled parametrization of Eq.~(\ref*{eq:params}). The limit $b=0$ removes both inter-dimer terms and leaves a sum of independent two-site Hamiltonians.
	
	Let
	\begin{align}
		S_{1,b}^z(t)=e^{iH_bt}S_1^ze^{-iH_bt},\qquad
		S_{1,0}^z(t)=e^{iH_0t}S_1^ze^{-iH_0t}.
		\label{eq:app_S_defs}
	\end{align}
	For fixed $t$, introduce
	\begin{align}
		F(s)=e^{iH_bs}S_{1,0}^z(t-s)e^{-iH_bs},
		\qquad 0\le s\le t .
		\label{eq:app_F_def}
	\end{align}
	Then
	\begin{align}
		F(0)=S_{1,0}^z(t),\qquad
		F(t)=S_{1,b}^z(t).
		\label{eq:app_F_endpoints}
	\end{align}
	Differentiating Eq.~\eqref{eq:app_F_def} gives
	\begin{align}
		\frac{dF}{ds}
		=&\,i e^{iH_bs}
		[H_b,S_{1,0}^z(t-s)]
		e^{-iH_bs}
		\nonumber\\
		&+e^{iH_bs}
		\frac{d}{ds}S_{1,0}^z(t-s)
		e^{-iH_bs}.
		\label{eq:app_dF_1}
	\end{align}
	Since
	\begin{align}
		\frac{d}{du}S_{1,0}^z(u)
		=i[H_0,S_{1,0}^z(u)],
		\label{eq:app_H0_eom}
	\end{align}
	one has
	\begin{align}
		\frac{d}{ds}S_{1,0}^z(t-s)
		=-i[H_0,S_{1,0}^z(t-s)].
		\label{eq:app_chain_rule}
	\end{align}
	Substitution into Eq.~\eqref{eq:app_dF_1} yields
	\begin{align}
		\frac{dF}{ds}
		=&\,i e^{iH_bs}
		[H_b-H_0,S_{1,0}^z(t-s)]
		e^{-iH_bs}
		\nonumber\\
		=&\,\frac{ib}{2} e^{iH_bs}
		[W,S_{1,0}^z(t-s)]
		e^{-iH_bs}.
		\label{eq:app_dF_final}
	\end{align}
	Integrating over $s$ gives the exact identity
	\begin{align}
		S_{1,b}^z(t)-S_{1,0}^z(t)
		=
		\frac{ib}{2}\int_0^t ds\,
		e^{iH_bs}
		[W,S_{1,0}^z(t-s)]
		e^{-iH_bs}.
		\label{eq:app_identity}
	\end{align}
	
	Taking the operator norm and using invariance under unitary conjugation,
	\begin{align}
		\bigl\|S_{1,b}^z(t)-S_{1,0}^z(t)\bigr\|
		\le
		\frac{|b|}{2}\int_0^t ds\,
		\bigl\|[W,S_{1,0}^z(t-s)]\bigr\|.
		\label{eq:app_norm_1}
	\end{align}
	Because $H_0$ is a sum of independent dimers and $[h_{mn},S_1^z]=0$ for all dimers $(m,n)$ with $m\ge3$, the Heisenberg evolution of $S_1^z$ under $H_0$ is generated entirely by $h_{12}$, which couples only sites $1$ and $2$. Hence  $S_{1,0}^z(t-s)$ is supported on sites $1$ and $2$ for all $t-s\ge0$, but no further. Hence all terms in $W$ commute with it except $w_{2,3}$:
	\begin{align}
		[W,S_{1,0}^z(t-s)]
		=
		[w_{2,3},S_{1,0}^z(t-s)].
		\label{eq:app_local_comm}
	\end{align}
	Using $\|[A,B]\|\le2\|A\|\|B\|$,
	\begin{align}
		\bigl\|[w_{2,3},S_{1,0}^z(t-s)]\bigr\|
		\le
		2\|w_{2,3}\|\,\|S_{1,0}^z(t-s)\|.
		\label{eq:app_comm_bound}
	\end{align}
	The second factor is
	\begin{align}
		\|S_{1,0}^z(t-s)\|=\|S_1^z\|=\frac{1}{2}.
		\label{eq:app_S_norm}
	\end{align}
	For the first factor,
	\begin{align}
		\|w_{2,3}\|
		&\le
		\|\sigma_2^x\sigma_3^x\|
		+\|\sigma_2^y\sigma_3^y\|
		+|J_z^f|\|\sigma_2^z\sigma_3^z\|
		\nonumber\\
		&=2+|J_z^f|.
		\label{eq:app_w_norm}
	\end{align}
	Equations~\eqref{eq:app_comm_bound}--\eqref{eq:app_w_norm} imply
	\begin{align}
		\bigl\|[W,S_{1,0}^z(t-s)]\bigr\|
		\le 2+|J_z^f|.
		\label{eq:app_final_comm_bound}
	\end{align}
	Substitution into Eq.~\eqref{eq:app_norm_1} gives
	\begin{align}
		\bigl\|S_{1,b}^z(t)-S_{1,0}^z(t)\bigr\|
		\le \frac{|b|}{2}\,(2+|J_z^f|)\,t .
		\label{eq:app_operator_bound}
	\end{align}
	This proves the operator bound. For any normalized state $|\psi\rangle$,
	\begin{align}
		\left|
		\langle\psi|
		S_{1,b}^z(t)-S_{1,0}^z(t)
		|\psi\rangle
		\right|
		\le
		\bigl\|S_{1,b}^z(t)-S_{1,0}^z(t)\bigr\|,
		\label{eq:app_expect_bound}
	\end{align}
	which proves Eq.~(\ref*{eq:expect_main}).
	
	\section{Leakage identity for the first-dimer magnetization}
	\label{app:leakageproof}
	
	This Appendix derives Eq.~(\ref*{eq:M12_current_main}). Define
	\begin{align}
		M_{12}^z=S_1^z+S_2^z
		=\frac{1}{2}(\sigma_1^z+\sigma_2^z).
		\label{eq:app_M12}
	\end{align}
	The Heisenberg equation gives
	\begin{align}
		\frac{d}{dt}M_{12}^z(t)=i[H_b,M_{12}^z](t).
		\label{eq:app_M12_eom}
	\end{align}
	All Hamiltonian terms supported on sites $3,4,\ldots,N$ commute with $M_{12}^z$. The intra-dimer bond $(1,2)$ also commutes with $M_{12}^z$, since its $XY$ part conserves the total $z$ magnetization on sites $1$ and $2$, and its $ZZ$ part is diagonal in the $z$ basis. The $ZZ$ part of the weak bond $(2,3)$ also commutes with $M_{12}^z$. Thus the only nonzero contribution comes from
	\begin{align}
		h_{23}^{xy}
		=
		\frac{b}{2}(\sigma_2^x\sigma_3^x+\sigma_2^y\sigma_3^y).
		\label{eq:app_h23xy}
	\end{align}
	Using
	\begin{align}
		[\sigma^x,\sigma^z]=-2i\sigma^y,\qquad
		[\sigma^y,\sigma^z]=2i\sigma^x,
		\label{eq:app_pauli_comm}
	\end{align}
	we find
	\begin{align}
		[\sigma_2^x\sigma_3^x,M_{12}^z]
		&=
		\frac{1}{2}[\sigma_2^x,\sigma_2^z]\sigma_3^x
		=
		-i\sigma_2^y\sigma_3^x,
		\label{eq:app_comm_x}\\
		[\sigma_2^y\sigma_3^y,M_{12}^z]
		&=
		\frac{1}{2}[\sigma_2^y,\sigma_2^z]\sigma_3^y
		=
		i\sigma_2^x\sigma_3^y .
		\label{eq:app_comm_y}
	\end{align}
	Therefore,
	\begin{align}
		i[h_{23}^{xy},M_{12}^z]
		=
		\frac{b}{2}(\sigma_2^y\sigma_3^x-\sigma_2^x\sigma_3^y),
		\label{eq:app_current_identity}
	\end{align}
	which proves
	\begin{align}
		\frac{d}{dt}M_{12}^z(t)
		=
		\frac{b}{2}(\sigma_2^y\sigma_3^x-\sigma_2^x\sigma_3^y)(t).
		\label{eq:app_M12_final}
	\end{align}
	Since each two-site Pauli product has norm one,
	\begin{align}
		\left\|
		\sigma_2^y\sigma_3^x-\sigma_2^x\sigma_3^y
		\right\|
		\le2 .
		\label{eq:app_current_norm}
	\end{align}
	Consequently, for any normalized state,
	\begin{align}
		\left|
		\frac{d}{dt}\langle M_{12}^z(t)\rangle
		\right|
		\le|b|,
		\label{eq:app_current_expect_bound}
	\end{align}
	and integration gives
	\begin{align}
		\left|
		\langle M_{12}^z(t)\rangle-\langle M_{12}^z(0)\rangle
		\right|
		\le|b|\,t .
		\label{eq:app_M12_integrated}
	\end{align}
	
	\section{First-order transition formula around the dimer limit}
	\label{app:perturbationproof}
	
	This Appendix derives Eq.~(\ref*{eq:transition_prob_main}). Consider
	\begin{align}
		H_b=H_0+\frac{b}{2}\,W,
		\label{eq:app_pert_H}
	\end{align}
	and let
	\begin{align}
		H_0|\alpha\rangle=E_\alpha|\alpha\rangle,\qquad
		H_0|\beta\rangle=E_\beta|\beta\rangle.
		\label{eq:app_H0_eigen}
	\end{align}
	In the interaction picture with respect to $H_0$, the state satisfies
	\begin{align}
		i\frac{d}{dt}|\psi_I(t)\rangle
		=
		\frac{b}{2}\,W_I(t)|\psi_I(t)\rangle,
		\quad
		W_I(t)=e^{iH_0t}We^{-iH_0t}.
		\label{eq:app_interaction_picture}
	\end{align}
	Starting from $|\psi_I(0)\rangle=|\alpha\rangle$, the first-order correction is
	\begin{align}
		|\psi_I^{(1)}(t)\rangle
		=
		-\frac{ib}{2}\int_0^t ds\,W_I(s)|\alpha\rangle.
		\label{eq:app_first_order_state}
	\end{align}
	For $\beta\neq\alpha$, the transition amplitude is
	\begin{align}
		A_{\beta\leftarrow\alpha}^{(1)}(t)
		&=
		\langle\beta|\psi_I^{(1)}(t)\rangle
		\nonumber\\
		&=
		-\frac{ib}{2}\,W_{\beta\alpha}
		\int_0^t ds\,e^{i(E_\beta-E_\alpha)s},
		\label{eq:app_amp_integral}
	\end{align}
	where $W_{\beta\alpha}=\langle\beta|W|\alpha\rangle$. For $\Delta_{\beta\alpha}=E_\beta-E_\alpha\neq0$,
	\begin{align}
		A_{\beta\leftarrow\alpha}^{(1)}(t)
		=
		\frac{b}{2}\,W_{\beta\alpha}
		\frac{1-e^{i\Delta_{\beta\alpha}t}}
		{\Delta_{\beta\alpha}} .
		\label{eq:app_amp_final}
	\end{align}
	Thus
	\begin{align}
		P_{\beta\leftarrow\alpha}^{(1)}(t)
		&=
		\left|A_{\beta\leftarrow\alpha}^{(1)}(t)\right|^2
		\nonumber\\
		&=
		b^2|W_{\beta\alpha}|^2
		\frac{\sin^2(\Delta_{\beta\alpha}t/2)}
		{\Delta_{\beta\alpha}^2}.
		\label{eq:app_prob_final}
	\end{align}
	For $\Delta_{\beta\alpha}=0$, the same result is understood by taking the limit,
	\begin{align}
		P_{\beta\leftarrow\alpha}^{(1)}(t)
		=
		\frac{b^2}{4}|W_{\beta\alpha}|^2t^2 .
		\label{eq:app_prob_resonant}
	\end{align}
	Equation~\eqref{eq:app_prob_final} is the finite-system first-order transition formula. The usual Fermi golden rule follows from its long-time limit only after replacing the discrete final spectrum by a continuum of states.
	
\section{Boundary normal form and asymptotically long dimer memory}
\label{app:boundary_normal_form}

This appendix proves Proposition~\ref*{prop:asymptotic_boundary_memory}. The argument is a finite-order perturbative construction near the decoupled-dimer limit. The full boundary Hamiltonian is
\begin{align}
	H_b = H_0 + \varepsilon W, \qquad \varepsilon = \frac{b}{2},
	\label{eq:app_nf_Hb}
\end{align}
where $H_0$ is the sum of isolated dimer Hamiltonians and $W$ is the sum of inter-dimer bonds; $\varepsilon$ is the small parameter of the construction. The physical boundary charge is
\begin{align}
	Q_0 = M_{12}^z = S_1^z + S_2^z .
	\label{eq:app_nf_Q0}
\end{align}
Because $H_0$ decomposes into isolated dimers, $Q_0$ commutes with $H_0$ exactly:
\begin{align}
	[H_0, Q_0] = 0 .
	\label{eq:app_nf_Q0_commutes}
\end{align}
Once $\varepsilon W$ is switched on, this exact conservation is broken. A direct calculation gives
\begin{align}
	[H_b, Q_0] = \varepsilon [W, Q_0],
	\label{eq:app_nf_Q0_leaks}
\end{align}
so $Q_0$ leaks at rate $O(\varepsilon)$. The strategy is to construct, order by order in $\varepsilon$, a \emph{dressed} operator $Q^{(p)}$ --- an operator close to $Q_0$ but modified by small inter-dimer corrections --- whose commutator with $H_b$ is $O(\varepsilon^{p+1})$ rather than $O(\varepsilon)$. Physically, $Q^{(p)}$ is the approximately conserved charge of $H_b$ that reduces to the exact boundary charge $Q_0$ when $\varepsilon=0$.

\paragraph{Motivation for the ansatz.}
Suppose we try to cancel the $O(\varepsilon)$ leakage in Eq.~\eqref{eq:app_nf_Q0_leaks} by adding a correction $\varepsilon Q_1$, so that $Q^{(1)} = Q_0 + \varepsilon Q_1$. Then
\begin{align}
	[H_b, Q^{(1)}]
	= [H_0, Q_0]
	+ \varepsilon\bigl([H_0, Q_1] + [W, Q_0]\bigr)
	+ \varepsilon^2 [W, Q_1].
	\label{eq:app_nf_first_order_example}
\end{align}
The first term vanishes by Eq.~\eqref{eq:app_nf_Q0_commutes}. For the $O(\varepsilon)$ bracket to vanish we must choose $Q_1$ to satisfy
\begin{align}
	[H_0, Q_1] = -[W, Q_0].
	\label{eq:app_nf_Q1_equation}
\end{align}
If this equation can be solved, the leading leakage is pushed to $O(\varepsilon^2)$. Repeating the same step at every order up to $p$ uniquely determines $Q_1, Q_2, \ldots, Q_p$ and yields the dressed charge
\begin{align}
	Q^{(p)} = Q_0 + \varepsilon Q_1 + \varepsilon^2 Q_2
	+ \cdots + \varepsilon^p Q_p,
	\label{eq:app_nf_Qp}
\end{align}
whose commutator with $H_b$ has no terms of order $\varepsilon, \varepsilon^2, \ldots, \varepsilon^p$. Each $Q_n$ is not a free choice: it is \emph{forced} by the requirement that the $O(\varepsilon^n)$ commutator vanishes, and it is uniquely determined by $Q_0, \ldots, Q_{n-1}$ whenever the solvability condition below holds.

\paragraph{Solvability: the non-resonance condition.}
The equation~\eqref{eq:app_nf_Q1_equation} and its higher-order analogues are all of the form $[H_0, X] = A$ for a known operator $A$. Define the adjoint action of $H_0$ by
\begin{align}
	\mathcal{L}_0(A) = [H_0, A].
	\label{eq:app_nf_L0}
\end{align}
In the eigenbasis $H_0 |\alpha\rangle = E_\alpha |\alpha\rangle$, the equation $\mathcal{L}_0(X) = A$ reads
\begin{align}
	(E_\alpha - E_\beta)\, X_{\alpha\beta} = A_{\alpha\beta},
	\label{eq:app_nf_matrix_equation}
\end{align}
with solution
\begin{align}
	X_{\alpha\beta} = \frac{A_{\alpha\beta}}{E_\alpha - E_\beta}.
	\label{eq:app_nf_inverse_matrix}
\end{align}
This is well defined whenever $E_\alpha \neq E_\beta$. If $E_\alpha = E_\beta$ and $A_{\alpha\beta} \neq 0$, no solution exists: the right-hand side of Eq.~\eqref{eq:app_nf_matrix_equation} is nonzero while the left-hand side is zero. We call such a pair $(\alpha, \beta)$ a \emph{local resonance} for $A$.

Equivalently, define the time-average projector
\begin{align}
	\mathcal{P}_0(A)
	= \lim_{T \to \infty} \frac{1}{T}
	\int_0^T ds\; e^{iH_0 s} A\, e^{-iH_0 s},
	\label{eq:app_nf_P0}
\end{align}
which retains only the part of $A$ that does not oscillate under $H_0$-evolution. To see why this matters, note that in the eigenbasis of $H_0$ the matrix element $\bigl(e^{iH_0 s}Ae^{-iH_0 s}\bigr)_{\alpha\beta} = e^{i(E_\alpha-E_\beta)s}A_{\alpha\beta}$. Time-averaging kills every oscillating entry ($E_\alpha\neq E_\beta$) and retains exactly the degenerate entries ($E_\alpha=E_\beta$): $\mathcal{P}_0(A)_{\alpha\beta} = A_{\alpha\beta}\, \delta_{E_\alpha,E_\beta}$. These are precisely the matrix elements for which the denominator in Eq.~\eqref{eq:app_nf_inverse_matrix} vanishes, so $\mathcal{P}_0(A)=0$ is the exact condition that no division by zero occurs. The equation $\mathcal{L}_0(X)=A$ therefore has a solution if and only if $\mathcal{P}_0(A)=0$, i.e., $A$ has no zero-frequency (``DC'') component in the $H_0$ eigenbasis. When this condition holds, the unique solution with zero diagonal is $X = \mathcal{I}_0(A)$, where $\mathcal{I}_0 = \mathcal{L}_0^{-1}$ denotes the inverse on the off-resonant subspace.

\paragraph{Non-resonance assumption.}
For the chosen order $p$, we assume that the recursively generated operators
\begin{align}
	R_n = [W, Q_{n-1}], \qquad n = 1, \ldots, p,
	\label{eq:app_nf_Rn}
\end{align}
satisfy $\mathcal{P}_0(R_n) = 0$ on the finite cluster of the first $n+1$ dimers, and that the inverse $\mathcal{I}_0$ is bounded on that cluster:
\begin{align}
	\bigl\|\mathcal{I}_0(R_n)\bigr\| \le \Gamma_p \|R_n\|,
	\qquad n = 1, \ldots, p,
	\label{eq:app_nf_inverse_bound}
\end{align}
for a finite constant $\Gamma_p$ that depends only on the minimum nonzero energy gap of $H_0$ on that cluster. Because this is a finite-cluster condition at each fixed $p$, it fails only at a finite discrete set of values of $J_z^f$. The isotropic point $J_z^f = 1$ is one such excluded value; see Sec.~\ref*{sec:isotropic_resonance} and Appendix~\ref{app:isotropic_channel_derivation}.

\paragraph{Construction of the dressed charge.}
We proceed by induction. The base case is $Q_0$, supported on the first dimer. Suppose $Q_{n-1}$ has been constructed and is supported on the first $n$ dimers. Because $W$ is a sum of inter-dimer bonds and each bond connects adjacent dimers, only the bonds touching the support of $Q_{n-1}$ can contribute a nonzero commutator. Hence $R_n = [W, Q_{n-1}]$ is supported on the first $n+1$ dimers. By the non-resonance assumption, we may invert $\mathcal{L}_0$ on $R_n$ and define
\begin{align}
	Q_n = -\mathcal{I}_0(R_n)
	= -\mathcal{I}_0\bigl([W, Q_{n-1}]\bigr),
	\qquad n = 1, \ldots, p.
	\label{eq:app_nf_recursion}
\end{align}
By construction, $\mathcal{L}_0(Q_n) = [H_0, Q_n] = -R_n = -[W, Q_{n-1}]$, which is precisely the cancellation condition
\begin{align}
	[H_0, Q_n] + [W, Q_{n-1}] = 0, \qquad n = 1, \ldots, p.
	\label{eq:app_nf_cancel}
\end{align}
Equation~\eqref{eq:app_nf_cancel} is the central identity of the construction. The first-order example~\eqref{eq:app_nf_first_order_example} illustrates how it works: the first bracket on the right is zero by Eq.~\eqref{eq:app_nf_Q0_commutes} and the second bracket is zero by Eq.~\eqref{eq:app_nf_cancel} with $n=1$, leaving only $\varepsilon^2 [W, Q_1]$ as the residual. The same cancellation repeats at every order.

\paragraph{Norm bounds.}
Because the support of $Q_n$ grows by at most one dimer at each step and the Hilbert space on any fixed finite cluster is finite-dimensional, the operator norm of $Q_n$ is bounded by a constant that depends only on $p$ and the dimer Hamiltonian, not on the total chain length. This chain-length independence is essential: it means the dressed charge and all subsequent bounds are bulk-independent properties of the boundary, and remain valid as $L\to\infty$. More precisely, using Eq.~\eqref{eq:app_nf_inverse_bound} and the fact that $\|R_n\|$ is bounded by finitely many commutators of bounded local operators, there exists $A_p < \infty$ such that
\begin{align}
	\|Q_n\| \le A_p, \qquad n = 1, \ldots, p.
	\label{eq:app_nf_Qn_bound}
\end{align}

\paragraph{Almost conservation of $Q^{(p)}$.}
We now compute $[H_b, Q^{(p)}]$ directly. Expanding $H_b = H_0 + \varepsilon W$ and collecting powers of $\varepsilon$,
\begin{align}
	[H_b, Q^{(p)}]
	&= \sum_{n=0}^{p} \varepsilon^n [H_0, Q_n]
	+ \sum_{n=0}^{p} \varepsilon^{n+1} [W, Q_n]
	\nonumber \\
	&= \varepsilon^{p+1} [W, Q_p].
	\label{eq:app_nf_commutator}
\end{align}
The second equality is an exact algebraic identity: for $n = 0$ the term $[H_0, Q_0] = 0$ by Eq.~\eqref{eq:app_nf_Q0_commutes}; for $n = 1, \ldots, p$ the combination $\varepsilon^n [H_0, Q_n] + \varepsilon^n [W, Q_{n-1}] = 0$ by Eq.~\eqref{eq:app_nf_cancel}; and the surviving term at order $\varepsilon^{p+1}$ is $[W, Q_p]$. Since $Q_p$ is supported on the first $p+1$ dimers, only finitely many terms of $W$ contribute, so for some constant $B_p < \infty$ independent of chain length,
\begin{align}
	\|[W, Q_p]\| \le B_p.
	\label{eq:app_nf_residual_bound}
\end{align}
Therefore
\begin{align}
	\|[H_b, Q^{(p)}]\| \le B_p |\varepsilon|^{p+1}.
	\label{eq:app_nf_comm_bound_final}
\end{align}

\paragraph{Proximity to the physical charge.}
The dressed charge $Q^{(p)}$ differs from $Q_0$ by correction terms starting at order $\varepsilon$. For $|\varepsilon| \le 1$, the bound~\eqref{eq:app_nf_Qn_bound} gives
\begin{align}
	\|Q^{(p)} - Q_0\|
	\le \sum_{n=1}^{p} |\varepsilon|^n \|Q_n\|
	\le p\, A_p\, |\varepsilon|
	\eqqcolon C_p |\varepsilon|,
	\label{eq:app_nf_close_final}
\end{align}
where we set $C_p = p A_p$. Equations~\eqref{eq:app_nf_comm_bound_final} and~(\ref*{eq:app_nf_close_final}) prove Eqs.~(\ref*{eq:dressed_charge_comm_main}) and~(\ref*{eq:dressed_charge_close_main}).

\paragraph{From almost conservation to a physical bound.}
It remains to translate the almost conservation of $Q^{(p)}$ into a statement about the observable $Q_0$. For any normalized state $|\psi\rangle$, let $\langle \cdot \rangle_\psi$ denote the expectation value in the state $e^{-iH_b t}|\psi\rangle$. The Heisenberg equation gives $\frac{d}{dt}\langle Q^{(p)}(t)\rangle_\psi = \langle i[H_b, Q^{(p)}]\rangle_\psi$, so integrating and applying the operator-norm bound,
\begin{equation}
	\begin{aligned}
		\bigl|
		\langle Q^{(p)}(t)\rangle_\psi
		- \langle Q^{(p)}(0)\rangle_\psi
		\bigr|
		&\le \int_0^{|t|} ds\, \bigl\|[H_b, Q^{(p)}]\bigr\| \\
		&\le B_p |\varepsilon|^{p+1} |t|.
		\label{eq:app_nf_dressed_drift}
	\end{aligned}
\end{equation}
To pass from $Q^{(p)}$ back to $Q_0$, note that for any time $t$,
\begin{align}
	\bigl|\langle Q^{(p)}(t)\rangle_\psi
	- \langle Q_0(t)\rangle_\psi\bigr|
	\le \|Q^{(p)} - Q_0\|
	\le C_p |\varepsilon|,
\end{align}
where the first inequality uses $|\langle\psi|A|\psi\rangle|\le\|A\|$ for any normalized state, and the second uses Eq.~\eqref{eq:app_nf_close_final}. The same bound holds at $t=0$. The factor of $2$ in Eq.~\eqref{eq:app_nf_Q0_bound} below arises from applying this estimate at both endpoints and then using the triangle inequality
\begin{equation}
\begin{aligned}
	|\langle Q_0(t)\rangle_\psi - \langle Q_0(0)\rangle_\psi|
	\le&
	|\langle Q_0(t)\rangle_\psi - \langle Q^{(p)}(t)\rangle_\psi| \\
	&+ |\langle Q^{(p)}(t)\rangle_\psi - \langle Q^{(p)}(0)\rangle_\psi| \\
	&+ |\langle Q^{(p)}(0)\rangle_\psi - \langle Q_0(0)\rangle_\psi|.
\end{aligned}
\end{equation}
Bounding the first and third terms by $C_p|\varepsilon|$ each, and the middle term by Eq.~\eqref{eq:app_nf_dressed_drift}, gives
\begin{align}
	\bigl|
	\langle Q_0(t)\rangle_\psi - \langle Q_0(0)\rangle_\psi
	\bigr|
	\le 2 C_p |\varepsilon|
	+ B_p |\varepsilon|^{p+1} |t|.
	\label{eq:app_nf_Q0_bound}
\end{align}
This is Eq.~(\ref*{eq:asymptotic_memory_bound_main}). The two terms have distinct physical origins. The first, $2C_p|\varepsilon|$, is a static dressing error: it is the price of replacing $Q_0$ by the dressed charge $Q^{(p)}$, and it is nonzero even at $t=0$ because $Q^{(p)}\neq Q_0$. It vanishes as $\varepsilon\to0$, so it is small whenever the inter-dimer coupling is weak. The second term, $B_p|\varepsilon|^{p+1}|t|$, is the accumulated dynamical drift of the dressed charge under $H_b$; it grows linearly in time but with a coefficient suppressed by $|\varepsilon|^{p+1}$, which can be made arbitrarily small by increasing $p$.

\paragraph{Time-window statement.}
The bound~\eqref{eq:app_nf_Q0_bound} controls the drift for all times, but its content is sharpest on time windows that grow as a power of the inverse coupling $b^{-1}$. Such windows are physically relevant because they can be made arbitrarily long as $b\to0$, capturing the experimentally observable slow decay of the edge magnetization. Concretely, since $\varepsilon = b/2$, substituting into Eq.~\eqref{eq:app_nf_Q0_bound} and evaluating on the time window $0 \le t \le b^{-\alpha}$, the dynamical drift contributes at most $B_p (b/2)^{p+1} b^{-\alpha} = O(b^{p+1-\alpha})$. Hence for every $\alpha < p+1$, up to harmless numerical prefactors,
\begin{align}
	\sup_{0 \le t \le b^{-\alpha}}
	\bigl|
	\langle Q_0(t)\rangle_\psi - \langle Q_0(0)\rangle_\psi
	\bigr|
	\le C_p b + C_p b^{p+1-\alpha}.
	\label{eq:app_nf_limit_bound}
\end{align}
The condition $\alpha < p+1$ is precisely what is needed for $b^{p+1-\alpha} \to 0$ as $b \to 0$. Both terms on the right-hand side therefore vanish in the weak-coupling limit at fixed $\alpha$ and $p$, proving Eq.~(\ref*{eq:asymptotic_memory_limit_main}). Since $p$ can be chosen as large as desired (subject to the non-resonance condition), the time window $b^{-\alpha}$ over which the initial charge is retained can be made to grow as an arbitrarily large power of the inverse coupling.

Finally, if the chosen topological edge sector satisfies $|\langle Q_0(0)\rangle_\psi| = q_* > 0$, the reverse triangle inequality applied to Eq.~\eqref{eq:app_nf_Q0_bound} gives the lower bound
\begin{align}
	|\langle Q_0(t)\rangle_\psi|
	\ge
	q_* - 2C_p|\varepsilon| - B_p|\varepsilon|^{p+1}|t|,
	\label{eq:app_nf_lower_bound}
\end{align}
which is Eq.~(\ref*{eq:asymptotic_memory_lower_bound_main}). For $|\varepsilon|$ small enough and $t$ within the window $b^{-\alpha}$, the right-hand side is strictly positive, confirming that the physical boundary charge retains a nonzero expectation value throughout that window.

	\section{Proof of the Initial-State Comparison Bound}
	\label{app:initialstatecomparison}
	
	This appendix proves Proposition~\ref*{prop:mbtomb_initial_state_comparison}.
	
	\begin{proof}
		Define the Heisenberg-picture operators
		\begin{align}
			S_{1,b}^z(t)  &= e^{iH_b t}\,S_1^z\,e^{-iH_b t}, \\
			S_{1,0}^z(t)  &= e^{iH_0 t}\,S_1^z\,e^{-iH_0 t}.
		\end{align}
		For any initial density matrix $\eta_i$, write
		\begin{align}
			\langle S_1^z(t)\rangle_{\eta_i,b}
			&= \Tr\!\bigl[\eta_i\,S_{1,b}^z(t)\bigr], &
			\langle S_1^z(t)\rangle_{\eta_i,0}
			&= \Tr\!\bigl[\eta_i\,S_{1,0}^z(t)\bigr].
		\end{align}
		By the triangle inequality,
		\begin{equation}
			\begin{aligned}
				\bigl|
				\langle S_1^z(t)\rangle_{\rho_i,b}
				-
				\langle S_1^z(t)\rangle_{\sigma_i,b}
				\bigr|
				\le{}&
				\bigl|
				\langle S_1^z(t)\rangle_{\rho_i,b}
				-
				\langle S_1^z(t)\rangle_{\rho_i,0}
				\bigr| \\
				&+
				\bigl|
				\langle S_1^z(t)\rangle_{\rho_i,0}
				-
				\langle S_1^z(t)\rangle_{\sigma_i,0}
				\bigr| \\
				&+
				\bigl|
				\langle S_1^z(t)\rangle_{\sigma_i,0}
				-
				\langle S_1^z(t)\rangle_{\sigma_i,b}
				\bigr|.
			\end{aligned}
			\label{eq:app_is_triangle}
		\end{equation}
		
		\paragraph{First and third terms.}
		The operator bound in Proposition~\ref*{prop:finite_time_stability_sptomb} and the estimate $|\Tr(\eta A)|\le\|A\|$ for any density matrix $\eta$ and bounded operator $A$ immediately give, for $\eta_i\in\{\rho_i,\sigma_i\}$,
		\begin{align}
			\bigl|
			\langle S_1^z(t)\rangle_{\eta_i,b}
			-
			\langle S_1^z(t)\rangle_{\eta_i,0}
			\bigr|
			\le
			\frac{|b|}{2}(2+|J_z^f|)\,t.
			\label{eq:app_is_opbound}
		\end{align}
		This bound is uniform in time and holds for both $\rho_i$ and $\sigma_i$.
		
		\paragraph{Middle term.}
		At $b=0$ the final Hamiltonian factorizes into independent dimers $H_0 = \sum_m h_{2m-1,2m}$, so $S_{1,0}^z(t)$ is supported only on sites $1$ and $2$. Using Eq.~(\ref*{eq:mbtomb_b0_reduction}), one may write
		\begin{align}
			\langle S_1^z(t)\rangle_{\rho_i,0}
			-
			\langle S_1^z(t)\rangle_{\sigma_i,0}
			=
			\Tr_{12}\!\Bigl[
			\bigl(\rho_{12}^{(i)}-\sigma_{12}^{(i)}\bigr)\,
			e^{ih_{12}t}\,S_1^z\,e^{-ih_{12}t}
			\Bigr].
			\label{eq:app_is_reduced}
		\end{align}
		The trace inequality $|\Tr(AB)|\le\|A\|_1\,\|B\|$ then gives
		\begin{equation}
			\begin{aligned}
				\Bigl|
				\Tr_{12}\!\Bigl[
				\bigl(\rho_{12}^{(i)}-\sigma_{12}^{(i)}\bigr)
				&\,e^{ih_{12}t}\,S_1^z\,e^{-ih_{12}t}
				\Bigr]
				\Bigr| \\
				&\le
				\bigl\|\rho_{12}^{(i)}-\sigma_{12}^{(i)}\bigr\|_1
				\,
				\bigl\|e^{ih_{12}t}\,S_1^z\,e^{-ih_{12}t}\bigr\|.
			\end{aligned}
		\end{equation}
		Unitary invariance of the operator norm and $\|S_1^z\|=\tfrac{1}{2}$ yield
		\begin{align}
			\bigl|
			\langle S_1^z(t)\rangle_{\rho_i,0}
			-
			\langle S_1^z(t)\rangle_{\sigma_i,0}
			\bigr|
			\le
			\frac{1}{2}
			\bigl\|\rho_{12}^{(i)}-\sigma_{12}^{(i)}\bigr\|_1.
			\label{eq:app_is_tracebound}
		\end{align}
		
		\paragraph{Conclusion.}
		Substituting Eqs.~\eqref{eq:app_is_opbound} and \eqref{eq:app_is_tracebound} into Eq.~\eqref{eq:app_is_triangle} gives
		\begin{equation}
			\begin{aligned}
				\bigl|
				\langle S_1^z(t)\rangle_{\rho_i,b}
				-\,&
				\langle S_1^z(t)\rangle_{\sigma_i,b}
				\bigr| \\
				&\le
				\frac{1}{2}
				\bigl\|\rho_{12}^{(i)}-\sigma_{12}^{(i)}\bigr\|_1
				+
				|b|(2+|J_z^f|)\,t,
			\end{aligned}
		\end{equation}
		which is Eq.~(\ref*{eq:mbtomb_initial_state_bound}). The bound makes precise that, near the dimer limit, changing the interacting preparation affects the edge dynamics only through the first-dimer reduced state, up to a controlled $O(bt)$ correction generated by weak inter-dimer leakage.
	\end{proof}
	
	\section{Dimer spectrum and the isotropic leakage channel}
	\label{app:isotropic_channel_derivation}
	This appendix gives the explicit spin-basis derivation of the dimer spectrum and of the resonant leakage channel used in Sec.~\ref*{sec:isotropic_resonance}. We use the main-text convention $a=J_1^f=2J_1^{xy,f}$, so the strong final intracell dimer Hamiltonian on sites $1,2$ is
	\begin{align}
		h_{12}
		=
		\frac{a}{2}
		\left(
		\sigma_1^x\sigma_2^x+\sigma_1^y\sigma_2^y
		+J_z^f\sigma_1^z\sigma_2^z
		\right).
		\label{eq:app_iso_h12}
	\end{align}
	The fully polarized states are
	\begin{align}
		|t_+\rangle=|\uparrow\uparrow\rangle,\qquad
		|t_-\rangle=|\downarrow\downarrow\rangle .
		\label{eq:app_iso_tpm}
	\end{align}
	The $XY$ term annihilates these states, while the $ZZ$ term gives eigenvalue $+1$. Hence
	\begin{align}
		E_{t_+}=E_{t_-}=\frac{a}{2}J_z^f .
		\label{eq:app_iso_Etpm}
	\end{align}
	In the zero-magnetization dimer subspace spanned by $|\uparrow\downarrow\rangle$ and $|\downarrow\uparrow\rangle$, the $ZZ$ term gives eigenvalue $-1$ on both basis states, while
	\begin{align}
		\left(\sigma_1^x\sigma_2^x+\sigma_1^y\sigma_2^y\right)
		|\uparrow\downarrow\rangle
		=
		2|\downarrow\uparrow\rangle ,
		\nonumber\\
		\left(\sigma_1^x\sigma_2^x+\sigma_1^y\sigma_2^y\right)
		|\downarrow\uparrow\rangle
		=
		2|\uparrow\downarrow\rangle .
		\label{eq:app_iso_xy_action}
	\end{align}
	Therefore
	\begin{align}
		h_{12}\big|_{S^z_{\rm tot}=0}
		=
		\begin{pmatrix}
			-\frac{a}{2}J_z^f & a\\
			a & -\frac{a}{2}J_z^f
		\end{pmatrix},
		\label{eq:app_iso_matrix}
	\end{align}
	whose eigenstates are
	\begin{align}
		|t_0\rangle
		=
		\frac{|\uparrow\downarrow\rangle+|\downarrow\uparrow\rangle}{\sqrt{2}},
		\qquad
		E_{t_0}
		=
		a-\frac{a}{2}J_z^f ,
		\label{eq:app_iso_Et0}\\
		|s\rangle
		=
		\frac{|\uparrow\downarrow\rangle-|\downarrow\uparrow\rangle}{\sqrt{2}},
		\qquad
		E_s
		=
		-a-\frac{a}{2}J_z^f .
		\label{eq:app_iso_Es}
	\end{align}
	Equations~\eqref{eq:app_iso_Etpm} and~\eqref{eq:app_iso_Et0} show that
	\begin{align}
		E_{t_+}=E_{t_0}=E_{t_-}
		\quad
		\Longleftrightarrow
		\quad
		J_z^f=1 .
		\label{eq:app_iso_triplet_degeneracy}
	\end{align}
	Thus $J_z^f=1$ is the isotropic point, where the dimer Hamiltonian is proportional to $\sigma_1^x\sigma_2^x+\sigma_1^y\sigma_2^y+\sigma_1^z\sigma_2^z$ and the triplet is exactly degenerate. Now consider two neighboring dimers, $(1,2)$ and $(3,4)$, coupled by the weak inter-dimer bond between sites $2$ and $3$. Its $XY$ part is
	\begin{align}
		V_{23}^{xy}
		=
		\frac{b}{2}
		\left(
		\sigma_2^x\sigma_3^x+\sigma_2^y\sigma_3^y
		\right),
		\qquad
		b=J_2^f .
		\label{eq:app_iso_Vxy}
	\end{align}
	This operator flips opposite spins on sites $2,3$. In particular,
	\begin{align}
		V_{23}^{xy}
		|t_+\rangle_{12}|t_-\rangle_{34}
		=
		b|\uparrow_1\downarrow_2\uparrow_3\downarrow_4\rangle .
		\label{eq:app_iso_action_channel}
	\end{align}
	The state on the right has nonzero overlap with
	\begin{align}
		|t_0\rangle_{12}|t_0\rangle_{34}
		=
		\frac{1}{2}
		\left(
		|\uparrow\downarrow\uparrow\downarrow\rangle
		+|\uparrow\downarrow\downarrow\uparrow\rangle
		+|\downarrow\uparrow\uparrow\downarrow\rangle
		+|\downarrow\uparrow\downarrow\uparrow\rangle
		\right),
		\label{eq:app_iso_t0t0_expand}
	\end{align}
	and therefore the inter-dimer $XY$ term connects
	\begin{align}
		|t_+\rangle_{12}|t_-\rangle_{34}
		\longleftrightarrow
		|t_0\rangle_{12}|t_0\rangle_{34}
		\label{eq:app_iso_channel}
	\end{align}
	with a nonzero matrix element. The detuning of this channel is
	\begin{align}
		\Delta_{\rm iso}
		&=
		2E_{t_0}-E_{t_+}-E_{t_-}
		\nonumber\\
		&=
		2\left(a-\frac{a}{2}J_z^f\right)
		-\frac{a}{2}J_z^f
		-\frac{a}{2}J_z^f
		\nonumber\\
		&=
		2a(1-J_z^f).
		\label{eq:app_iso_detuning}
	\end{align}
	Thus the channel is exactly resonant at $J_z^f=1$. Combining Eq.~\eqref{eq:app_iso_detuning} with the resonant limit of first-order perturbation theory in Eq.~\eqref{eq:app_prob_resonant} gives
	\begin{align}
		P_{\beta\leftarrow\alpha}^{(1)}(t)
		=
		\frac{b^2}{4}|W_{\beta\alpha}|^2t^2 ,
		\label{eq:app_iso_resonant_growth}
	\end{align}
	rather than the bounded oscillatory form obtained at fixed nonzero detuning. Thus the isotropic point opens a leakage channel whose transition amplitude adds coherently in time, giving parametrically stronger early-time leakage from the boundary dimer. This is the microscopic origin of the isotropic leakage resonance discussed in the main text.
	
	\section{Free evolution from an interacting initial state}
	\label{app:mbtosp_free_reduction}
	
	This Appendix proves the exact reduction used in Sec.~\ref*{sec:mbtosp}. The final Hamiltonian is quadratic,
	\begin{align}
		H_f=\sum_{m,n=1}^{N}c_m^\dagger h_{mn}^f c_n,
		\label{eq:app_mbtosp_Hf}
	\end{align}
	where $h_f$ is the final single-particle SSH matrix. Define
	\begin{align}
		U_f(t)=e^{-ih_f t}.
		\label{eq:app_mbtosp_U}
	\end{align}
	The Heisenberg equation gives the linear operator evolution
	\begin{align}
		c_j(t)
		=
		e^{iH_f t}c_j e^{-iH_f t}
		=
		\sum_{m=1}^{N}[U_f(t)]_{jm}c_m .
		\label{eq:app_mbtosp_c_evol}
	\end{align}
	Therefore
	\begin{align}
		n_j(t)
		=
		c_j^\dagger(t)c_j(t)
		=
		\sum_{m,n=1}^{N}
		[U_f(t)]_{jm}^*
		[U_f(t)]_{jn}
		c_m^\dagger c_n .
		\label{eq:app_mbtosp_density}
	\end{align}
	Taking the expectation value in an arbitrary initial state $|\Psi_i\rangle$ gives
	\begin{align}
		\langle n_j(t)\rangle
		=
		\sum_{m,n=1}^{N}
		[U_f(t)]_{jm}^*
		(C_i)_{mn}
		[U_f(t)]_{jn},
		\label{eq:app_mbtosp_density_expect}
	\end{align}
	where
	\begin{align}
		(C_i)_{mn}
		=
		\langle\Psi_i|c_m^\dagger c_n|\Psi_i\rangle .
		\label{eq:app_mbtosp_Ci}
	\end{align}
	Here $C_i$ denotes the site-basis one-body density matrix. Later in this Appendix, $\mathcal{C}$ denotes the same one-body density matrix represented in the final eigenmode basis.
	In matrix form, with this convention for $C_i$,
	\begin{align}
		\langle n_j(t)\rangle
		=
		\left[
		U_f^*(t)C_iU_f^T(t)
		\right]_{jj}.
		\label{eq:app_mbtosp_C_evol}
	\end{align}
	No Wick decomposition has been used, so Eq.~\eqref{eq:app_mbtosp_C_evol} is valid even when $|\Psi_i\rangle$ is an interacting many-body ground state. Since $S_j^z=n_j-1/2$, this proves Eq.~(\ref*{eq:mbtosp_exact_C}).
	
	We now derive the final-eigenbasis expression. Let $\{|\mu\rangle\}_{\mu=1}^{N}$ be an orthonormal eigenbasis of $h_f$,
	\begin{align}
		h_f|\mu\rangle=\varepsilon_\mu|\mu\rangle,
		\qquad
		\varphi_\mu(j)=\langle j|\mu\rangle .
		\label{eq:app_mbtosp_eigenbasis}
	\end{align}
	Define the corresponding normal-mode operators by
	\begin{align}
		d_\mu^\dagger=\sum_{j=1}^{N}\varphi_\mu(j)c_j^\dagger,
		\qquad
		d_\mu=\sum_{j=1}^{N}\varphi_\mu^*(j)c_j .
		\label{eq:app_mbtosp_d_def}
	\end{align}
	Then
	\begin{align}
		c_j^\dagger=\sum_{\mu=1}^{N}\varphi_\mu^*(j)d_\mu^\dagger,
		\qquad
		c_j=\sum_{\mu=1}^{N}\varphi_\mu(j)d_\mu,
		\label{eq:app_mbtosp_inverse}
	\end{align}
	and
	\begin{align}
		d_\mu(t)=e^{-i\varepsilon_\mu t}d_\mu .
		\label{eq:app_mbtosp_d_evol}
	\end{align}
	Substituting Eq.~\eqref{eq:app_mbtosp_inverse} into $n_j(t)$ gives
	\begin{align}
		\langle n_j(t)\rangle
		=
		\sum_{\mu,\nu=1}^{N}
		\varphi_\mu^*(j)\varphi_\nu(j)
		e^{i(\varepsilon_\mu-\varepsilon_\nu)t}
		\langle d_\mu^\dagger d_\nu\rangle_i .
		\label{eq:app_mbtosp_spectral}
	\end{align}
	For the interacting initial state, the final-mode representation of the one-body density matrix is
	\begin{align}
		\mathcal{C}_{\mu\nu}^{\rm MB}(J_z^i)
		&=
		\langle\Psi_i(J_z^i)|d_\mu^\dagger d_\nu|\Psi_i(J_z^i)\rangle
		\nonumber\\
		&=
		\sum_{m,n=1}^{N}
		\varphi_\mu(m)
		[C_i^{\rm MB}(J_z^i)]_{mn}
		\varphi_\nu^*(n).
		\label{eq:app_mbtosp_mathcalC}
	\end{align}
	Equation~\eqref{eq:app_mbtosp_spectral} then proves Eq.~(\ref*{eq:mbtosp_spectral_formula}).
	
	The infinite-time average follows by integrating Eq.~\eqref{eq:app_mbtosp_spectral} over $[0,T]$ and taking $T\to\infty$:
	\begin{align}
		\lim_{T\to\infty}
		\frac{1}{T}\int_0^T dt\,
		e^{i(\varepsilon_\mu-\varepsilon_\nu)t}
		=
		\begin{cases}
			1,& \varepsilon_\mu=\varepsilon_\nu,\\
			0,& \varepsilon_\mu\neq\varepsilon_\nu.
		\end{cases}
		\label{eq:app_mbtosp_time_average_phase}
	\end{align}
	This gives Eq.~(\ref*{eq:mbtosp_diag_ensemble}). If the final single-particle spectrum is nondegenerate, only terms with $\mu=\nu$ remain, giving Eq.~(\ref*{eq:mbtosp_diag_non_degenerate}).
	
	Finally, let $C_i^{\rm MB}$ and $C_i^{\rm SP}$ be the site-basis initial one-body density matrices for the MB and SP preparations at the same $\delta_i$. Their difference in the evolved edge magnetization is
	\begin{align}
		\delta S_1^z(t)
		=
		\left[
		U_f^*(t)(C_i^{\rm MB}-C_i^{\rm SP})U_f^T(t)
		\right]_{11}.
		\label{eq:app_mbtosp_deltaS}
	\end{align}
	Define the normalized vector $\eta_1(t)$ by
	\begin{align}
		[\eta_1(t)]_m=[U_f(t)]_{1m},
		\qquad
		\|\eta_1(t)\|=1.
		\label{eq:app_mbtosp_eta}
	\end{align}
	Then Eq.~\eqref{eq:app_mbtosp_deltaS} is equivalently
	\begin{align}
		\delta S_1^z(t)
		=
		\eta_1^\dagger(t)
		(C_i^{\rm MB}-C_i^{\rm SP})
		\eta_1(t).
		\label{eq:app_mbtosp_deltaS_vector}
	\end{align}
	By the definition of the operator norm,
	\begin{align}
		|\delta S_1^z(t)|
		\le
		\|C_i^{\rm MB}-C_i^{\rm SP}\|.
		\label{eq:app_mbtosp_delta_bound}
	\end{align}
	This proves the uniform-in-time bound in Eq.~(\ref*{eq:mbtosp_difference_bound}).

	\section{Numerically exact validation of interacting analytical results}
	\label{app:interacting_benchmarks}
	This appendix benchmarks the interacting analytical statements used in Secs.~\ref*{sec:sptomb}, \ref*{sec:mbtosp}, and \ref*{sec:mbtomb}. For the non-interacting SSH results, numerical validation is already given in Appendix~\ref{app:sp_supplementary}. Here we focus only on the interacting statements that enter the later physical interpretation.
	
	\subsection{Benchmark of the exact free reduction for the MB to SP quench}
	\label{app:mbtosp_benchmark}
	In this subsection we benchmark Proposition~\ref*{prop:mbtosp_exact_reduction}, Corollaries~\ref*{cor:mbtosp_difference} and \ref*{cor:mbtosp_diagonal_ensemble} for quenches from an interacting topological initial state to a non-interacting final Hamiltonian whose ground state lies on the trivial side of the equilibrium phase diagram.
	All calculations are performed in the fixed total zero-magnetization sector, $S_{\rm tot}^z=0$, which is equivalent after the Jordan--Wigner transformation to half-filling. The initial and final dimerizations are fixed to the main-text values $\delta_i=-0.95$ and $\delta_f=+0.95$. For the exact real-time benchmark against full many-body evolution we use a finite chain of length $N=22$, for which the half-filled Hilbert space remains small enough for exact diagonalization while still resolving the edge-localized initial state.
	
	The first test is the exact reduction itself. Figure~\ref{fig:mbtosp_benchmark_real_time}(a) compares the edge magnetization obtained from direct many-body evolution under the free final Hamiltonian with the reduced one-body formula in Eq.~(\ref*{eq:mbtosp_exact_C}). The two curves are visually indistinguishable on the scale of the plot, and the lower panel shows that the absolute difference remains at numerical roundoff level throughout the full time window. This directly verifies that, once $J_z^f=0$, the MB to SP quench is completely determined by the interacting initial one-body density matrix and requires no Gaussian assumption. Figure~\ref{fig:mbtosp_benchmark_real_time}(b) benchmarks the exact comparison identity in Corollary~\ref*{cor:mbtosp_difference}. The direct difference $S_{1,{\rm MB}\to{\rm SP}}^z(t;J_z^i)-S_{1,{\rm SP}\to{\rm SP}}^z(t)$ is compared with the propagated site-1 matrix element $\left[U_f^*(t)\Delta C_i(J_z^i)U_f^T(t)\right]_{11}$, where the subscript $11$ denotes the diagonal matrix element at the edge site $j=1$. Their agreement verifies Eq.~(\ref*{eq:mbtosp_difference_identity}) over the full time interval, while the dotted horizontal lines at $\pm\|\Delta C_i(J_z^i)\|$ confirm the uniform bound in Eq.~(\ref*{eq:mbtosp_difference_bound}). This isolates the origin of the deviation from the SP to SP reference and shows that it is entirely encoded in the interacting preparation through $\Delta C_i(J_z^i)$.
	
	The remaining two benchmarks test the final-Hamiltonian control of the spectral and stationary components. Figure~\ref{fig:mbtosp_benchmark_spectral}(a) shows the Fourier spectra of $S_1^z(t)$ for several values of $J_z^i$ as functions of the Bohr frequency $\omega=\varepsilon_\mu-\varepsilon_\nu$ in units with $\hbar=1$. As the initial interaction is varied, the spectral weights are redistributed, but the peak locations remain fixed by the same final SSH Hamiltonian. The dotted vertical lines do not mark all formal energy differences of the finite system. Instead, they mark only the dominant edge-active frequencies selected from the grouped complex coefficients in Eq.~(\ref*{eq:mbtosp_spectral_formula}), namely those frequencies for which the summed coefficient $\sum_{\mu,\nu:\,\varepsilon_\mu-\varepsilon_\nu=\omega}\varphi_\mu^*(1)\varphi_\nu(1)\mathcal{C}^{\rm MB}_{\mu\nu}(J_z^i)$ remains appreciable after interference at fixed $\omega$ is taken into account. This is the precise frequency-domain content of Eq.~(\ref*{eq:mbtosp_spectral_formula}). Figure~\ref{fig:mbtosp_benchmark_spectral}(b) shows the corresponding running time average together with the free diagonal-ensemble prediction from Corollary~\ref*{cor:mbtosp_diagonal_ensemble}. The running averages approach the predicted stationary values for every interacting initial state considered, confirming that the long-time average is fixed by the free diagonal ensemble of the final Hamiltonian, with the dependence on $J_z^i$ entering only through the initial one-body matrix elements in the final eigenbasis. Taken together, these benchmarks validate the four numerical diagnostics identified in Sec.~\ref*{sec:mbtosp}: the exact real-time reduction to one-body data, the exact and uniformly bounded comparison with the SP to SP reference, the fact that the observable oscillation frequencies are fixed by the final non-interacting SSH Hamiltonian and selected by edge weight together with initial coherence, and the fact that the stationary value is fixed by the corresponding free diagonal ensemble.
	
	\begin{figure*}[t]
		\centering
		\begin{subfigure}[t]{0.48\textwidth}
			\includegraphics[width=\linewidth]{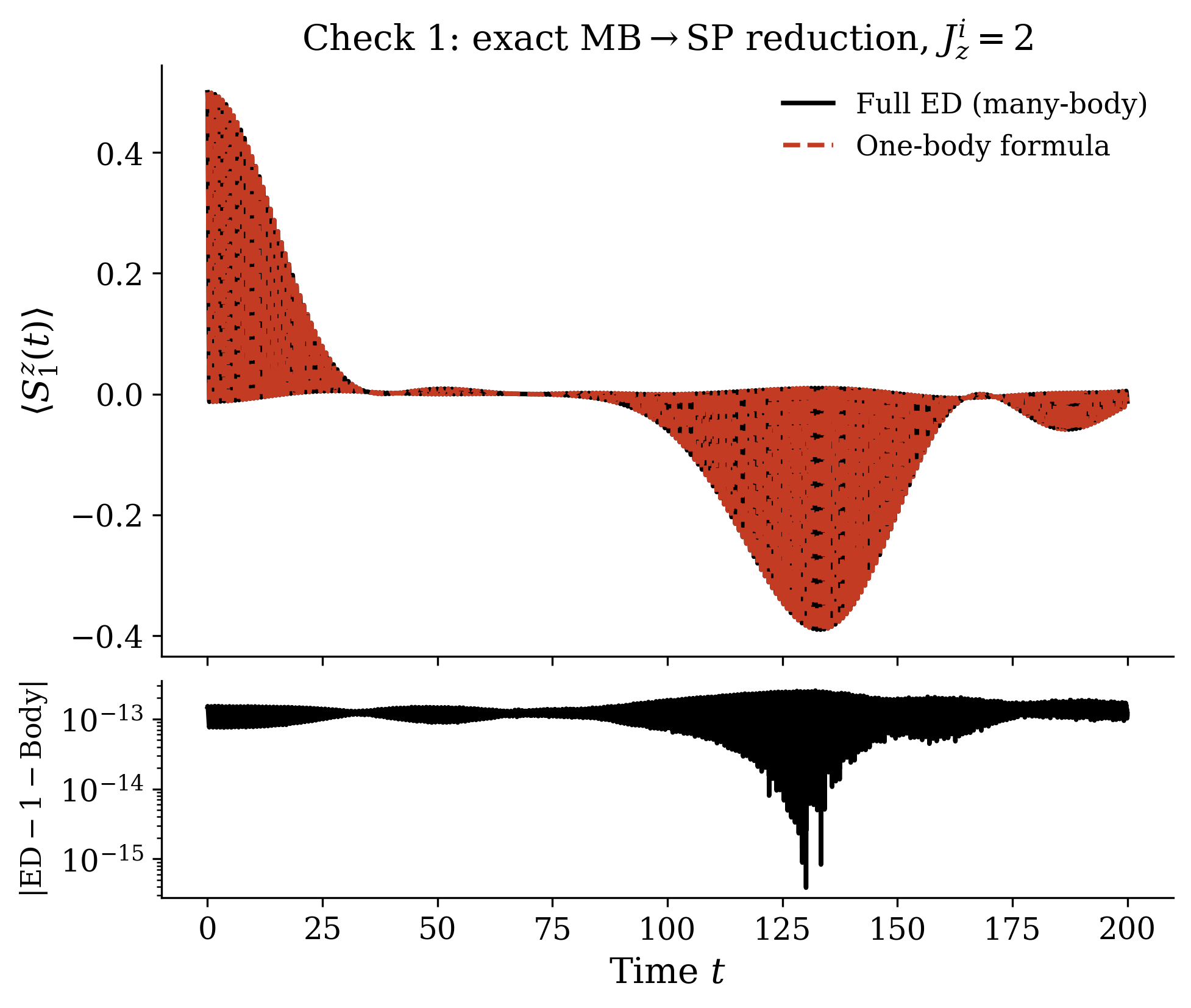}
			\caption{Direct many-body evolution versus the exact one-body reduction in Eq.~(\ref*{eq:mbtosp_exact_C}). The lower panel shows the pointwise absolute error.}
		\end{subfigure}\hfill
		\begin{subfigure}[t]{0.48\textwidth}
			\includegraphics[width=\linewidth]{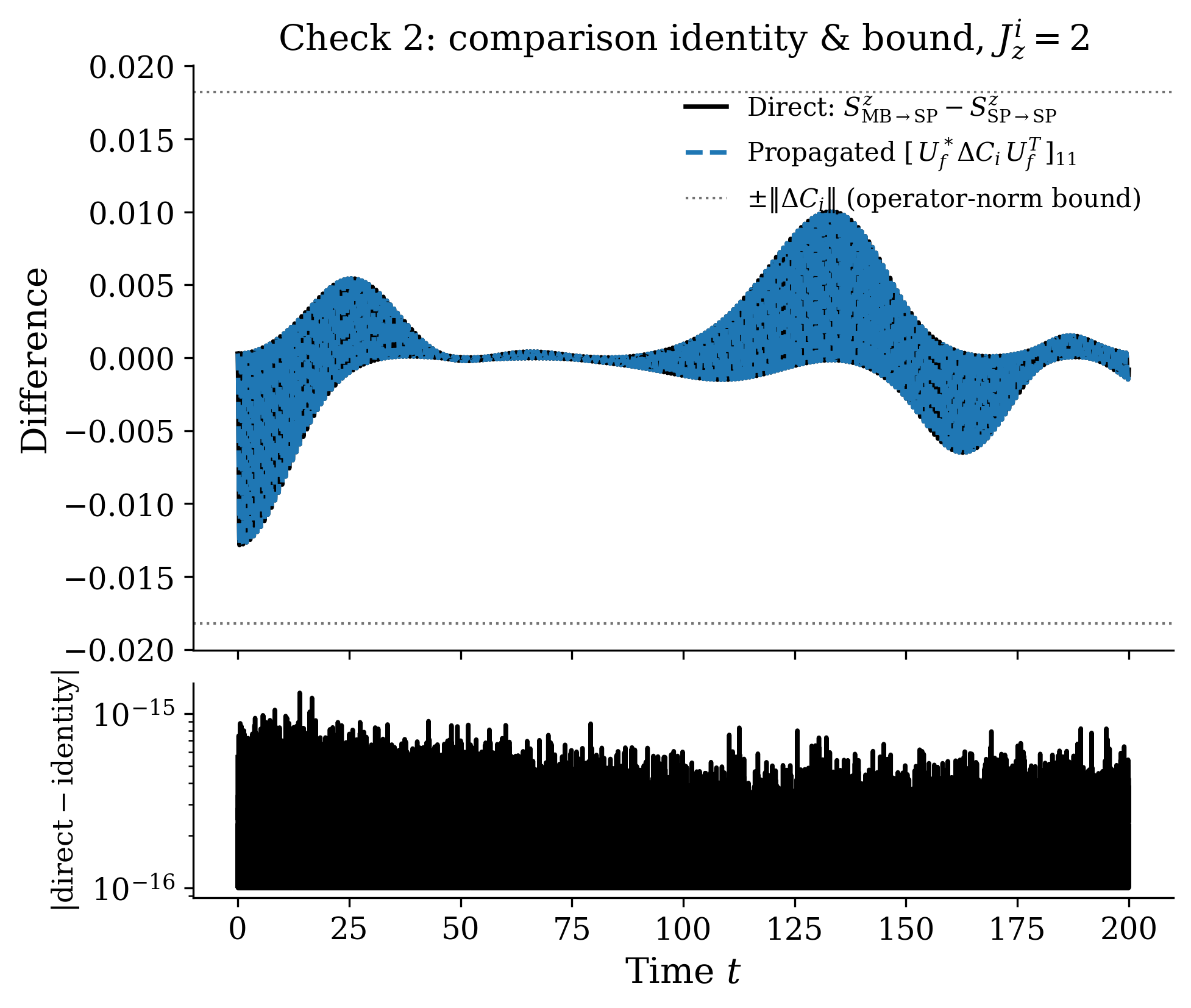}
			\caption{Benchmark of Eq.~(\ref*{eq:mbtosp_difference_identity}) and the uniform bound in Eq.~(\ref*{eq:mbtosp_difference_bound}). The propagated quantity is the site-1 matrix element $\left[U_f^*(t)\Delta C_i(J_z^i)U_f^T(t)\right]_{11}$.}
		\end{subfigure}
		\caption{Real-time validation of the MB to SP reduction for $\delta_i=-0.95$, $\delta_f=+0.95$, and fixed $S_{\rm tot}^z=0$ for $N=22$. Panel (a) compares the exact many-body dynamics with the reduced one-body prediction and shows roundoff-level disagreement only. Panel (b) compares the direct difference from the SP to SP reference with the propagated site-1 matrix element of $\Delta C_i(J_z^i)$ and shows the corresponding operator-norm bound $\pm\|\Delta C_i(J_z^i)\|$.}
		\label{fig:mbtosp_benchmark_real_time}
	\end{figure*}
	
	\begin{figure*}[t]
		\centering
		\begin{subfigure}[t]{0.48\textwidth}
			\includegraphics[width=\linewidth]{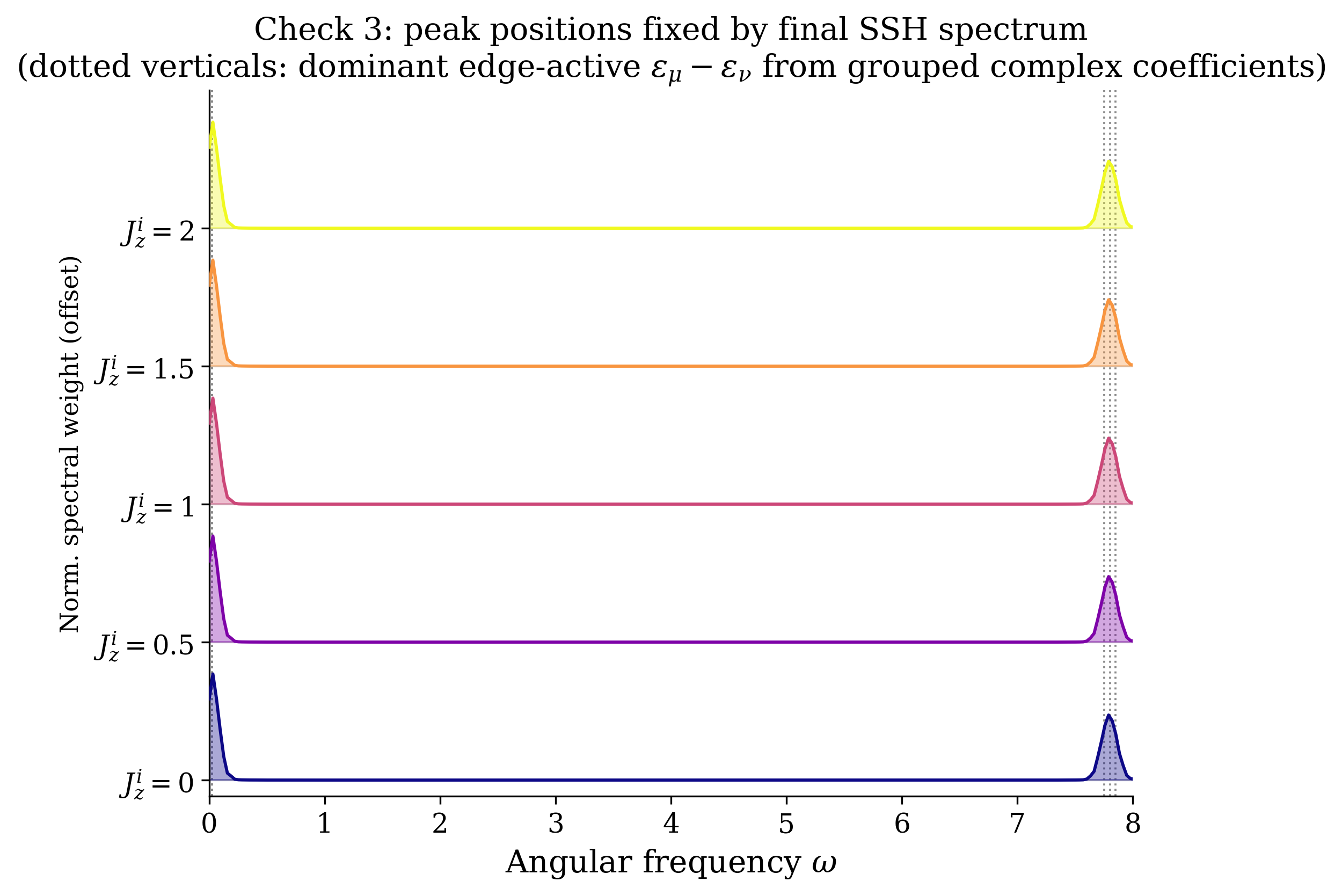}
			\caption{Fourier spectra for several interacting initial states as functions of the frequency $\omega=\varepsilon_\mu-\varepsilon_\nu$ with $\hbar=1$. The dotted vertical lines mark only the dominant edge-active frequencies selected from the grouped complex coefficients in Eq.~(\ref*{eq:mbtosp_spectral_formula}).}
		\end{subfigure}\hfill
		\begin{subfigure}[t]{0.48\textwidth}
			\includegraphics[width=\linewidth]{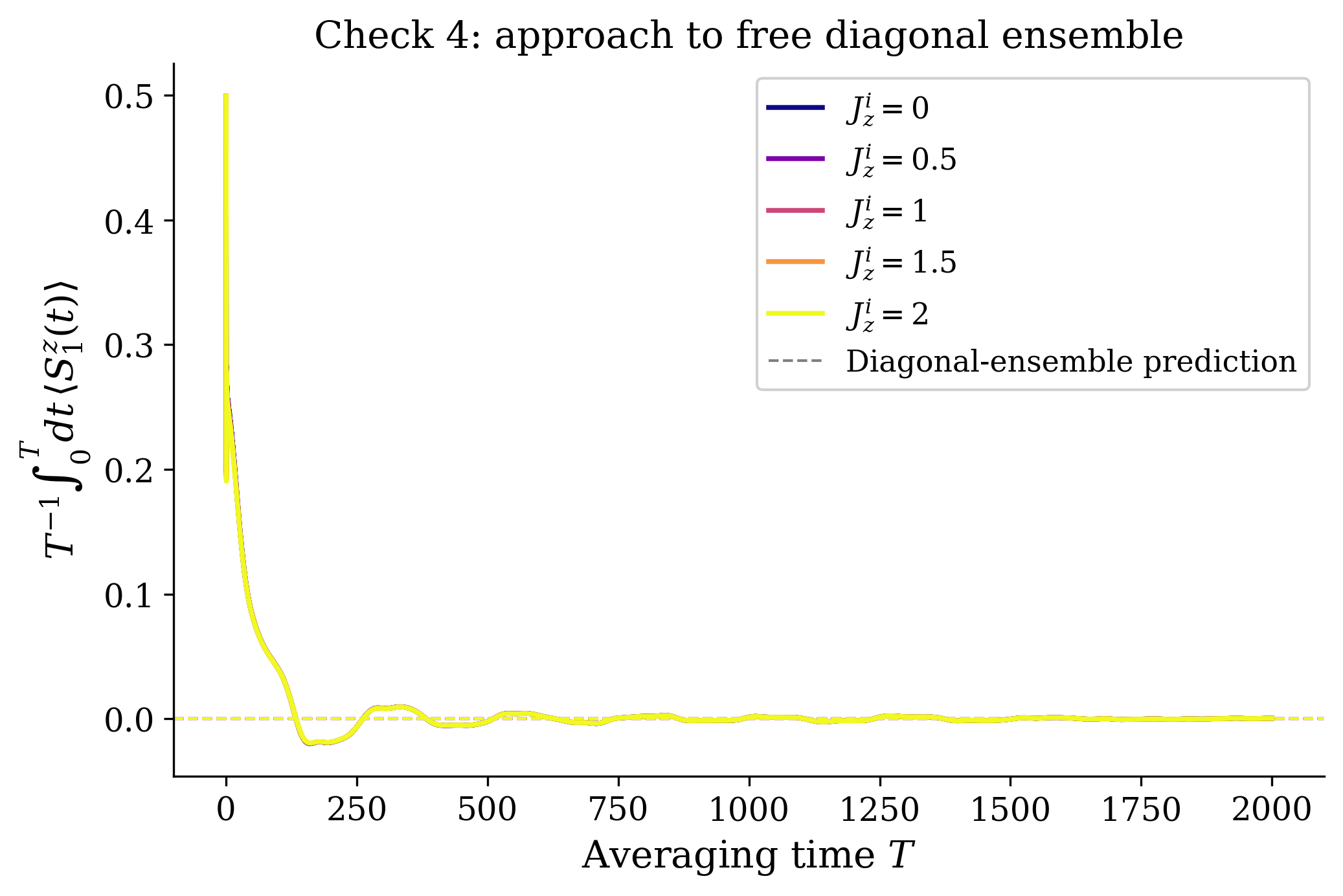}
			\caption{Running averages and free diagonal-ensemble predictions from Eq.~(\ref*{eq:mbtosp_diag_non_degenerate}).}
		\end{subfigure}
		\caption{Spectral and long-time diagnostics for the MB to SP quench for $\delta_i=-0.95$, $\delta_f=+0.95$, and fixed $S_{\rm tot}^z=0$ for $N=22$. Panel (a) shows that varying $J_z^i$ changes spectral weights but not the observable frequency locations fixed by the final SSH Hamiltonian. The guide lines are restricted to the dominant frequencies that remain after edge weighting and interference at fixed $\omega$ are taken into account. Panel (b) shows the approach of the running average to the free diagonal-ensemble value predicted by Corollary~\ref*{cor:mbtosp_diagonal_ensemble}.}
		\label{fig:mbtosp_benchmark_spectral}
	\end{figure*}

	\subsection{Exact numerical checks near the decoupled-dimer limit}
	\label{app:dimer_benchmarks}
	This subsection collects the exact numerical checks used to validate the six analytical statements that underpin the controlled perturbative analysis near the strongly dimerized limit in Secs.~\ref*{sec:sptomb} and~\ref*{sec:mbtomb}.
	All calculations are performed in the fixed total zero-magnetization sector, $S_{\rm tot}^z = 0$, which is equivalent after the Jordan--Wigner transformation to half-filling; the propositions themselves are operator-norm identities that hold for any normalized state, but the sector is fixed by the conservation law of both $H_i$ and $H_f$ and by the physical quench protocol.
	The initial and final dimerizations are fixed to the main-text values $\delta_i = -0.95$ and $\delta_f = +0.95$, giving, via Eq.~\eqref{eq:J_param}, the final SSH one-body hopping amplitudes $a = J_1^f = 2(1+\delta_f) = 3.90$ (intracell, strong bond) and $b = J_2^f = 2(1-\delta_f) = 0.10$ (intercell, weak bond), so that the perturbation parameter is $b/a \approx 0.026$. These are the same $a$ and $b$ that appear in Theorem~\ref*{thm:free_decay}, Corollary~\ref*{cor:tau_scaling}, and Appendix~\ref{app:proof}; they denote SSH one-body hopping amplitudes and satisfy $a > b > 0$ for $\delta_f > 0$.
	The final interaction strength is set to $J_z^f = 1$, the isotropic point discussed in Sec.~\ref*{sec:isotropic_resonance}.
	All five checks are performed by exact diagonalization on a chain of length $L$ in the half-filled sector, for which the Hilbert space has dimension $\binom{L}{L/2}$, and all observables are evaluated to numerical precision.
	\paragraph*{Initial-state preparation and edge-sector selection.}
	The topological ground state in the $S_{\rm tot}^z = 0$ sector of a finite open chain is two-fold degenerate at leading order in the thermodynamic limit: the two nearly zero-energy edge modes can be occupied in two distinct ways, giving rise to two degenerate (or exponentially near-degenerate) many-body ground states with opposite signs of $\langle S_1^z \rangle$ and $\langle S_N^z \rangle$. To select a definite member of this doublet and thereby obtain an initial state with $\langle S_1^z(0) \rangle > 0$ (left edge occupied), a small staggered boundary field
	\begin{equation}
		V_{\rm pin} = -\epsilon\bigl(2S_1^z - 2S_N^z\bigr), \qquad \epsilon = 10^{-6},
		\label{eq:edgepin}
	\end{equation}
	is added to the initial Hamiltonian $H_i$ before diagonalization. The field $V_{\rm pin}$ lifts the edge degeneracy by an amount $2\epsilon \ll J_2^i$, leaving the bulk spectrum and all other properties of the ground state unchanged to within numerical precision. The edge magnetization of the selected initial state satisfies $|\langle S_1^z(0) \rangle| > 0.25$ for all parameter values used here, which is verified as a runtime check; if this condition fails, the initialization is rejected. The boundary pin is applied \emph{only} to the initial Hamiltonians $H_i^{\rm SP}$ (with $J_z^i = 0$) and $H_i^{\rm MB}$ (with $J_z^i \neq 0$) used for state preparation. The evolution Hamiltonians $H_b = H_0 + bW$ and $H_0$, which enter the analytical bounds and are the objects whose dynamics the checks verify, are strictly unpinned ($\epsilon = 0$). All analytical claims and the numerical bounds reported below therefore pertain to the unpinned evolution Hamiltonians, and the pin affects only the sector of the initial state. Figure~\ref{fig:dimer_benchmarks} summarizes the results.
	
	\paragraph*{Check~1: finite-time stability bound.}
	Proposition~\ref*{prop:finite_time_stability_sptomb} states that for all $t \geq 0$,
	\begin{align}
		\bigl|\langle S_{1,b}^z(t) \rangle_\psi
		- \langle S_{1,0}^z(t) \rangle_\psi\bigr|
		\leq \frac{b}{2}\bigl(2 + |J_z^f|\bigr)t,
		\label{eq:check1_bound}
	\end{align}
	where $S_{1,b}^z(t)$ and $S_{1,0}^z(t)$ are the edge magnetizations evolved under $H_b$ and $H_0$, respectively, starting from the SP topological initial state.
	Panel~(a) of Fig.~\ref{fig:dimer_benchmarks} shows the numerically computed difference $|\langle S_{1,b}^z(t)\rangle - \langle S_{1,0}^z(t)\rangle|$ alongside the analytical bound $\frac{b}{2}(2 + J_z^f)t$ over the window $0 \leq t \leq 5$. The difference lies strictly below the bound throughout the plotted interval, confirming Eq.~\eqref{eq:check1_bound} at the chosen parameters. The observed departure remains comfortably below the bound, consistent with the weakness of the inter-dimer coupling in the strongly dimerized regime at $b = J_2^f = 0.10$.
	
	\paragraph*{Check~2: boundary-dimer leakage bound for SP$\to$MB.}
	Proposition~\ref*{prop:boundary_dimer_continuity} gives
	\begin{align}
		\bigl|\langle M_{12}^z(t)\rangle - \langle M_{12}^z(0)\rangle\bigr|
		\leq |b|\,t,
		\label{eq:check2_leakage}
	\end{align}
	where $M_{12}^z = S_1^z + S_2^z$ and the evolution is under $H_b$ starting from the SP topological initial state.
	Panel~(d) of Fig.~\ref{fig:dimer_benchmarks} shows the numerically computed leakage $|\langle M_{12}^z(t)\rangle - \langle M_{12}^z(0)\rangle|$ alongside the bound $bt$.
	The data lie strictly below the bound for all $t$ in the window, confirming Eq.~\eqref{eq:check2_leakage}.
	The $z$-magnetization flows out of the boundary dimer only through the weak inter-dimer XY bond connecting sites $2$ and $3$, and the bound captures precisely this channel.
	
	\paragraph*{Check~3: first-order transition probability.}
	Equation~(\ref*{eq:transition_prob_main}) gives the leading-order leakage probability
	\begin{align}
		\mathcal{P}_{\alpha\to\beta}(t)
		= \frac{b^2|W_{\beta\alpha}|^2}{\Delta E^2}
		\sin^2\!\Bigl(\frac{\Delta E\,t}{2}\Bigr)
		+ \mathcal{O}(b^4),
		\label{eq:check3_tdpt}
	\end{align}
	valid for non-resonant pairs with $\Delta E = E_\beta - E_\alpha \neq 0$ and in the regime $b|W_{\beta\alpha}|t \ll 1$. Panel~(b) of Fig.~\ref{fig:dimer_benchmarks} compares the first-order prediction~\eqref{eq:check3_tdpt} with the exact transition probability obtained from the full time-evolved state $e^{-iH_b t}|\alpha\rangle$, restricted to the validity window $t \leq t_{\rm safe}$ defined by $b|W_{\beta\alpha}|t_{\rm safe} \ll 1$. Here $|\alpha\rangle$ is the ground state of $H_0$ in the $S_{\rm tot}^z = 0$ sector. To specify a representative leakage channel, we diagonalize $H_0$, evaluate the perturbative matrix elements $|W_{n\alpha}|$ from $|\alpha\rangle$ to all other eigenstates $|n\rangle$, order the connected channels by decreasing coupling strength, and choose the first state $|\beta\rangle$ satisfying the non-resonance filter $|E_\beta - E_\alpha| > 0.3$. This numerical cutoff is not part of the analytical claim itself; it is used only to exclude near-resonant channels for which a one-channel first-order comparison would be less transparent. This is important at $J_z^f=1$, where the isotropic dimer spectrum contains exact resonances of the type discussed in Sec.~\ref*{sec:isotropic_resonance}. With this choice, the perturbative formula reproduces the exact oscillation with a small relative error throughout the plotted validity window, confirming that higher-order corrections in $b$ remain negligible for the selected channel at the chosen parameters.

	\paragraph*{Checks~4--5: exact dimer-magnetization conservation and MB$\to$MB leakage.}
	Corollary~\ref*{cor:mbtomb_M12_b0} states that at $b = 0$ exactly, $\langle M_{12}^z(t)\rangle_{b=0} = \langle M_{12}^z(0)\rangle$ for all $t \geq 0$, since $[h_{12}, M_{12}^z] = 0$.
	For $0 < b \ll 1$, Eq.~(\ref*{eq:mbtomb_M12_leakage_bound}) gives the bound in Eq.~\eqref{eq:check2_leakage} for the MB$\to$MB quench, now with the interacting initial state $|\Psi_i(J_z^i)\rangle$ prepared in the $S_{\rm tot}^z = 0$ sector. Panel~(c) of Fig.~\ref{fig:dimer_benchmarks} shows both results simultaneously: the $b = 0$ curve is flat at numerical zero up to machine precision, while the $0 < b \ll 1$ curve exhibits a small but nonzero leakage that remains strictly below the bound $bt$ throughout the plotted interval. This joint presentation verifies the exact corollary at $b = 0$ and the controlled finite-$b$ bound in a single panel, and confirms that Eq.~(\ref*{eq:mbtomb_M12_leakage_bound}) applies to the MB$\to$MB quench as well.
	
	\paragraph*{Check~6: initial-state comparison bound.}
	Proposition~\ref*{prop:mbtomb_initial_state_comparison} bounds the difference between two edge-magnetization signals that are evolved under the same final Hamiltonian $H_b$ but start from distinct initial states. In the present benchmark we instantiate the proposition with $\rho_i=\rho_i(J_z^i=1)$ and $\sigma_i=\rho_i(J_z^i=2)$, namely the topological ground states of the initial Hamiltonian at $J_z^i=1$ and $J_z^i=2$, both prepared with the same pinning procedure described above. Both states are evolved under the same final Hamiltonian $H_b$ with $\delta_f=0.95$, $J_z^f=1$, and $b=J_2^f=2(1-\delta_f)=0.10$, with $N=22$. We compute their first-dimer reduced density matrices $\rho_{12}^{(i)}(1)$ and $\rho_{12}^{(i)}(2)$ by exact diagonalization and evaluate the trace distance $\frac12\|\rho_{12}^{(i)}(1)-\rho_{12}^{(i)}(2)\|_1$. Figure~\ref{fig:initial_state_comparison_bound} then compares the time-dependent difference $|\langle S_1^z(t)\rangle_{1}-\langle S_1^z(t)\rangle_{2}|$ with the analytical bound $\frac12\|\rho_{12}^{(i)}(1)-\rho_{12}^{(i)}(2)\|_1+|b|(2+|J_z^f|)t$. The numerical difference remains strictly below the bound throughout the plotted window, confirming Proposition~\ref*{prop:mbtomb_initial_state_comparison}.

	\begin{figure*}[t]
		\centering
		\includegraphics[width=\linewidth]{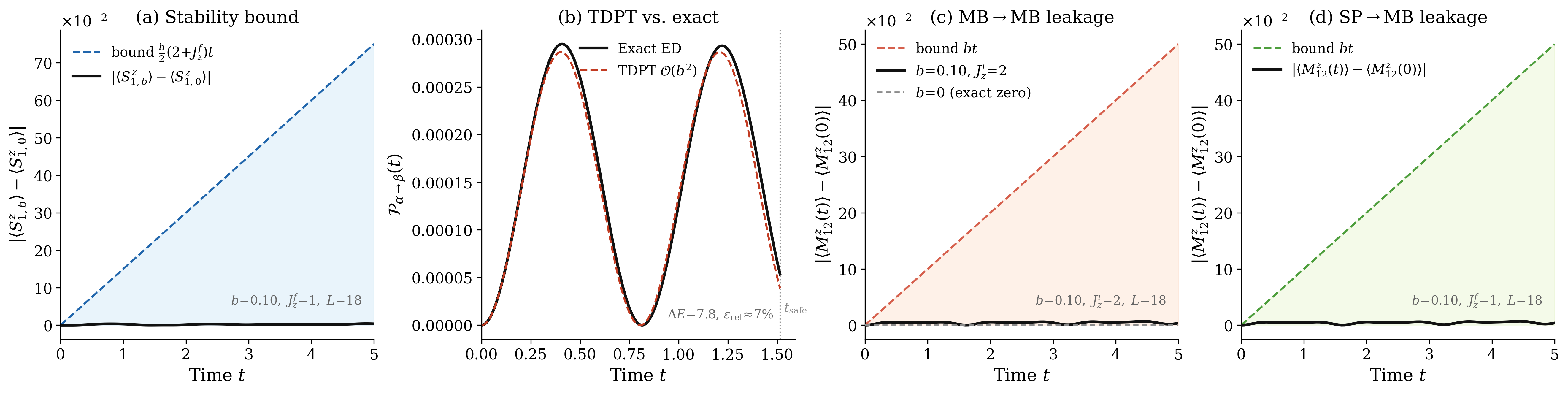}
		\caption{Numerical verification of five of the six analytical statements entering the controlled perturbative analysis near the decoupled-dimer limit, for $\delta_i = -0.95$, $\delta_f = +0.95$, $J_z^f = 1$, and $b = J_2^f = 0.10$, obtained by exact diagonalization in the $S_{\rm tot}^z = 0$ sector on a chain of length $N=18$. The value $J_z^f=1$ is the isotropic point discussed in Sec.~\ref*{sec:isotropic_resonance}. Panel~(a): finite-time stability bound~\eqref{eq:check1_bound} for the SP$\to$MB quench---the exact difference $|\langle S_{1,b}^z\rangle - \langle S_{1,0}^z\rangle|$ (solid) lies strictly below the analytical bound $b(2 + J_z^f)t$ (dashed, shaded) throughout the plotted interval. Panel~(b): first-order perturbative transition probability Eq.~\eqref{eq:check3_tdpt} (dashed) versus exact ED (solid) over the validity window $t \leq t_{\rm safe}$ for a representative nonresonant connected channel selected by the numerical protocol; the agreement remains quantitatively accurate throughout the plotted window. Panel~(c): dimer-magnetization leakage for the MB$\to$MB quench---the $b = 0$ curve is flat at machine precision, verifying exact conservation from Corollary~\ref*{cor:mbtomb_M12_b0}, while the finite-$b$ curve (solid) lies strictly below the bound $bt$ (dashed, shaded), confirming Eq.~(\ref*{eq:mbtomb_M12_leakage_bound}). Panel~(d): leakage bound~\eqref{eq:check2_leakage} for the SP$\to$MB quench---the exact leakage (solid) remains strictly below the bound $bt$ (dashed, shaded) throughout the plotted interval.}
		\label{fig:dimer_benchmarks}
	\end{figure*}
	
	\begin{figure}[htbp]
		\centering
		\includegraphics[width=\linewidth]{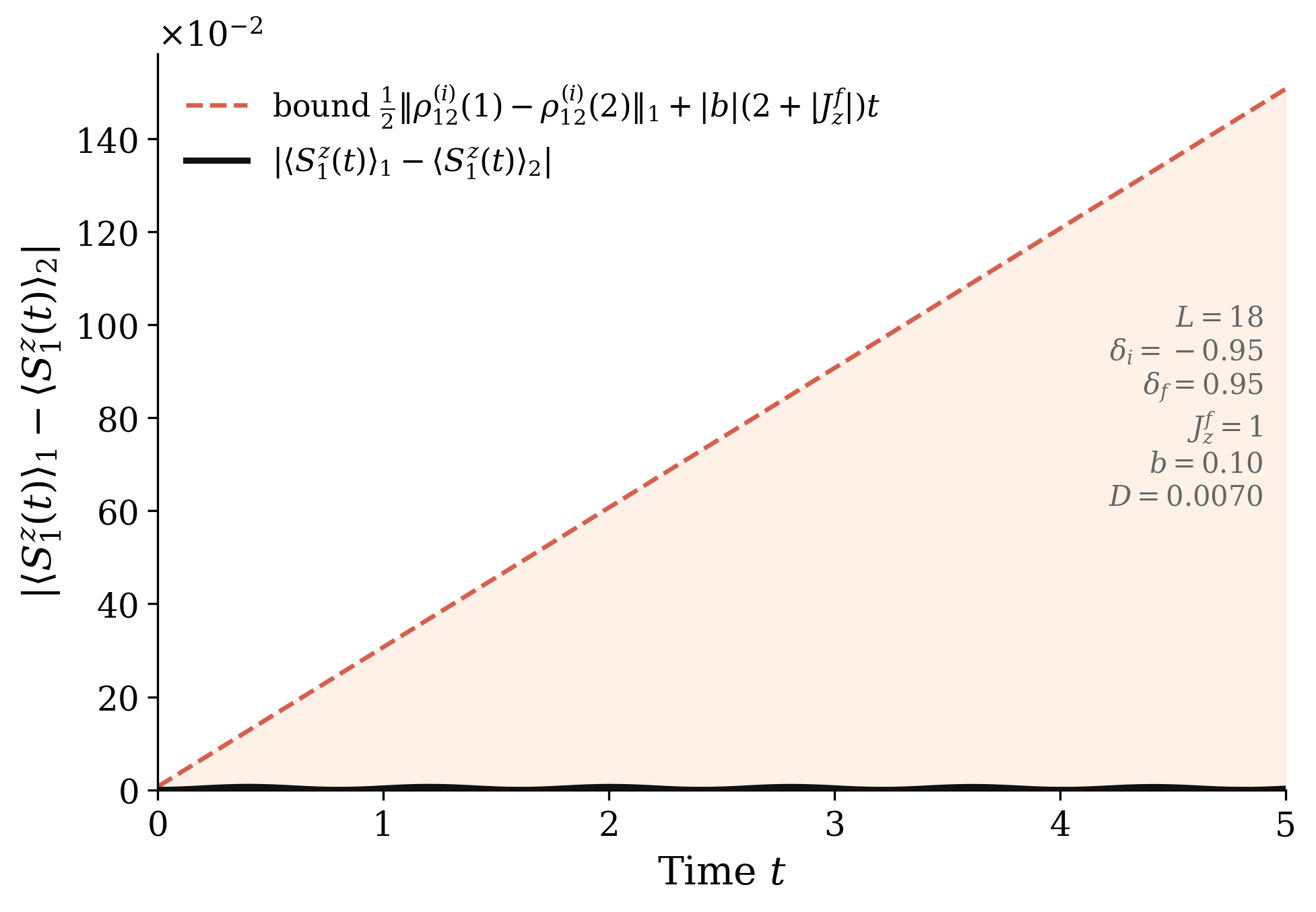}
		\caption{Verification of Proposition~\ref*{prop:mbtomb_initial_state_comparison}. The solid curve shows the absolute difference of the left-edge magnetization for two quenches that start from the topological initial ground states with $J_z^i=1$ and $J_z^i=2$, both evolved under the same final Hamiltonian $H_b$ with $N=22$. The parameters are $\delta_i=-0.95$, $\delta_f=0.95$, $J_z^f=1$, and $b=J_2^f=2(1-\delta_f)=0.10$, so the final Hamiltonian is at the isotropic point discussed in Sec.~\ref*{sec:isotropic_resonance}. The dashed curve shows the analytical bound $D+|b|(2+|J_z^f|)t$, where $D=\frac12\|\rho_{12}^{(i)}(1)-\rho_{12}^{(i)}(2)\|_1$. The numerical difference lies everywhere below the bound.}
		\label{fig:initial_state_comparison_bound}
	\end{figure}

\end{document}